\documentclass[12pt]{iopart}

\usepackage{graphicx,bbm,amsthm,amsfonts,amssymb,amsopn}
\usepackage{dsfont}
\usepackage{iopams}
\usepackage{hyperref}

\def\be{\begin{equation}}
\def\ee{\end{equation}}
\def\bra#1{\mathinner{\langle{#1}|}}
\def\ket#1{\mathinner{|{#1}\rangle}}
\def\expect#1{\langle#1\rangle}
\def\ul#1{{\underline{#1}}}
\newcommand{\un}[1]{{{#1}'}}
\def\ol#1{\overline{#1}}

\def\vmbb#1{\mathbb{#1}}

\def\aa{{\rm a}}
\def\bbb{{\rm b}}

\newtheorem{lem}{Lemma}
\newtheorem{thm}{Theorem}

\newcommand{\bbra}[1]{\langle\!\langle{#1}|}
\newcommand{\kket}[1]{|{#1}\rangle\!\rangle}
\newcommand{\braket}[2]{\langle #1 \vert #2 \rangle}

\newcommand{\ave}[1]{{\langle #1\rangle}}
\newcommand{\rvac}{\ket{\rm vac}}
\newcommand{\rrvac}{\kket{\rm vac}}
\newcommand{\lvac}{\bra{\rm vac}}
\newcommand{\llvac}{\bbra{\rm vac}}
\newcommand{\ua}{\uparrow}
\newcommand{\da}{\downarrow}
\newcommand{\ii}{ {\rm i} }
\newcommand{\dd}{ {\rm d} }

\newcommand{\ZZ}{\mathbb{Z}}
\newcommand{\RaR}{\mathbb{R}}
\newcommand{\CC}{\mathbb{C}}
\newcommand{\y}{{\rm y}}
\newcommand{\x}{{\rm x}}
\newcommand{\z}{{\rm z}}

\newcommand{\PP}{{\hat{\cal P}}}
\newcommand{\LL}{{\hat{\cal L}}}
\newcommand{\UU}{{\hat{\cal U}}}

\newcommand{\DD}{{\hat{\cal D}}}
\newcommand{\half}{{\textstyle\frac{1}{2}}}
\newcommand{\ihalf}{{\textstyle\frac{\ii}{2}}}
\newcommand{\mm}[1]{{\mathbf{#1}}}
\newcommand{\bb}[1]{{\mathbf{#1}}}
\newcommand{\quart}{{\textstyle\frac{1}{4}}}
\def\tr{{{\rm tr}}}
\def\ad{{\,{\rm ad}\,}}
\def\End{{\,{\rm End}\,}}
\def\one{\mathbbm{1}}
\def\Re{{\,{\rm Re}\,}}

\def\Re{{\,{\rm Re}\,}}

\eqnobysec

\begin{document}

\topical{Matrix product solutions of boundary driven quantum chains}
\author{Toma\v{z} Prosen}

\address{Department of physics, Faculty of Mathematics and Physics, University of Ljubljana,
Jadranska 19, SI-1000 Ljubljana, Slovenia}

\date{\today}

\begin{abstract}
We review recent progress on constructing non-equilibrium steady state density operators of boundary driven locally interacting quantum chains, where driving is implemented via Markovian dissipation channels attached to the chain's ends. We discuss explicit solutions in three  different classes of quantum chains, specifically, the paradigmatic (anisotropic)  Heisenberg spin-$1/2$ chain, the Fermi-Hubbard chain, and the Lai-Sutherland spin-1 chain, and discuss universal concepts which characterize these solutions, such as matrix product ansatz and a more structured walking graph state ansatz. The central theme is the connection between the matrix product form of nonequilibrium states and
the integrability structures of the bulk Hamiltonian, such as the Lax operators and the Yang-Baxter equation. 

However, there is a remarkable distinction with respect to the conventional quantum inverse scattering method, namely addressing nonequilibrium steady state density operators
requires non-unitary irreducible representations of Yang-Baxter algebra which are typically of infinite dimensionality. 
Such constructions result in non-Hermitian, and often also non-diagonalisable families of commuting transfer operators which in turn result in novel conservation laws of the integrable bulk Hamiltonians.
For example, in the case of anisotropic Heisenberg model, quasi-local conserved operators which are odd under spin reversal (or spin flip) can be constructed, 
whereas the conserved operators stemming from orthodox Hermitian
transfer operators (via logarithmic differentiation) are all even under spin reversal.
\end{abstract}

\pacs{02.30.Ik, 02.50.Ga, 05.60.Gg, 75.10.Pq}

\maketitle

\pagestyle{empty}
\tableofcontents
\pagestyle{headings}

\section{Introduction}

\subsection{The key concepts of quantum integrability}

Exact solutions of nontrivial yet simple physical models are of paramount importance in statistical and quantum physics. On one hand, an exactly solved model often characterises universal behaviour of a more general class of possibly unsolvable models
and thus represents the best possible {\em exact} understanding of physical reality. On the other hand, the tricks developed in the course of deriving such solutions have often lead to general development in mathematical methodology or even to novel mathematical concepts. Famous examples being for instance, Hans Bethe's solution of the Heisenberg model of magnetism in 1931 and Lars Onsager's solution of planar (2D) Ising model in 1944. Remarkably, these two threads have merged in the works of C.~N.~Yang \cite{Y67} and Rodney Baxter \cite{B82}, giving birth to the celebrated Yang-Baxter equation (YBE), also known as the star-triangle equation, representing the most general characterisation of integrability known to date. Moreover, abstract algebraic characterisation of YBE lead Vladimir Drinfeld in 1980's to introduce the concept of quantum groups and develop their representation theory, together with Jimbo, Reshetikhin, Sklyanin and many others.

An important ingredient of the theory of quantum integrability is the concept of {\em auxiliary} Hilbert space which interacts with the {\em physical} quantum Hilbert space via the so-called Lax (or scattering) operator, solving the YBE together with the so-called 
$R$-matrix which represents the `internal' scattering between a pair of auxiliary spaces. Considering explicit matrix representations pertaining to {\em finite} dimensional auxiliary space --- typically being a {\em fundamental} representation of the corresponding quantum- or Lie-group symmetry ---
resulted in the very successful quantum inverse scattering method \cite{Skly,KBI93,F94} for diagonalising integrable many-body Hamiltonians and computing their {\em equilibrium} correlation functions, also known as the algebraic Bethe ansatz (ABA).
Remarkably, literally the same technique has beed adapted to solve certain nonequilibrium classical driven diffusive many-body systems in one-dimension, namely the so-called simple exclusion processes (SEP), 
for a simple reason that their Markov chain generator can be identified with the Heisenberg-like Hamiltonian.  However, ABA calculations often result in an implicit solution written in terms of a coupled set of algebraic (or transcendental) equations [the Bethe equations (BE)], or in the thermodynamic limit (TL), in terms of coupled integral or functional equations (so-called Bethe-Takahashi equations).

In a seminal paper \cite{DEVP}, Derrida and coworkers have been able to circumvent this problem by writing an explicit solution for the steady state of symmetric and asymmetric simple exclusion processes (ASEP/SSEP) in terms of a matrix product ansatz (MPA).
MPA in turn also allowed for explicit calculation of all physical observables and correlation functions in the nonequilibrium steady state (NESS), see e.g. \cite{derrida,S01,BE07} for review. It is notable however, that a particular MPA appeared earlier as
a ground state of a valence bond solid in one-dimension (AKLT model \cite{AKLT}), or as a convenient family of the so-called finitely correlated states of quantum spin chains \cite{FNW}, and nowadays represents a cornerstone of the 
density-matrix-renormalization-group (DMRG) method, a state-of-the art technique for simulation of strongly correlated systems in one-dimension \cite{Scholl}. It has later also been recognised that ABA method for equilibrium integrable 
Heisenberg spin chains can be re-phrased in terms of MPA for all eigenstates, from which BE can be equivalently derived \cite{AL04,KM10,MKV}.

However, the problem of integrability in the combined paradigm of driven diffusive systems and quantum many-body interactions resisted until 2011 when NESS density operator of the boundary driven 
Lindblad master equation for the anisotropic Heisenberg spin 1/2 chain ($XXZ$ model) has been solved in terms of MPA \cite{P11a,P11b} by the author of the present topical review.\footnote{The term `driven diffusive systems' can in fact be misleading in our context since, as we shall see later, the competition between dissipation and coherent quantum many-body interactions can lead to a plethora of transport
behaviours, ranging from ballistic, anomalous, diffusive, to insulating.}
The solution appeared disconnected from conventional theory of quantum integrability at the first sight, but has later \cite{PIP,IZ14,I14} been related to infinite-dimensional solutions of YBE pertaining to non-unitary (or in mathematical terminology, {\em non-integrable}) irreducible representations of quantum group symmetry.

\subsection{The purpose and summary of the review}

We have by now managed to derive a number of exact MPA solutions of NESS of boundary driven quantum chains for different types of integrable  locally interacting  bulk Hamiltonians and different boundary dissipators (diffusive driving). 
The unifying picture of such {\em non-equilbrium quantum integrability} has been emerging slowly from studying various quite specific situations scattered over several rather technical papers, so it is perhaps now a good moment to wrap these results within a single topical review article.
The purpose of the present paper is thus to present a coherent and up-to-date review of the progress on solving the challenging problem of integrable boundary driven diffusive (or better to say, disipatively driven) quantum systems. The problem may be of quite general interest for mathematical physics and
statistical physics community as it represents in some sense a minimal description of dissipative or incoherent driving of a quantum many-body system without affecting coherent macroscopic character of its bulk. One can view this approach as an {\em incoherent} or {\em dissipative} analogue to a standard trick in nonequilibrium transport of {\em coherently} driven systems, where  one introduces a bias in electro-chemical potential replacing a real electric field in the bulk.
In contrast to most of literature on integrable systems, this review takes what might be considered a {\em bottom-up} approach. Namely, we first demonstrate how to explicitly construct various exact physical solutions and only later think or elaborate on their abstract mathematical properties and meaning.

In the following subsection \ref{bdqme} we physically motivate the paradigm of boundary driven quantum master equation of the Lindblad form.
In section \ref{heis} we then describe the MPA solution of NESS in case of the simplest generic integrable model, specifically the $XXX$ or more generally, the $XXZ$ spin $1/2$ chain.
Having the explicit MPA form, we also discuss in detail explicit analytic computation of observables and correlation functions in NESS, and its relation to the quantum group $U_q(\mathfrak{sl}_2)$ and the corresponding solutions of YBE. In sections \ref{hubb} and \ref{lai} we then discuss MPA solutions in two other families of boundary driven systems, specifically in the one-dimensional fermionic Hubbard model and in Lai-Sutherland spin-$1$ chain, respectively.
The former is intriguing as the corresponding Lax operator generating the MPA seems not to be connected to obvious symmetries of the model neither to the celebrated Shastry's $R$-matrix, while the latter is fundamentally interesting since the NESS is macroscopically degenerate and can be parametrised by a thermodynamic variable yielding a nonequilibrium analog of the grand-canonical ensemble.
In section \ref{qlcl} we discuss the relevance of our exact  far-from-equilibrium results for linear response physics of quantum transport,  such as the existence of novel quasilocal conserved operators and 
rigorous lower bounds on zero frequency Green-Kubo transport coefficients. In section \ref{disc} we outline a subset of most exciting and urgent open problems and conclude.

We stress that the present review contains also several original results w.r.t. the existing literature. In particular, a mixture of asymmetric coherent (boundary fields) and incoherent (boundary dissipation) has been worked out completely for $XXZ$ chain in subsection \ref{subsect:asym}, as well as a fully analytic calculation of the nonequilibrium partition functions in the isotropic $XXX$ chain for a variety of asymmetrized boundary drivings in subsection~\ref{subsect:neZ}.

\subsection{The paradigm of boundary driven quantum master equation}
\label{bdqme}

Let us consider a finite $d-$dimensional local {\em physical} Hilbert space ${\cal H}_{\rm p} \simeq \CC^d$, e.g. of a quantum spin $s=(d-1)/2$, 
and consider a quantum chain of length $n$ defined on a tensor product space ${\cal H}_{\rm p}^{\otimes n}=\bigotimes_{x=1}^{n}{\cal H}_{\rm p}$ with a Hamiltonian which can be written in terms of a sum of local interactions
$h_{x,x+1}$ acting nontrivially only over a pair of neighbouring sites $(x,x+1)$
\begin{equation}
H = \sum_{x=1}^{n-1} h_{x,x+1}.
\label{eq:ham}
\end{equation}

Within the theory of open quantum systems \cite{BP02} (see Refs.~\cite{AL,davies} for more mathematical accounts) incoherent markovian quantum dissipation can be completely described by a set of quantum jump operators 
$\{ L_\mu \in {\rm End}({\cal H}_{\rm p}^{\otimes n}); \mu=1,2\ldots \}$, also called Lindblad operators, so that the evolution of the
system's many-body density operator $\rho(t)$ satisfies the Lindblad-Gorini-Kossakowski-Sudarshan master equation \cite{L76,GKS76}
\begin{equation}
\frac{\dd}{\dd t}\rho(t) = \LL\rho(t):= -\ii [H,\rho(t)] + \sum_\mu \left(2L_\mu \rho(t) L_\mu^\dagger - \{L_\mu^\dagger L_\mu,\rho(t)\}\right).
\label{eq:lindeq}
\end{equation}
Eq.~(\ref{eq:lindeq}) defines a general family of markovian dynamical semigroups $\UU(t) = \exp(t\LL)$, $\UU(t)\UU(t') = \UU(t+t')$, $t,t'\in\RaR^+$, 
which preserve hermiticity, positivity and trace of the density matrix $\rho(t)$ at all times. In fact, the semigroup (\ref{eq:lindeq}) can be derived from the
microscopic unitary evolution of the {\em universe} $=$ {\em system} $\otimes$ {\em environment}, provided the following conditions are met: (i) the initial state of the universe is factorized $\rho(0)\otimes \rho_{\rm env}$,
(ii) the coupling between the system and the environment is weak so the second-order Born-Dyson series can be used, and (iii) the dynamical correlation functions of environment observables that are coupled to the system decay
on time scales for which the system's dynamics can be considered as frozen ({\em secular approximation}). 

\begin{figure}[!t]
\begin{center}
\vspace{-2mm}
\includegraphics[scale=0.8]{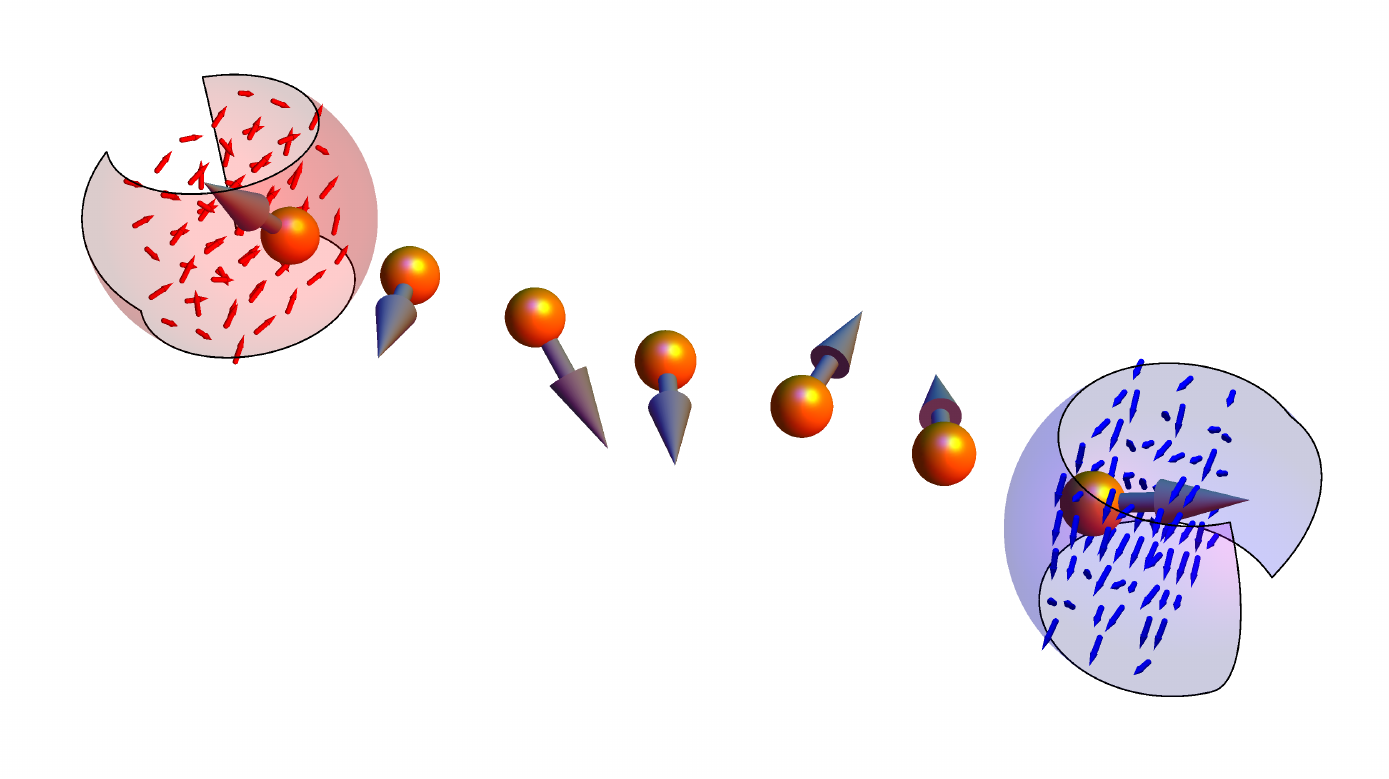}
\vspace{-6mm}
\caption{
Illustration of a dissipatively boundary driven quantum spin chain where the first and the last particle of the chain are immersed into a pair of distinct markovian quantum baths corresponding to different values of
thermodynamic potentials.
\label{fig:scheme}}
\end{center}
\end{figure}

We shall furthermore assume that the incoherent quantum processes are {\em ultra-local}, meaning that all jump operators are of the form 
\begin{equation}
L_\mu = \ell_\mu \otimes \one_{d^{n-1}}\quad {\rm or}\quad L_\mu = \one_{d^{n-1}} \otimes \ell_\mu, \qquad \ell_\mu \in {\rm End}({\cal H}_{\rm p}),
\label{eq:boundL}
\end{equation} 
i.e., they act nontrivially either on the left or the right boundary of the chain (Fig~.\ref{fig:scheme}).
Within the microscopic derivation \cite{BP02}, this additional assumption is generically justified only if the on-site part of $h_{x,x+1}$ is much larger (in operator norm) than the genuine interaction part \cite{W07}, 
meaning that a local disturbance on the boundary site
does not spread to the interior before it gets dissipated so the dissipator can be assumed to act locally.
However, there is an alternative phenomenological description of the boundary driven quantum master equation (\ref{eq:lindeq},\ref{eq:boundL}) in terms of the so-called {\em repeated interactions} protocol \cite{KP09,clark} which is free from any assumptions about the microscopic dynamics. In this protocol it is assumed that end spins/particles of the chain are put in contact (interaction) for a short amount of time $\delta t$ with a pair of auxiliary spins/particles which are assumed to be prepared in two different thermal, canonical, grandcanonical, or any other equilibrium states. In the next time interval $\delta t$, the auxiliary spins, the states of which may have already slightly changed, are replaced by a fresh, independent pair of thermal auxiliary spins, and so on. One may imagine two running belts carrying thermally prepared auxiliary spins in separable states and moving along each side of the chain at some speed while interacting with the ends of the central chain.
Writing dynamical evolution of the central chain density matrix in the high-speed limit $\delta t\to 0$ one precisely recovers (\ref{eq:lindeq}). It should be noted that closely related protocols could nowadays be implemented using contemporary cold atom experiments.

We stress, however, that our main motivation for studying the problem  (\ref{eq:lindeq},\ref{eq:boundL}) is in its conceptual simplicity and mathematical elegance of formulation as the minimal model which can describe nonequilibrium driving of a coherent many-body system in one dimension. 
Such a setup allows for an efficient numerical (DMRG) simulation of quantum transport for generic non-integrable bulk interactions \cite{PZ09,clark2}.
Even though this review focuses on the case where the bulk Hamiltonian is strongly interacting, we start by mentioning that in the non-interacting case, where the Hamiltonian is usually mapped to $XY$ spin $1/2$ chain and the Lindblad jump operators are {\em linear} in the corresponding Wigner-Jordan fermi operators, one is able to calculate analytically the full NESS density operator, all its observables, as well as the relaxation dynamics  of the master equation (\ref{eq:lindeq}), in terms of nonequilibrium dissipative quasi-particle excitations \cite{P08,P10,DK11}.
 Furthermore, one is able to write exact solutions in the non-interacting case even in some situations where the jump operators are {\em quadratic} but Hermitian, such as for the so-called {\em dephasing} noise, which
 for exact solvability has to be homogeneously distributed in the bulk of the chain \cite{Z10}. This solution can in fact be written in terms of a simple MPA \cite{Z11} with $4\times 4$ auxiliary matrices and allows for some further solvable generalisations \cite{eisler}. 
Such models can be further generalized as hybrid quantum-classical markov chains \cite{temme}, where the incoherent part of dynamics exactly coincides with the classical SEP.

Focusing on the concept boundary of driven many-body systems in this review, where the baths are described only effectively, we have to refrain from discussing a bulk of related and also highly topical literature on nonequilibrium steady states with infinite (microscopically formulated) baths.

\section{Anisotropic Heisenberg spin-$1/2$ chain}
\label{heis}

We shall start by considering arguably the simplest integrable model with strong interactions --- the $XXZ$ model, namely a homogeneous spin $1/2$ chain ($d=2$) with nearest neighbour anisotropic Heisenberg interaction with anisotropy parameter $\Delta$,
\begin{equation}
h = 2\sigma^+ \otimes \sigma^- + 2\sigma^-\otimes \sigma^+ + \Delta \sigma^\z \otimes \sigma^\z,
\end{equation}
where $\sigma^\pm = \half(\sigma^\x \pm \ii \sigma^\y)$ and $\sigma^\x,\sigma^\y,\sigma^\z$ are the usual $2\times 2$ Pauli matrices, so that the Hamiltonian (\ref{eq:ham}) can be written as
\begin{equation}
H = \sum_{x=1}^{n-1} \one_{2^{x-1}} \otimes h \otimes \one_{2^{n-x-1}}.
\label{eq:HXXZ}
\end{equation}
Embedding the Pauli operators into ${\rm End}({\cal H}_{\rm p}^{\otimes n})$, $\sigma^\alpha_x := \one_{2^{x-1}}\otimes \sigma^\alpha \otimes \one_{2^{n-x}}$, translationally shifted interactions read $h_{x,x+1} = 2 \sigma^+_x \sigma^-_{x+1}+2 \sigma^-_x \sigma^+_{x+1} + \Delta\sigma^\z_x \sigma^\z_{x+1}$.

We start by considering the simplest nontrivial dissipative driving with only two jump processes coupled to bulk unitary dynamics at a single dissipation rate $\varepsilon\in\RaR^+$,
\begin{equation}
L_1 = \sqrt{\varepsilon}\sigma^+_1,\quad L_2 = \sqrt{\varepsilon}\sigma^-_n.
\label{eq:lindpure}
\end{equation}
As the total z-component of magnetization $M=\sum_{x=1}^n \sigma^\z_x$ and hence the number of up-spins are conserved in the bulk
\begin{equation}
[H,M]=0,
\end{equation}
the incoherent processes (\ref{eq:lindpure}) can be interpreted as a pure {\em source} of up-spins on the left end and a pure {\em sink} of up-spins on the right end.
NESS density operator $\rho_\infty = \lim_{t\to\infty} \exp(t \LL) \rho(0)$ can be considered as a fixed point of the propagator, or null eigenvector of the Liouvillian
\begin{equation}
\LL \rho_\infty = 0,\qquad{\rm where}\quad \LL = -\ii\,{\rm ad}\,H + \varepsilon \DD_{\sigma^+_1} +  \varepsilon \DD_{\sigma^-_n},
\label{eq:LXXZ}
\end{equation}
where the Lie derivative map $({\rm ad}\,H)\rho:=[H,\rho]$ and an elementary dissipator map 
\begin{equation}
\DD_L(\rho) = 2 L \rho L^\dagger - \{L^\dagger L,\rho\}
\label{eq:candis}
\end{equation} 
have been introduced.

\subsection{Uniqueness of NESS}
\label{unique}

Let us first show \cite{P12a} that under quite general conditions, the Liouvillian (\ref{eq:LXXZ}) possesses a {\em unique} NESS, i.e. the fixed point $\rho_\infty$ is independent of the initial state $\rho(0)$. 

We start by noting a theorem due Evans \cite{evans} and Frigeiro \cite{frigeiro} (with the subject nicely reviewed by Spohn~\cite{spohn}) which essentially states that NESS is unique iff the set of operators $\{H,L_1,L_1^\dagger,L_2,L_2^\dagger \ldots\}$ generates, under multiplication and addition, the entire algebra 
${\rm End}({\cal H}^{\otimes n}_{\rm p})$ of operators over a quantum chain on $n$ sites. Indeed, this is easy to demonstrate explicitely even considering only a triple of operators $\{H, \sigma^+_1, \sigma^-_1\}$ 
while uniqueness then trivially extends to all cases including, and generalizing, (\ref{eq:lindpure}) where the set $\{ L_\mu,L^\dagger_\mu;\mu=1,2\ldots \}$ contains, up to scalar prefactors, either a pair  $\sigma^+_1,\sigma^-_1$, or a pair $\sigma^+_n,\sigma^-_n$ due left-right inversion symmetry.

Namely, one observes the following recursive operator identities:
\begin{eqnarray}
\sigma^+_2  &=& \frac{1}{4}\sigma^{\rm z}_1 [\sigma^+_1,[H,\sigma^{\rm z}_1]], \label{eq:opid1}\\
\sigma^+_{x} &=& -\sigma^+_{x-2}  - \frac{1}{2}\sigma^{\rm z}_{x-1} [\sigma^-_{x-1},\sigma^+_{x-1} H \sigma^+_{x-1}], \quad x=3,4\ldots, n,\label{eq:opid2}
\end{eqnarray}
which generate the entire set $\{\sigma^+_x; x=1,\ldots,n\}$ starting from just $H$ and $\sigma^+_1$. Similarly, 
$\{\sigma^-_x; x=1,\ldots,n\}$ are generated by Hermitian adjoints of Eqs.~(\ref{eq:opid1},\ref{eq:opid2}) starting from $H$ and $\sigma^-_1$.
The set $\{\sigma^+_x,\sigma^-_x; x=1,\ldots,n\}$ then generates all elements of ${\rm End}({\cal H}^{\otimes n}_{\rm p})$ by multiplication and addition.

For related recent general results on (non)uniqueness of fixed points and characterization of the space of steady states of Lindbladian dynamics, see, e.g., Refs.~\cite{baumgartner,kraus}.

\subsection{Matrix product solution -- isolating defect operator method}

We shall now present an {\em ad hoc} method which generates the MPA of NESS fixed point $\rho_\infty$ for the $XXZ$ model following Ref.~\cite{P11b} (also \cite{P11a}), which we term an
{\em isolated defect operator} (IDO) approach. Later we shall in subsect.~\ref{subsect:XXX2} rederive this solution in a more elegant way using a non-Hermitian infinite-dimensional Lax matrix of the $XXZ$ model, but we believe it may be of
interest to investigate both approaches. For instance, if for some model the appropriate Lax matrix for the problem is not known, IDO strategy may sometimes be used to devise a systematic method for determining MPA by an automated procedure (see e.g. Ref.~\cite{P14a}).

Let us first show that NESS density operator allows a particular factorization
in terms of a non-Hermitian {\em amplitude operator} $\Omega_n(\varepsilon)$.
\begin{lem}
Let  $\Omega_n \in {\rm End}({\cal H}_{\rm p}^{\otimes n})$ satisfy the following conditions (defining relations):\\ 
(i) a recursion identity for the bulk, setting $\Omega_1 := \sigma^0 := \one_2$,
\begin{eqnarray}
[H,\Omega_n] = -\ii \varepsilon( \sigma^\z \otimes \Omega_{n-1} - \Omega_{n-1}\otimes\sigma^\z),
\label{eq:def}
\end{eqnarray}
and (ii) the recursion identity for the boundaries
\begin{equation}
\Omega_n = \sigma^0 \otimes \Omega_{n-1} + \sigma^+ \otimes \Omega^+_{n-1} = \Omega_{n-1} \otimes \sigma^0 + \Omega^-_{n-1} \otimes \sigma^-
\label{eq:bcOm}
\end{equation}
for some unspecified operators $\Omega^\pm_{n-1} \in {\rm End}({\cal H}_{\rm p}^{\otimes (n-1)})$.
Then, the density operator
\begin{equation}
\rho_\infty = \frac{R_\infty}{\tr R_\infty},\quad R_\infty= \Omega_n \Omega^\dagger_n,
\label{eq:rhoinf}
\end{equation}
satisfies the fixed point condition (\ref{eq:LXXZ}).
\end{lem}
\begin{proof}
We need to prove that $\LL (\Omega_n \Omega^\dagger_n) = 0$, i.e. 
\begin{equation}
\ii\varepsilon^{-1} [H,\Omega_n \Omega^\dagger_n]=\DD_{\sigma^+_1}(\Omega_n \Omega^\dagger_n) + \DD_{\sigma^-_n}(\Omega_n \Omega^\dagger_n).
\label{eq:stlind}
\end{equation}
Using the Leibniz rule for Lie derivative, and (i), LHS of (\ref{eq:stlind}) can be transformed to
\begin{eqnarray}
& \Omega_n \left(\ii\varepsilon^{-1}[H,\Omega_n]\right)^\dagger + \ii\varepsilon^{-1}[H,\Omega_n]\Omega_n^\dagger = \nonumber \\ 
&\!\!\!\!\!\!\!\!\!\!\!\!\!\!\!\!\!\!\!\!
\Omega_n (\sigma^{\rm z}\otimes \Omega_{n-1}^\dagger) - \Omega_n (\Omega_{n-1}^\dagger\otimes \sigma^{\rm z})  +  (\sigma^{\rm z}\otimes  \Omega_{n-1}) \Omega_n^\dagger - (\Omega_{n-1}\otimes \sigma^{\rm z}) \Omega_n^\dagger.
\label{eq:b}
\end{eqnarray}
Applying, respectively, the first and the second identity of (ii) on the first two and the last two occurrences of $\Omega_n$ in (\ref{eq:b}), we obtain
\begin{eqnarray}
&2\sigma^{\rm z}\otimes \Omega_{n-1}\Omega_{n-1}^\dagger - \sigma^+ \otimes \Omega^+_{n-1}\Omega_{n-1}^\dagger - \sigma^- \otimes \Omega_{n-1} (\Omega^+_{n-1})^\dagger - \nonumber \\
& 2\Omega_{n-1}\Omega_{n-1}^\dagger \otimes \sigma^{\rm z} - \Omega^-_{n-1}\Omega_{n-1}^\dagger \otimes \sigma^- - \Omega_{n-1}(\Omega^-_{n-1})^\dagger \otimes \sigma^+.
\label{eq:6terms}
\end{eqnarray}
The first term on the RHS of (\ref{eq:stlind}) can again be transformed by applying the first identity of (ii) to 
\begin{eqnarray}
&\DD_{\sigma^+}(\sigma^0) \otimes \Omega_{n-1} \Omega^\dagger_{n-1} + \DD_{\sigma^+}(\sigma^-)\otimes \Omega_{n-1} (\Omega^+_{n-1})^\dagger +  \nonumber \\
&\DD_{\sigma^+}(\sigma^+)\otimes \Omega^+_{n-1}\Omega^\dagger_{n-1} + \DD_{\sigma^+}(\sigma^+\sigma^-)\otimes\Omega^+_{n-1}(\Omega^+_{n-1})^\dagger,
\end{eqnarray}
which, by observing the complete local action of the dissipators
\begin{eqnarray}
&& \DD_{\sigma^+}(\sigma^0) =  2\sigma^\z,\quad \DD_{\sigma^+}(\sigma^\pm) = -\sigma^\pm,\quad  \DD_{\sigma^+}(\sigma^+\sigma^-)=0, \\ 
&& \DD_{\sigma^-}(\sigma^0) =  -2\sigma^\z,\quad \DD_{\sigma^-}(\sigma^\pm) = -\sigma^\pm,\quad  \DD_{\sigma^-}(\sigma^-\sigma^+)=0, \label{eq:endproof}
\end{eqnarray}
result in exactly the first three terms of expression (\ref{eq:6terms}). Analogously, the second term on the RHS of (\ref{eq:stlind}) results in the last three terms of (\ref{eq:6terms}).
\end{proof}

Note that the condition (ii), Eq. (\ref{eq:bcOm}), implies that $\Omega_n$ is {\em unit diagonal} and {\em upper triangular} matrix in the joint eigenbasis  $\{\ket{\nu_1,\ldots,\nu_n}; \nu_x\in\{0,1\}\}$ of $\{\sigma^\z_x\}$, $\sigma^\z_x \ket{\ul{\nu}} = (1-2\nu_x) \ket{\ul{\nu}}$, where the ordering is determined by the binary ${\rm ord}(\ul{\nu}) := \sum_{x=1}^n \nu_x 2^{x-1}$,
\begin{equation}
\bra{\ul{\nu}}\Omega_n\ket{\ul{\nu}}=1,\qquad
\bra{\ul{\nu}}\Omega_n\ket{\ul{\nu}'}=0 \quad {\rm if}\quad {\rm ord}(\ul{\nu}) > {\rm ord}(\ul{\nu}'),
\label{eq:ut}
\end{equation}
which can be shown by a straightforward induction. The factorisation of NESS (\ref{eq:rhoinf}) can thus be considered as a (reverse) {\em many-body Cholesky decomposition}. In the canonical Cholesky decomposition, though, 
the matrix of $\Omega_n$ would have to be lower-triangular, but this can be trivially mended by considering a spin-reversed problem. Namely, noting the spin-reversal symmetry of the Hamiltonian
\begin{equation}
P H P^{-1} = H,\quad P =  P^{-1} =\prod_{x=1}^n \sigma^{\x}_x = \pmatrix{0 & 1\cr 1 & 0}^{\otimes n},
\label{eq:parity}
\end{equation}
the spin-reversed density operator
\begin{equation}
R'_\infty = P R_\infty P^{-1} = \Omega'_n (\Omega'_n)^\dagger,\quad \Omega'_n=  P \Omega_n P^{-1},
\label{eq:Omp}
\end{equation}
where the matrix of $\Omega'_n$ is {\em unit diagonal} and {\em lower triangular},
solves the reverse NESS fixed point condition with the source and sink being swapped:
\begin{equation}
L_1' = P L_1 P^{-1} = \sqrt{\varepsilon} \sigma^-_1,\qquad
L_2' = P L_2 P^{-1} = \sqrt{\varepsilon} \sigma^+_n.
\label{eq:revdriv}
\end{equation}

We shall now define an abstract auxiliary Hilbert space ${\cal H}_{\rm a}$, with a particular state $\ket{0}\in{\cal H}_{\rm a}$, and postulate an MPA for the amplitude operator
\footnote{An impatient reader should here be directed straight to subsect.~\ref{subsect:lax}.}
\begin{equation}
\Omega_n = \sum_{\alpha_1,\ldots,\alpha_n\in\{ +,-,0\}} \bra{0}\mm{A}_{\alpha_1}\mm{A}_{\alpha_2}\cdots \mm{A}_{\alpha_n}\ket{0} \sigma^{\alpha_1}\otimes \sigma^{\alpha_2} \otimes \cdots \sigma^{\alpha_n},
\label{eq:MPA1}
\end{equation}
with a yet-to-be specified tripple of matrices $\mm{A}_\pm,\mm{A}_0 \in{\rm End}({\cal H}_{\rm a})$.
Throughout this paper we write in {\em roman-bold} letters operators which act non-trivially, i.e. not as scalars, over the auxiliary space ${\cal H}_{\rm a}$.
We should note that the terms with $\sigma^\z_x$ have been omitted in the ansatz (\ref{eq:MPA1}) which is the key to the solution of the problem.
Let us call $\sigma^\z$ a {\em defect operator} here.
Looking at the pair of sites without a defect, one finds that commutation with Hamliltonian density produces exactly one defect, either on the left or on the right tensor factor
\begin{equation}
[h,\sigma^{\alpha}\otimes\sigma^{\alpha'}] = \sum_{\beta\in\{+,-,0\}}\left(\gamma^{\alpha,\alpha'}_\beta \sigma^\beta \otimes \sigma^\z + \gamma^{\alpha',\alpha}_\beta \sigma^\z\otimes \sigma^\beta\right),
\label{eq:struct}
\end{equation}
$\alpha,\alpha'\in\{+,-,0\}$, 
with the structure constants $\gamma^{\alpha,\alpha'}_\beta$, writing out only the non-vanishing elements:
\begin{equation}
\gamma^{\pm,0}_\pm = \pm 2\Delta,\quad \gamma^{0,\pm}_\pm = \mp 2, \quad \gamma^{\pm,\mp}_0 = \pm 1.
\end{equation}
The commutator $[H,\Omega_n]$ and the entire defining relation (\ref{eq:def}) should then contain Pauli operators with exactly one defect $\sigma^\z_x$. Considering all the terms 
where the defect operator appears in the bulk $1 < x < n$, a sufficient condition for cancelation involves a projection on a triple of sites $(x-1,x,x+1)$
\begin{equation}
\!\!\!\!\!\!\!\!\!\!\!\!\!\!\!\!\!\!\!\!\!\!\!\!\!\!\!\!\!\!\!\!\!\!\!\!\sum_{\alpha,\alpha',\alpha''\in\{+,-,0\}}\!\!\!\!\!\!\!\!\!\tr\left( \sigma^{-\beta}\otimes \sigma^{\z} \otimes\sigma^{-\beta'} [h\otimes \one_2 + \one_2 \otimes h,\sigma^{\alpha}\otimes \sigma^{\alpha'} \otimes\sigma^{\alpha''}]\right)
\mm{A}_{\alpha}\mm{A}_{\alpha'}\mm{A}_{\alpha''}=0,
\label{eq:project}
\end{equation}
or, equivalently, in terms of the structure constants $\gamma^{\alpha,\alpha'}_\beta$
\begin{equation}
\sum_{\alpha,\alpha'\in\{+,-,0\}} \left(\gamma^{\alpha,\alpha'}_{\beta} \mm{A}_{\alpha}\mm{A}_{\alpha'}\mm{A}_{\beta'} + \gamma^{\alpha',\alpha}_{\beta'} \mm{A}_\beta\mm{A}_\alpha\mm{A}_{\alpha'}\right) = 0,
\end{equation}
which represent $9$ homogeneous cubic identities for the matrices $\mm{A}_\pm,\mm{A}_0$ for $\beta,\beta'\in\{+,-,0\}$.
Writing these explicitly, we find that only $8$ of them are linearly independent:
\begin{eqnarray}
&& [\mm{A}_0,\mm{A}_\pm\mm{A}_\mp] = 0,\quad \{ \mm{A}_0,\mm{A}^2_\pm\} = 2\Delta \mm{A}_\pm\mm{A}_0\mm{A}_\pm,\quad2\Delta [\mm{A}_0^2,\mm{A}_\pm]  = [\mm{A}_\mp,\mm{A}_\pm^2] , \nonumber \\
&& 2\Delta\{\mm{A}^2_0,\mm{A}_\pm\} - 4\mm{A}_0\mm{A}_\pm\mm{A}_0 = \{\mm{A}_\mp,\mm{A}^2_\pm\} - 2\mm{A}_\pm\mm{A}_\mp\mm{A}_\pm.  \label{eq:algebra}
\end{eqnarray}
Considering the remaining terms of the defining relation (\ref{eq:def}) where the defect operator sits at the boundary we obtain
two triples of sufficient conditions projecting (in analogy to (\ref{eq:project})) on $\sigma^\z \otimes \sigma^\alpha$ for $x=1$, or $\sigma^\alpha \otimes \sigma^\z$ for $x=n$, for a pair of sites near the boundary
\begin{eqnarray}
\sum_{\alpha,\alpha'\in\{+,-,0\}} \gamma^{\alpha',\alpha}_\beta \bra{0}\mm{A}_\alpha \mm{A}_{\alpha'} &=&-\ii \varepsilon  \bra{0}\mm{A}_\beta, \qquad \beta\in\{+,-,0\}\nonumber\\
\sum_{\alpha,\alpha'\in\{+,-,0\}} \gamma^{\alpha,\alpha'}_\beta \mm{A}_\alpha \mm{A}_{\alpha'}\ket{0} &=&-\ii \varepsilon \mm{A}_\beta\ket{0}. \label{eq:bc1}
\end{eqnarray}
In order to fulfil the second defining relation (\ref{eq:bcOm}) with the MPA (\ref{eq:MPA1}), simple additional sufficient conditions are
\begin{equation}
\bra{0}\mm{A}_- = 0,\qquad \mm{A}_+\ket{0}=0, \qquad \bra{0}\mm{A}_0 = \bra{0},\qquad \mm{A}_0 \ket{0}=\ket{0}.
\label{eq:bc2}
\end{equation}
We have thus shown that a representation of $\mm{A}_\pm,\mm{A}_0$ fulfilling the cubic bulk relations (\ref{eq:algebra}) together with quadratic-linear boundary
conditions (\ref{eq:bc1},\ref{eq:bc2}) will provide MPA for NESS of $XXZ$ chain. 
Considering $\ket{0}$ as a {\em highest weight state}, we define a tower of states $\ket{k}$ spanning an infinite-dimensional
auxiliary representation space 
\begin{equation}
{\cal H}_{\rm a}={\rm lsp}\{ \ket{k} := (\mm{A}_-)^k\ket{0}, k = 0,1,2\ldots\}.
\end{equation}
Using Dirac's notation with the dual basis $\bra{k}$, satisfying $\braket{k}{k'}=\delta_{k,k'}$, the auxiliary operators then naturally take the following tridiagonal form
\begin{equation}
\mm{A}_0 = \sum_{k=0}^\infty a_k \ket{k}\!\bra{k},\;\;
\mm{A}_+ = \sum_{k=0}^\infty b_k \ket{k}\!\bra{k+1},\;\;
\mm{A}_- = \sum_{k=0}^\infty \ket{k+1}\!\bra{k},
\label{eq:AAA}
\end{equation}
where the $8$ relations (\ref{eq:algebra}) become equivalent to a pair of recurrence relations for $a_k,b_k$
\begin{eqnarray}
&& a_{k+1} - 2\Delta a_k + a_{k-1} = 0,\\
&& b_{k+1} - b_k = 2 a_{k+1}(\Delta a_{k+1} - a_k)
\end{eqnarray}
with boundary conditions (\ref{eq:bc1},\ref{eq:bc2}) yielding the initial conditions for the recurrence
\begin{equation}
a_0 = 1, \quad a_1 = \Delta + \frac{\ii\varepsilon}{2}, \quad b_0 = \ii \varepsilon.
\end{equation}
The solution can be expressed in terms of Chebyshev polynomials in $\Delta$, or more compactly, re-writing the anisotropy parameter
\begin{equation}
\cos\eta := \Delta,
\end{equation}
as
\begin{eqnarray}
a_k &=& \cos(k \eta) + \frac{\ii\varepsilon}{2} \frac{\sin(k \eta)}{\sin\eta},\nonumber\\
b_k &=& \sin((k+1)\eta)\left(\frac{\ii\varepsilon}{\sin\eta}\cos(\eta k) - \left(1 + \frac{\varepsilon^2}{4\sin^2\eta}\right)\sin(\eta k)\right).
\label{eq:result1}
\end{eqnarray}

\begin{figure}
 \centering	
\vspace{-1mm}
\includegraphics[width=0.7\columnwidth]{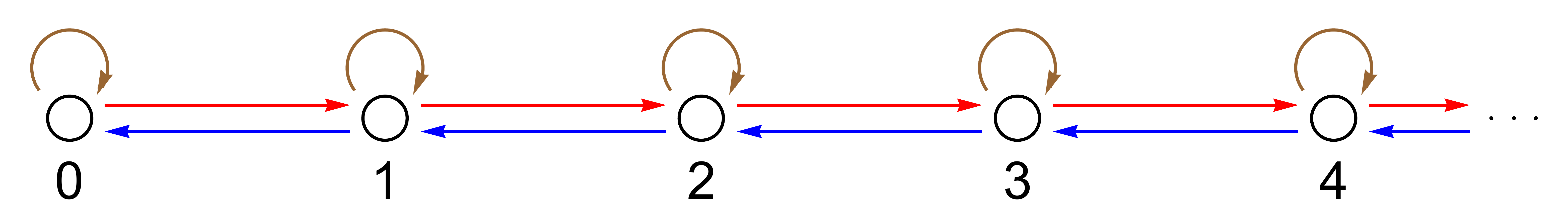}
\vspace{-1mm}
\caption{A semi-infinite graph showing the {\em allowed} transitions 
for building up the matrix product operator $\Omega_n$ (\ref{eq:WGSXXZ}) for the $XXZ$ chain. Brown, red, blue edges $e$ denote the index function values $\omega(e)=\sigma^0, \sigma^+, \sigma^-$, respectively.}
\label{fig:XXZD}
\end{figure}

MPA (\ref{eq:MPA1}) can be given another appealing interpretation. Associating the auxiliary states $\ket{k}$ with vertices of a directed graph, ${\cal V} =\{0,1,2,\ldots\}$, and defining a set of edges encoding all possible transitions
${\cal E}=\{(k,k),(k,k+1),(k+1,k);k=0,1,\ldots\}$ the amplitudes $\bra{0}\mm{A}_{\alpha_1}\mm{A}_{\alpha_2}\cdots\mm{A}_{\alpha_n}\ket{0}$ can be written in terms of a collection of products of transition amplitudes along all possible $n$-step {\em walks}, i.e. sequences of connecting edges starting at vertex $0$ and ending back at $0$ in exactly $n$ steps, ${\cal W}_n(0,0)=\{e_1,e_2\ldots, e_n\in {\cal E}; p(e_1) = 0, q(e_j)=p(e_{j+1}), q(e_n)=0 \}$, where $e\equiv (p(e),q(e))$, namely
\begin{equation}
\Omega_n = \sum_{(e_1,\ldots,e_n)\in{\cal W}_n(0,0)} a_{e_1}a_{e_2}\cdots a_{e_n} \omega(e_1)\otimes \omega(e_2)\otimes \cdots \omega(e_n).
\label{eq:WGSXXZ}
\end{equation}
The amplitudes are encoded as $a_{(k,k)} = a_k$, $a_{(k,k+1)}=b_k$, $a_{(k+1,k)}=1$, and $\omega : {\cal E}\to {\rm End}({\cal H}_{\rm p})$ is what we shall call an {\em index function} which associates a local physical operator with each edge of the graph. In our case $\omega(k,k) = \sigma^0$, $\omega(k,k+1) = \sigma^+$, $\omega(k+1,k) = \sigma^-$, where the {\em defect} $\sigma^\z$ is not in the image
of $\omega$ (see Fig.~\ref{fig:XXZD}). Expression (\ref{eq:WGSXXZ}) shall be named a {\em walking graph state} (WGS) representation, and is a useful concept encapsulating the locality of infinite-dimensional MPA. 

\subsection{Lax operator for a complex $q-$deformed spin and some notation}
\label{subsect:lax}
So far we followed a pedestrian approach and have not used the powerful quantum group  $U_q(\mathfrak{sl}_2)$ symmetry \cite{K95} of the $XXZ$ spin chain \cite{PS} which is deeply rooted behind the integrability of all equilibrium problems for the model. A fundamental characterisation of this symmetry can be written (see e.g. Refs.~\cite{F94,KBI93}) in terms of the so-called $RLL$ relation, a version of the YBE over a tensor product triple ${\cal H}_{\rm a}\otimes {\cal H}_{\rm p}\otimes{\cal H}_{\rm p} $
\begin{equation}
R_{1,2}(\varphi_1-\varphi_2)\mm{L}_{1}(\varphi_1)\mm{L}_{2}(\varphi_2) = 
\mm{L}_{1}(\varphi_2)\mm{L}_{2}(\varphi_1)R_{1,2}(\varphi_1-\varphi_2).
\label{eq:RLL}
\end{equation}
Remember that bold-roman letters denote symbols acting (nontrivially) over auxiliary space ${\cal H}_\aa$ while indices denote the label of the physical space.
Here $R(\varphi)$ is, up to a permutation, the standard trigonometric $4\times 4$ six-vertex $R$-matrix which we choose to write in terms of spin operators over ${\cal H}_{\rm p}\otimes{\cal H}_{\rm p}$
\begin{equation}
R(\varphi) = \frac{\sin\varphi}{2}(h + \one\cos\eta) - \frac{1+\cos\varphi}{2} \one\sin\eta + \frac{1-\cos\varphi}{2}\sigma^\z \otimes \sigma^\z\sin\eta,\quad
\label{eq:Rff}
\end{equation}
while the Lax operator (or $L$-operator) $\mm{L}\in{\rm End}({\cal H}_{\rm a}\otimes {\cal H}_{\rm p})$ has a universal $U_{q}(\mathfrak{sl}_2)$ symmetric form:
\begin{equation}
\mm{L}(\varphi,s) = \pmatrix{
\sin(\varphi+\eta \mm{s}^\z_s) & (\sin\eta) \mm{s}^-_s \cr
(\sin\eta) \mm{s}^+_s & \sin(\varphi-\eta \mm{s}^\z_s) \cr
} = \sum_{\alpha \in{\cal J}} \mm{L}^\alpha(\varphi,s) \otimes \sigma^\alpha,
\label{eq:Lax}
\end{equation}
where ${\cal J}=\{+,-,0,\z\}$, and $\mm{L}^\alpha \in {\rm End}({\cal H}_{\rm a})$ are its physical space components
\begin{equation}
\!\!\!\!\!\!\!\!\!\!\!\!\!\!\!\!\!\!
\mm{L}^0(\varphi,s) = \sin\varphi \cos(\eta \mm{s}^\z_s),\;
\mm{L}^\z(\varphi,s)  = \cos\varphi \sin(\eta \mm{s}^\z_s),\;
\mm{L}^\pm(\varphi,s) = (\sin \eta) \mm{s}_s^\mp.\quad
\end{equation}
For mathematical applications of quantum group symmetry to abstract construction of $L$-operators and other quantum integrability concepts, see e.g. Refs.~\cite{klumper1,klumper2}.
The $RLL$ relation (\ref{eq:RLL}) is in fact equivalent to a complete set of commutation relations for the generators $\mm{s}^\pm_s,\mm{s}^\z_s$ of a $q-$deformed angular momentum algebra with deformation parameter $q=e^{\ii \eta}$
\begin{equation}
[\mm{s}_s^+,\mm{s}_s^-] = [2\mm{s}^\z_s]_q=\frac{\sin(2\eta \mm{s}^\z_s)}{\sin\eta},\quad
[\mm{s}_s^\z,\mm{s}_s^\pm] = \pm \mm{s}_s^\pm,
\label{eq:algebraq}
\end{equation}
where $[x]_q:=(q^x-q^{-x})/(q-q^{-1})$. We shall here consider the highest weight representation with $\mm{s}^+\ket{0}=0, \mm{s}^\z_s\ket{0}=s \ket{0}$ carried by ${\cal H}_{\rm a} \equiv {\cal V}_{s}$, the so-called Verma module,
corresponding to a {\em complex} spin representation parameter $s\in\CC$:
\begin{eqnarray}
\mm{s}^\z_s &=& \sum_{k=0}^\infty (s-k) \ket{k}\bra{k}, \nonumber\\
\mm{s}^+_s &=& \sum_{k=0}^\infty \frac{\sin(k+1)\eta}{\sin\eta} \ket{k}\bra{k+1}, \label{verma} \\
\mm{s}^-_s &=& \sum_{k=0}^\infty \frac{\sin(2s-k)\eta}{\sin\eta} \ket{k+1}\bra{k}. \nonumber
\end{eqnarray}
We stress that the operator $\mm{s}^\z_s$ only exists within a representation, and not as an element of the $U_q(\mathfrak{sl}_2)$ where only
$q^{\pm \mm{s}^\z}$ exist.
Note that {\em only} for half-integer spin $s\in\half\ZZ^+=\{0,\half,1,\frac{3}{2}\ldots\}$, the module ${\cal V}_s$ is truncated, for any $\eta$, to a common finite-($2s+1$)-dimensional irrep, since $\ket{2s}$ then becomes the lowest weight state, $\mm{s}^-_s\ket{2s}=0$.
The module ${\cal V}_s$ is also truncated to a finite-$m$-dimensional one, for any $s$, when $q$ is a {\em generic} $m$-th root of unity, i.e., $\eta = \pi l/m, l,m\in \ZZ^+$.
For generic values of parameters $\eta,s$, the module ${\cal V}_s$ carries an infinite-dimensional irrep and Lax operator (\ref{eq:Lax}) is {\em non-Hermitian}, as in general $(\mm{s}^-_s)^\dagger \neq \mm{s}^+_s$, 
$(\mm{s}^\z_s)^\dagger \neq \mm{s}^\z_s$.

The matrix $R(\varphi)$ satisfies the following useful relations: it is symmetric under transposition,
its derivative at $\varphi=0$ is proportional to hamiltonian interaction, and it has a simple inverse proportional to $R(-\varphi)$:
\begin{eqnarray}
&& R(\varphi)^T = R(\varphi), \\
&& \partial_\varphi R(\varphi)|_{\varphi=0} = \half(h + \cos\eta \one), \\
&& R(\varphi)R(-\varphi) = (\sin^2\!\eta - \sin^2\!\varphi) \one.
\label{eq:invR}
\end{eqnarray}
Thus, expanding the $RLL$ relation (\ref{eq:RLL}) for $\varphi_{1,2} = \varphi \pm \delta$ to first order in $\delta$
one obtains a differential form of YBE, or the so-called {\em Sutherland relation} \cite{S70}
\begin{equation}
\!\!\!\!\!\!\!\!\!\!\!\! 
[h_{1,2},\mm{L}_{1}(\varphi,s)\mm{L}_{2}(\varphi,s)] = -2\sin\eta \left(\widetilde{\mm{L}}_{1}(\varphi,s)\mm{L}_{2}(\varphi,s) - \mm{L}_{1}(\varphi,s)\widetilde{\mm{L}}_{2}(\varphi,s)\right)
\label{eq:LOD}
\end{equation}
where 
\begin{equation}
\widetilde{\mm{L}}(\varphi,s) = \partial_\varphi\mm{L}(\varphi,s)  = \cos\varphi \cos(\eta\mm{s}^\z_s) \otimes \sigma^0 - \sin\varphi \sin(\eta\mm{s}^\z_s)\otimes \sigma^\z.
\end{equation}
This relation (\ref{eq:LOD}) is sometimes also referred to as {\em local operator divergence} condition.
Let us denote by 
\begin{equation}
\mm{L}^T(\varphi,s)=\sum_{\alpha\in{\cal J}} \mm{L}^\alpha(\varphi,s)\otimes (\sigma^\alpha)^T 
\end{equation}
the (partial) transposition with respect to the physical space, noting
$\widetilde{\mm{L}}^T(\varphi,s) = \widetilde{\mm{L}}(\varphi,s)$.
Sutherland relation transforms under partial transposition in physical spaces to
\begin{equation}
\!\!\!\!\!\!\!\!\!\!\!\!
[h_{1,2},\mm{L}^T_{1}(\varphi,s)\mm{L}^T_{2}(\varphi,s)] = 2\sin\eta \left(\widetilde{\mm{L}}_{1}(\varphi,s)\mm{L}^T_{2}(\varphi,s) - \mm{L}_{1}^T(\varphi,s)\widetilde{\mm{L}}_{2}(\varphi,s)\right).
\label{eq:LODt}
\end{equation}
We can think of 
$\mm{L}^T(\pi-\varphi,s)$
also as the Lax matrix corresponding to the transposed, lowest-weight representation ${\cal V}^T_s$ of $U_q(\mathfrak{sl}_2)$ which has exactly the canonical form (\ref{eq:Lax}) with spin operators transformed under the following algebra-(\ref{eq:algebraq})-preserving canonical transformation 
\begin{equation}
\mm{s}^{\z}_s \to -\mm{s}^\z_s,\quad \mm{s}^\pm_s \to \mm{s}^\mp_s,\label{eq:auxflip}
\end{equation}
which is just the {\em spin reversal} in auxiliary space.
The Lax operator (\ref{eq:Lax}) is invariant under the combined spin reversal in the physical and the auxiliary space
\begin{equation}
\sigma^\x \mm{L}(\varphi,s) \sigma^\x = \mm{L}^T(\pi-\varphi,s).
\label{eq:spinrevL}
\end{equation}

It will turn useful to study also the product of complex spin representations ${\cal V}_s^T \otimes {\cal V}_t$, for some $s,t\in\CC$.
Namely, defining a {\em double} Lax matrix as the following operator over the tensor product ${\rm End}({\cal H}_{\rm a}\otimes {\cal H}_{\rm b}\otimes {\cal H}_{\rm p})$
with a pair of auxiliary spaces carrying irreducible representations of $U_q(\mathfrak{sl}_2)$ with spin parameters $s,t\in\CC$, 
${\cal H}_{\rm a}={\cal V}^T_s$, ${\cal H}_{\rm b}={\cal V}_t$, and the corresponding spectral parameters $\varphi,\vartheta\in\CC$,
\begin{eqnarray}
\vmbb{L}_{x}(\varphi,\vartheta,s,t) &=& \mm{L}^T_{\aa,x}(\varphi,s) \mm{L}_{\bbb,x}(\vartheta,t), \label{eq:doubleLax}\\
\widetilde{\vmbb{L}}_{x}(\varphi,\vartheta,s,t) &=& \partial_\delta \left(\mm{L}^T_{\aa,x}(\varphi+\delta,s)\mm{L}_{\bbb,x}(\vartheta-\delta,t)\right)_{\delta=0} \nonumber \\
&=&\partial_\varphi\mm{L}^T_{\aa,x}(\varphi,s)\mm{L}_{\bbb,x}(\vartheta,t)-\mm{L}^T_{\aa,x}(\varphi,s)\partial_\vartheta\mm{L}_{\bbb,x}(\vartheta,t), 
\end{eqnarray} 
where $\mm{L}_{\aa,x} = \sum_{\alpha\in{\cal J}} \mm{L}^\alpha\otimes \one_{\rm b} \otimes \sigma^\alpha_x$, 
$\mm{L}^T_{\aa,x} = \sum_{\alpha\in{\cal J}} \mm{L}^\alpha\otimes \one_{\rm b} \otimes (\sigma^\alpha_x)^T$, 
$\mm{L}_{\bbb,x} = \sum_{\alpha\in{\cal J}} \one_{\rm a} \otimes \mm{L}^\alpha\otimes \sigma^\alpha_x$, 
we find again the corresponding Sutherland relation
\begin{equation}
[h_{1,2},\vmbb{L}_{1} \vmbb{L}_{2}] = 2\sin\eta 
\left( \widetilde{\vmbb{L}}_{1} \vmbb{L}_{2} - \vmbb{L}_{1} \widetilde{\vmbb{L}}_{2} \right).
\label{eq:LOD2}
\end{equation}
This identity can be proven directly using the Leibniz rule and Sutherland relations (\ref{eq:LOD},\ref{eq:LODt}), or it can be again derived by differentiating YBE
for the triple \footnote{Note that the physical spin space carries the fundamental representation of $U_q(\mathfrak{sl}_2)$, ${\cal H}_{\rm p}\equiv {\cal V}_{\frac{1}{2}}$.}
 ${\cal V}_{\frac{1}{2}}\otimes {\cal V}_{\frac{1}{2}} \otimes ({\cal V}^T_s\otimes {\cal V}_t)$, 
\begin{eqnarray}
&& R_{1,2}(\delta_1-\delta_2)\vmbb{L}_{1}(\varphi+\delta_1,\vartheta-\delta_1,s,t)\vmbb{L}_{2}(\varphi+\delta_2,\vartheta-\delta_2,s,t)  \\
&& =\vmbb{L}_{1}(\varphi+\delta_2,\vartheta-\delta_2,s,t)\vmbb{L}_{2}(\varphi+\delta_1,\vartheta-\delta_1,s,t) R_{1,2}(\delta_1-\delta_2).\label{eq:YBE2}
\end{eqnarray}
For a notational convenience, we are using double-strike-roman fonts to designate operators which act non-trivially over a tensor product of a pair of auxiliary spaces ${\cal H}_{\rm a}\otimes {\cal H}_{\rm b}$.

\subsection{Matrix product solution -- Lax operator method}
\label{subsect:XXX2}

Sutherland relations can be straightforwardly facilitated to solve/satisfy defining relations (\ref{eq:def},\ref{eq:bcOm}) for the amplitude operator. This idea has first been implemented in Ref.~\cite{KPS13}, even though
the corresponding Lax structure and Yang-Baxter equation have been identified only later in Ref.~\cite{PIP}. 

Writing (\ref{eq:LODt}) for a pair of physical sites $(x,x+1)$, multiplying with $\mm{L}^T_{1}\cdots\mm{L}^T_{x-1}$ from the left and with
$\mm{L}^T_{x+1}\cdots\mm{L}^T_{n}$ from the right, and summing over $x=1,\ldots,n$, one obtains a telescopic series yielding
\begin{equation}
[H,\mm{L}^T_{1}\mm{L}^T_{2} \cdots \mm{L}^T_{n}] = 2\sin\eta (\widetilde{\mm{L}}_{1}\mm{L}^T_{2}\cdots\mm{L}^T_{n} - 
\mm{L}^T_{1}\cdots \mm{L}^T_{n-1}\widetilde{\mm{L}}_{n}).
\label{HLT}
\end{equation}
Making an ansatz
\begin{equation}
\Omega_n = \frac{1}{\sin^n(\varphi+\eta s)}\bra{0}\mm{L}^T_{1}\mm{L}^T_2\cdots \mm{L}^T_n\ket{0}
\label{eq:ansatz1}
\end{equation}
one sees that since 
\begin{eqnarray}
\bra{0}\widetilde{\mm{L}} &=&  (\sigma^0  \cos\varphi\cos\eta s - \sigma^\z\sin\varphi\sin\eta s )\bra{0},\nonumber \\
\widetilde{\mm{L}}\ket{0} &=&  \ket{0}(\sigma^0\cos\varphi\cos\eta s  - \sigma^\z\sin\varphi\sin\eta s), \label{Ltb}
\end{eqnarray}
$\Omega_n$ satisfies (\ref{eq:def}) if
\begin{equation}
\varphi=\frac{\pi}{2}, \quad \tan\eta s = \frac{\ii \varepsilon}{2\sin\eta}.
\label{eq:ans1}
\end{equation}
The second condition (\ref{eq:bcOm}) of Lemma 1 is satisfied as well due to normalisation of the ansatz (\ref{eq:ansatz1}) and the lowest weight nature of representation.
Equivalently, since $[h,\sigma^\z\otimes \sigma^\z]=0$ one can use another gauge and take a twisted Lax operator $\mm{L}^T\sigma^\z$ again solving the Sutherland equation. Therefore, another ansatz
\begin{equation}
\Omega_n = \frac{1}{\sin^n(\varphi+\eta s)}\bra{0}\mm{L}^T_1\mm{L}^T_2\cdots \mm{L}^T_n\ket{0}(\sigma^\z)^{\otimes n}
\label{eq:ansatz2}
\end{equation}
solves the same defining equations (\ref{eq:def},\ref{eq:bcOm}), and provides the same NESS density operator (\ref{eq:rhoinf}) according to Lemma 1, if
\begin{equation}
\varphi=0,\quad \cot \eta s = -\frac{\ii \varepsilon}{2\sin\eta}.
\label{eq:ans2}
\end{equation}
Of course, the ansatz (\ref{eq:ansatz1}) provides just an alternative formulation of MPA (\ref{eq:MPA1}) with the matrices identified as
\begin{eqnarray}
&&\mm{A}_0=(\sec\eta s)\mm{L}^0(\frac{\pi}{2},s) = \frac{\cos(\eta \mm{s}^\z_s)}{\cos\eta s},\nonumber \\
&&\mm{A}_\pm= (\sec{\eta s})\mm{L}^\mp(\frac{\pi}{2},s) = \frac{\sin\eta}{\cos(\eta s)}\mm{s}^\pm_s,\quad
 \tan\eta s = \frac{\ii \varepsilon}{2\sin\eta}.\quad \label{eq:L2A}
\end{eqnarray}
Or, alternatively, picking representation (\ref{eq:ansatz2},\ref{eq:ans2}):
\begin{eqnarray}
\mm{A}_0=(\csc\eta s)\mm{L}^0(0,s) = \frac{\sin(\eta \mm{s}^\z_s)}{\sin\eta s},\nonumber\\
\mm{A}_\pm= (\csc{\eta s})\mm{L}^\mp(0,s) = \frac{\sin\eta}{\sin(\eta s)}\mm{s}^\pm_s,\quad \cot \eta s = -\frac{\ii \varepsilon}{2\sin\eta}. \label{eq:L2Ab}
\end{eqnarray}
Both cases reproduce exactly the result (\ref{eq:result1}), up to a gauge transformation $\ket{k}\to c_k\ket{k},\bra{k}\to c^{-1}_k\bra{k}$ which does not affect the MPA.
For instance, $c_k$ can even be chosen to make all transition amplitudes of $\mm{A}^\pm,\mm{A}$ {\em linear} in dissipation $\varepsilon$ as in Ref.~\cite{P11b}.

We note that, as a consequence of the spin-reversal symmetry of the Lax operator (\ref{eq:spinrevL}), 
the non-transposed, highest-weight Lax operator at (\ref{eq:ans1}) reproduces the lower-triangular amplitude operator $\Omega'_n$ (\ref{eq:Omp})
for the {\em reverse driving} (\ref{eq:revdriv})
\begin{equation}
\Omega'_n = \frac{1}{\sin^n(\varphi+\eta s)}\bra{0}\mm{L}_1\mm{L}_2\cdots \mm{L}_n\ket{0}.
\end{equation}

\subsection{Matrix product solution -- the case of isotropic ($XXX$) chain}

\label{subsect:XXX}

Let us here briefly list the result (appearing e.g. in Ref.~\cite{PIP}) for $SU(2)$ symmetric $XXX$ chain with $\Delta=1$ which correspond to the limit $\eta\to 0$, properly
normalised when needed, of the results derived in the previous subsections. The spectral parameter we set now as $\lambda = \varphi/\eta$, so the $R$-matrix
and the Lax operator read, respectively,
\begin{eqnarray}
R(\lambda)&\longleftarrow&\lim_{\eta\to 0} \frac{1}{\eta}R(\varphi) = \one_4 + \lambda P_{1,2} \\
\mm{L}(\lambda,s)&\longleftarrow&\lim_{\eta\to 0} \frac{1}{\eta} \mm{L}(\varphi,s) = \pmatrix{ \lambda\one_{\rm a} + \mm{s}^\z_s & \mm{s}^-_s \cr 
\mm{s}^+_s & \lambda\one_{\rm a} - \mm{s}^\z_s} = \lambda \one + \vec{\mm{s}}_s\cdot\vec{\sigma}, \nonumber
\end{eqnarray}
where $\vec{\mm{s}}_s=(\half(\mm{s}^+_s+\mm{s}^-_s),-\ihalf(\mm{s}^+_s - \mm{s}^-_s),\mm{s}^\z_s)$ and $P_{1,2}=\half(\vec{\sigma}_1\cdot\vec{\sigma}_2 + \one)$ is a permutation operator over ${\cal H}_{\rm p}\otimes {\cal H}_{\rm p}$ and $\mm{s}^\z_s,\mm{s}^\pm_s$ have become standard
generators of angular-momentum algebra $\mathfrak{sl}_2$ for a complex spin $s$
\begin{eqnarray}
\mm{s}^\z_s &=& \sum_{k=0}^\infty (s-k) \ket{k}\!\bra{k}, \nonumber\\
\mm{s}^+_s &=& \sum_{k=0}^\infty (k+1) \ket{k}\!\bra{k+1}, \label{verma} \\
\mm{s}^-_s &=& \sum_{k=0}^\infty (2s-k) \ket{k+1}\!\bra{k}, \nonumber
\end{eqnarray}
which is genuinely infinite-dimensional, unless $s \in \half\ZZ^+$.
The MPA (\ref{eq:MPA1}) for NESS amplitude operator has now generating matrices, obtained from the case (\ref{eq:L2Ab}), as
\begin{equation}
\mm{A}_0 = \frac{\mm{s}^\z_s}{s},\quad
\mm{A}_\pm = \frac{\mm{s}^\pm_s}{s},\quad
s = \frac{2\ii}{\varepsilon}.
\end{equation}

\subsection{Generalizations of boundary driving}

The solution of NESS for $XXZ$ model  presented so far refers to an extremely minimalistic boundary condition with a pure source and pure sink of 
equal rates at each end. Using generalisations of our approach we can expand the integrable boundary conditions in three different directions: (i) allowing arbitrary left-right asymmetry in the source-sink rates combined with additional arbitrary magnetic fields at the boundary sites,
(ii) allowing the source and sink jump operators to act with respect to non-parallel axes (like $\z$-axis in our previous example),
so the two target states cannot be chosen mutually diagonal, (iii) allowing four different rates at the two boundaries for two independent in\&out 
processes at each end, but only perturbatively in the system-bath coupling constant.
We shall describe these developments in some detail in the three paragraphs below.

\subsubsection{Left-right asymmetry and combined coherent-incoherent driving.}
\label{subsect:asym}

Here we are considering the quantum master equation with a combination of asymmetric coherent and incoherent driving \cite{PIP15}. The former is provided by adding an arbitrary magnetic field
at the chain ends
\begin{equation}
H_{\rm b} = \sum_{x=1}^{n-1} h_{x,x+1} + b_{\rm L} \sigma^\z_1 + b_{\rm R} \sigma^\z_n = H + b_{\rm L} \sigma^\z_1 + b_{\rm R} \sigma^\z_n,
\label{eq:XXZbf}
\end{equation}
while for the latter we allow arbitrary rates of the source and the sink
\begin{equation}
L_1 = \sqrt{\Gamma_{\rm L}}\sigma^+_1,\quad L_2 = \sqrt{\Gamma_{\rm R}}\sigma^-_n,
\label{eq:asymL}
\end{equation}
so the total Liovillian generator reads $\LL = -\ii\,{\rm ad}\, H_{\rm b} + \Gamma_{\rm L} \DD_{\sigma^+_1} + \Gamma_{\rm R} \DD_{\sigma^-_n}$.

We start by making the following ansatz for NESS
\begin{equation}
\!\!\!\!\!\!\!\!\!\!\!\!\!
R_\infty = \bra{0,0}\prod_{x=1}^n \vmbb{L}_{x}(\varphi,\vartheta,s,t)\ket{0,0} K^{\otimes n} = \bra{0,0}(\vmbb{L}(\varphi,\vartheta,s,t)K)^{\otimes n}\ket{0,0}
\label{eq:Ransatz}
\end{equation}
where $\vmbb{L}_x$ is the double Lax operator (\ref{eq:doubleLax}) over the tensor product of a pair of Verma modules ${\cal H}_{\rm a}\otimes{\cal H}_{\rm b}\equiv {\cal V}^T_s\otimes{\cal V}_t$ and 
\begin{equation}
K = K(\chi) := \pmatrix{ 
\chi^{1/2} & 0 \cr 
0 & \chi^{-1/2} }
\end{equation}
is a magnetization-shift matrix satisfying, for any $\chi \in\RaR^+$,
\begin{equation}
[h,K\otimes K]=0.\label{eq:MM}
\end{equation}
Using Sutherland relation (\ref{eq:LOD2}) and Eq.~(\ref{eq:MM}) one again implements the telescopic series to find for the commutator with the entire Hamiltonian (\ref{eq:XXZbf}):
\begin{eqnarray}
[H_{\rm b}, (\vmbb{L} K)^{\otimes n}] &=& \left(2(\sin\eta)\widetilde{\vmbb{L}}K+b_{\rm L}[\sigma^\z,\vmbb{L}K]\right) \otimes (\vmbb{L}K)^{\otimes n-1} \nonumber \\
&-& (\vmbb{L}K)^{\otimes n-1}\otimes \left(2(\sin\eta)\widetilde{\vmbb{L}}K-b_{\rm R}[\sigma^\z,\vmbb{L}K]\right).
\label{eq:bulk}
\end{eqnarray}
Hence the sufficient condition for the steady state Lindblad equation 
\begin{equation}
\ii [H_{\rm b},R_\infty] = \Gamma_L \DD_{\sigma^+_1}(R_\infty) + \Gamma_R \DD_{\sigma^-_n}(R_\infty)
\end{equation}
to hold is to satisfy a pair of boundary equations on ${\cal H}_{\aa}\otimes {\cal H}_{\bbb}\otimes {\cal H}_{\rm p}$:
\begin{eqnarray}
&& \bra{0,0} \left(-2\ii(\sin\eta) \widetilde{\vmbb{L}}K + \Gamma_{\rm L} \DD_{\sigma^+}(\vmbb{L}K) - \ii b_{\rm L}[\sigma^\z,\vmbb{L}K] \right) = 0, \nonumber \\
&& \left(2\ii(\sin\eta)\widetilde{\vmbb{L}}K + \Gamma_{\rm R} \DD_{\sigma^-} (\vmbb{L}K)  - \ii b_{\rm R}[\sigma^\z,\vmbb{L}K]\right)\ket{0,0} = 0.
\label{eq:be1}
\end{eqnarray}
For convenience of calculations one may write the local left and right dissipators as explicit matrix maps over ${\rm End}({\cal H}_{\rm p})$
\begin{equation}
\DD_{\sigma^+} \pmatrix{a & b \cr c & d} = \pmatrix{2d & -b \cr -c & -2d},\quad
\DD_{\sigma^-} \pmatrix{a & b \cr c & d} = \pmatrix{-2a & -b \cr -c & 2a}.
\end{equation}
Let us write the two representation parameters as $s,t\in\CC$ and the corresponding generators of $U_q(\mathfrak{sl}_2)$ as 
\begin{eqnarray}
\mm{s}^\pm\equiv\mm{s}^\pm_s\otimes \one_{\rm b},\;
\mm{s}\equiv \mm{s}^\z_s\otimes\one_{\rm b},\;
\mm{t}^\pm\equiv\one_{\rm a}\otimes \mm{s}^\pm_t,\;
\mm{t}\equiv \one_{\rm a}\otimes\mm{s}^\z_t,
\end{eqnarray}
hence
\begin{eqnarray}
\!\!\!\!\!\!\!\!\!\!\!\!\!\!\!\!\!\!\!\!\!\!\!\!\!\!\!\!\!\!\!\!\!\!\!\!
\vmbb{L} &=& \pmatrix{ \sin(\varphi\!-\!\eta\mm{s})\sin(\vartheta\!-\!\eta\mm{t}) + (\sin^2\eta)\mm{s}^+\mm{t}^+ &
(\sin(\varphi\!-\!\eta\mm{s})\mm{t}^- +\sin(\vartheta\!+\!\eta\mm{t})\mm{s}^+)\sin\eta\cr
(\sin(\vartheta\!-\!\eta\mm{t})\mm{s}^- + \sin(\varphi\!+\!\eta\mm{s})\mm{t}^+)\sin\eta & 
\sin(\varphi\!+\!\eta\mm{s})\sin(\vartheta\!+\!\eta\mm{t}) + (\sin^2\eta)\mm{s}^-\mm{t}^-}\!, \nonumber \\
\!\!\!\!\!\!\!\!\!\!\!\!\!\!\!\!\!\!\!\!\!\!\!\!\!\!\!\!\!\!\!\!\!\!\!\!
\widetilde{\vmbb{L}} &=&
\pmatrix{ \sin(\varphi-\eta\mm{s}-\vartheta+\eta \mm{t}) &
(\cos(\vartheta\!+\!\eta\mm{t})\mm{s}^+\!-\!\cos(\varphi\!-\!\eta\mm{s})\mm{t}^-)\sin\eta 
\cr
(\cos(\vartheta\!-\!\eta\mm{t})\mm{s}^-\!-\!\cos(\varphi\!+\!\eta\mm{s})\mm{t}^+) \sin\eta & 
\sin(\varphi+\eta\mm{s}-\vartheta-\eta\mm{t})}\!\!.\quad \label{eq:LL2}
\end{eqnarray}
Noting the identities 
\begin{eqnarray}
&& \mm{s}\ket{0,0} = s\ket{0,0},\quad \mm{s}^+\ket{0,0} = 0, \quad \mm{s}^-\ket{0,0}= (\sin(2\eta s)/\sin\eta)\ket{1,0},\nonumber \\
&& \bra{0,0}\mm{s}=s\bra{0,0} \quad \bra{0,0}\mm{s}^+=\bra{1,0},\quad \bra{0,0}\mm{s}^-=0,\nonumber \\
&& \mm{t}\ket{0,0} = t\ket{0,0}, \quad \mm{t}^+\ket{0,0} = 0,\quad \mm{t}^-\ket{0,0}= (\sin(2\eta t)/\sin\eta)\ket{0,1}, \nonumber \\
&& \bra{0,0}\mm{t}=t\bra{0,0}, \quad \bra{0,0}\mm{t}^+=\bra{0,1},\quad \bra{0,0}\mm{t}^-=0, 
\end{eqnarray} 
the boundary equations (\ref{eq:be1}) amount to two sets of $2\times 2$ equations (components in ${\rm End}({\cal H}_{\rm p})$), where only 5 are independent for 5 unknown parameters $s,t,\varphi,\vartheta,\chi$, specifically:
\begin{eqnarray}
&&\tan(\varphi-\eta s) = \frac{2 \ii \sin\eta}{\Gamma_{\rm L}-2\ii b_{\rm L}},\label{eq:asymb}\\
&&\tan(\varphi+\eta s) = -\frac{2 \ii \sin\eta}{\Gamma_{\rm R}+2\ii b_{\rm R}}, \\
&&\tan(\vartheta-\eta t) = -\frac{2 \ii \sin\eta}{\Gamma_{\rm L}+2\ii b_{\rm L}},\\ 
&&\tan(\vartheta+\eta t) = \frac{2 \ii \sin\eta}{\Gamma_{\rm R}-2\ii b_{\rm R}},\label{eq:asymc} \\
&&\chi^2 \frac{\sin(\vartheta+\eta t - \varphi - \eta s)}{\sin(\vartheta-\eta t - \varphi + \eta s)} = -\frac{\Gamma_{\rm L}}{\Gamma_{\rm R}} \frac{\sin(\varphi-\eta s)}{\sin(\varphi+\eta s)} \frac{\sin(\vartheta-\eta t)}{\sin(\vartheta+\eta t)},
\label{eq:asyme}
\end{eqnarray}
while the other 3 equations reduce to identities. One then easily finds an explicit solution
\begin{eqnarray}
&&\varphi = \bar{\vartheta} = \frac{\ii}{2}(z_{\rm L} - z_{\rm R}),\label{eq:varphi}\\
&&\eta s = \overline{\eta t} = \frac{\ii}{2}(z_{\rm L} + z_{\rm R}),\\
&&\chi = \bar{\chi} =  \frac{\Gamma_{\rm L}}{\Gamma_{\rm R}} \sqrt{\frac{(\quart\Gamma_{\rm R}^2-b_{\rm R}^2 -\sin^2 \eta)^2+ \Gamma^2_{\rm R} b^2_{\rm R}}{(\quart\Gamma_{\rm L}^2-b_{\rm L}^2 - \sin^2\eta)^2 + \Gamma^2_{\rm L} b^2_{\rm L}}}, 
\label{eq:parXXZ1}
\end{eqnarray}
where
\begin{equation}
z_{\rm L} := \frac{1}{\ii}\arctan \frac{2\ii \sin\eta}{\Gamma_{\rm L}-2\ii b_{\rm L}},\quad
z_{\rm R} := \frac{1}{\ii}\arctan \frac{2\ii \sin\eta}{\Gamma_{\rm R}+2\ii b_{\rm R}}.
\end{equation}
Note that such NESS is again of Cholesky form, namely defining 
\begin{equation}
\Omega_n(\varphi,s,\chi) := \bra{0}\mm{L}^T_1(\varphi,s)\mm{L}^T_2(\varphi,s)\cdots\mm{L}^T_n(\varphi,s)\ket{0} \pmatrix{\chi^{1/4} & 0 \cr 0 & \chi^{-1/4}}^{\otimes n}\!\!\!\!\!\!,\;
\label{eq:OmegaXXZ}
\end{equation} 
which is a lower-triangular matrix, one can write the non-normalized density operator of NESS, since $[\Omega(\varphi,s,\chi)]^\dagger = [\Omega(\bar{\varphi},\bar{s},\chi)]^T$,
as
\begin{equation}
R_\infty = \Omega_n(\varphi,s,\chi) [\Omega_n(\varphi,s,\chi)]^\dagger.
\label{eq:nessXXZ1}
\end{equation}

In the isotropic case $\Delta=1$ the NESS solution, after writing $\lambda=\varphi/\eta$ and taking the limit $\eta\to 0$, can be written compactly as
\begin{eqnarray}
&&R_\infty = \Omega_n(\lambda,s,\chi)[\Omega_n(\lambda,s,\chi)]^\dagger,\nonumber\\
&& \Omega_n(\lambda,s,\chi) = \bra{0}\mm{L}^T(\lambda,s)^{\otimes n}\ket{0} K(\sqrt{\chi})^{\otimes n},
\label{eq:nessXXX1}
\end{eqnarray}
with the spectral, representation and magnetisation parameters, respectively
\begin{eqnarray}
\lambda &=& \frac{\ii}{\Gamma_{\rm L} - 2\ii b_{\rm L}} - \frac{\ii}{\Gamma_{\rm R} + 2\ii b_{\rm R}}, \nonumber \\
s &=& \frac{\ii}{\Gamma_{\rm L} - 2\ii b_{\rm L}} + \frac{\ii}{\Gamma_{\rm R} + 2\ii b_{\rm R}}, \nonumber\\
\chi &=& \frac{(\Gamma_{\rm R}^2 + 4 b_{\rm R}^2)\Gamma_{\rm L}}{(\Gamma_{\rm L}^2 + 4 b_{\rm L}^2)\Gamma_{\rm R}}.
\label{eq:parXXX1}
\end{eqnarray}

\subsubsection{$SU(2)-$twisted boundary driving.}

\label{subsect:twist}

In the isotropic case $\Delta=1$ we shall now further exploit the $SU(2)$ symmetry of the problem to map the solution (\ref{eq:nessXXX1},\ref{eq:parXXX1}) to the 
NESS for a more general class of dissipators (essentially following Refs.~\cite{KPS13,PKS13}). We start by elaborating on rotation symmetry of the double Lax operator $\vmbb{L}$ entering the boundary conditions (\ref{eq:be1}).
We construct a {\em pseudo-representation} of the rotation group over ${\cal H}_\aa\otimes {\cal H}_\bbb \otimes {\cal H}_{\rm p} \equiv {\cal V}^T_s \otimes {\cal V}_t\otimes {\cal V}_{\frac{1}{2}}$,
namely choosing an angle $\theta$ and a {\em unit vector} $\vec{u}$ (axis of rotation) we define
\begin{eqnarray}
\vmbb{U}(\theta,\vec{u}) &=& \exp\left(\ii \theta \vec{u}\cdot\left(\vec{\mm{s}}\otimes \one_2 +\vec{\mm{t}}\otimes \one_2 + \one_\aa\otimes\one_\bbb\otimes\frac{\vec{\sigma}}{2}\right)\right) \\
                  &=& \mm{U}_s(\theta,\vec{u})\otimes \mm{U}_t(\theta,\vec{u}) \otimes U(\theta,\vec{u}), \label{eq:factorSU2}
\end{eqnarray}
where $\mm{U}_s(\theta,\vec{u}) = \exp(\ii\theta \vec{u}\cdot\vec{\mm{s}}_s)$, $U(\theta,\vec{u}) = \exp(\ii\frac{\theta}{2}\vec{u}\cdot\vec{\sigma})$, {\em formally} implementing
\footnote{One perhaps needs to stress that the operator $\mm{U}_s$ may not exist as an element of ${\rm End}({\cal H}_{\rm a})$ but its action on the highest weight state $\mm{U}_s \ket{0}$, 
$\bra{0}\mm{U}_s $ is well-defined and computable, and that is all what we need here.}
the full $SL(2)$ symmetry of the non-Hermitian double Lax operator
\begin{equation}
\vmbb{U}(\theta,\vec{u}) \vmbb{L} \vmbb{U}(-\theta,\vec{u}) = \vmbb{L},\quad \vmbb{U}(\theta,\vec{u}) \widetilde{\vmbb{L}} \vmbb{U}(-\theta,\vec{u}) = \widetilde{\vmbb{L}}.
\label{eq:fullSU2}
\end{equation}
The boundary equations (\ref{eq:be1}), substituting RHSs of (\ref{eq:fullSU2}) by the corresponding LHSs, 
and using factorisation (\ref{eq:factorSU2}), and requiring $\chi=1$ so that $[U(\theta,\vec{u}),K]\equiv 0$,
map to
\begin{eqnarray}
&&\!\!\!\!\!\!\!\!\!\!\!\!\!\!\!\!
(\bra{0}\mm{U}_s\otimes \bra{0}\mm{U}_t)\left(-2\ii(\sin\eta) \widetilde{\vmbb{L}} + \Gamma_{\rm L} \DD_{U\sigma^+U^{\dagger}}(\vmbb{L}) - \ii b_{\rm L}[U\sigma^\z U^\dagger,\vmbb{L}] \right) = 0, \nonumber \\
&&\!\!\!\!\!\!\!\!\!\!\!\!\!\!\!\!
\left(2\ii(\sin\eta)\widetilde{\vmbb{L}} + \Gamma_{\rm R} \DD_{U \sigma^- U^\dagger} (\vmbb{L})  - \ii b_{\rm R}[U \sigma^\z U^\dagger,\vmbb{L}]\right)
(\mm{U}^\dagger_s \ket{0}\otimes \mm{U}^\dagger_t\ket{0}) = 0. \qquad
\label{eq:betwisted}
\end{eqnarray}
Note that the two formal $SL(2)$ transformations for two, left and right boundary conditions can be independent, say $\vmbb{U}(\theta_{\rm L},\vec{u}_{\rm L})$ and 
$\vmbb{U}(\theta_{\rm R},\vec{u}_{\rm R})$, and without loss of generality we may chose the axes of rotation to lie in the $x-y$ plane, 
$\vec{u}_{\rm L/R}=(\sin\phi_{\rm L/R},-\cos\phi_{\rm L/R},0)$.
Thus the Eqs.~(\ref{eq:betwisted}), in combination with $SU(2)$ invariant bulk condition (\ref{eq:bulk}) [for $K=\one_2$] and noting that $t=\bar{s}$, implies that
the density matrix
\begin{eqnarray}
&&R_\infty = \Omega^{\rm twist}_n(\lambda,s,\chi)[\Omega^{\rm twist}_n(\lambda,s,\chi)]^\dagger,\nonumber\\
&& \Omega^{\rm twist}_n(\lambda,s,\chi) = \bra{\psi_{\rm L}} \mm{L}^T(\lambda,s)^{\otimes n}\ket{\psi_{\rm R}} ,\\
&& \bra{\psi_{\rm L}} = \bra{0} \mm{U}_s (\theta_{\rm L},\vec{u}_{\rm L}),\nonumber\\
&& \ket{\psi_{\rm R}} = \mm{U}_s(-\theta_{\rm R},\vec{u}_{\rm R})\ket{0}, \label{eq:coh1}
\end{eqnarray}
with $\lambda$ and $s$ determined from the first two lines of (\ref{eq:parXXX1}),
represents an exact NESS of the Lindbladian dynamics with the twisted jump operators
\begin{eqnarray}
&& L^{\rm twisted}_1 = U(\theta_{\rm L},\vec{u}_{\rm L}) \sigma^+_1 U(-\theta_{\rm L},\vec{u}_{\rm L})\\
&&=  \frac{e^{-\ii\varphi_{\rm L}}}{2}(
\cos\theta_{\rm L}\cos\phi_{\rm L}-\ii\sin\phi_{\rm L},
\cos\theta_{\rm L}\sin\phi_{\rm L}+\ii\cos\phi_{\rm L},
-\sin\theta_{\rm L})\cdot\vec{\sigma}_1, \nonumber\\
&&L^{\rm twisted}_2 = U(\theta_{\rm R},\vec{u}_{\rm R}) \sigma^-_n U(-\theta_{\rm R},\vec{u}_{\rm R}) = \nonumber\\
&&=  \frac{e^{\ii\varphi_{\rm R}}}{2}(
\cos\theta_{\rm R}\cos\phi_{\rm R}+\ii\sin\phi_{\rm R},
\cos\theta_{\rm R}\sin\phi_{\rm R}-\ii\cos\phi_{\rm R},
-\sin\theta_{\rm R})\cdot\vec{\sigma}_n, \nonumber 
\end{eqnarray}
and for the Hamiltonian with twisted boundary fields:
\begin{eqnarray}
H^{\rm twist}_{\rm b} = \sum_{x=1}^{n-1} h_{x,x+1} &+& b_{\rm L} (\sin\theta_{\rm L}\cos\phi_{\rm L},\sin\theta_{\rm L}\sin\phi_{\rm L},\cos\theta_{\rm L})\cdot\vec{\sigma}_1\nonumber\\
 &+& b_{\rm R}   (\sin\theta_{\rm R}\cos\phi_{\rm R},\sin\theta_{\rm R}\sin\phi_{\rm R},\cos\theta_{\rm R})\cdot\vec{\sigma}_n.
\label{eq:XXZtw}
\end{eqnarray}
The states $\bra{\psi_{\rm L}}$ and $\ket{\psi_{\rm R}}$, (\ref{eq:coh1}), are just the $SU(2)$ coherent states over the complex spin Verma module, namely  
\begin{eqnarray}
&&\!\!\!\!\!\!\!\!\!\!\!\!\!\!\!\!\!\!\!\!\!\!\!
\bra{\psi_{\rm L}} = \bra{0}\exp( \psi_{\rm L}\mm{s}^+_s-\bar{\psi}_{\rm L}\mm{s}^-_s) = |\!\cos\theta_{\rm L}|^{2s}\sum_{k=0}^\infty
 {2s\choose k}\left(-\tan{\frac{\theta_{\rm L}}{2}}\right)^k e^{-\ii k\phi_{\rm L}}\bra{k}
 ,\nonumber \\
&&\!\!\!\!\!\!\!\!\!\!\!\!\!\!\!\!\!\!\!\!\!\!\!
\ket{\psi_{\rm R}} = \exp( \psi_{\rm R} \mm{s}_s^+ - \bar{\psi}_{\rm R}\mm{s}^-_s)\ket{0} = |\!\cos\theta_{\rm R}|^{2s}\sum_{k=0}^\infty
\left(-\tan\frac{\theta_{\rm R}}{2}\right)^k e^{\ii k\phi_{\rm R}}\ket{k},
\label{eq:cs}
\end{eqnarray}
where the complex coherent-state parameters read
\begin{equation}
\psi_{\rm L} = -\frac{\theta_{\rm L}}{2}e^{-\ii\phi_{\rm L}},\quad \psi_{\rm R} = \frac{\theta_{\rm R}}{2}e^{-\ii\phi_{\rm R}}.
\end{equation}
The expansions are derived using a well known $SU(2)$ disentangling formula:
\begin{equation}
\!\!\!\!\!\!\!\!\!\!\!\!\!\!\!\!\!\!\!\!\!\!\!\!\!\!\!\!\!\!\!\!\!\!\!
\exp(\theta e^{\ii \phi} \mm{s}^+ - \theta e^{-\ii \phi} \mm{s}^-) =\exp(-\mm{s}^- e^{-\ii\phi}\tan\theta) \exp(2\mm{s}^\z\log|\!\cos\theta|)  \exp(\mm{s}^+ e^{\ii\phi}\tan\theta).
\end{equation}
Note that solvability condition $(\chi=1)$ in this case
\begin{equation}
(\Gamma^2_{\rm L} + 4 h^2_{\rm L})\Gamma_{\rm R} = (\Gamma^2_{\rm R} + 4 h^2_{\rm R})\Gamma_{\rm L}.\quad
\end{equation}
still allows some left-right asymmetry, in which case our solution goes beyond what has been discussed in Ref.~\cite{KPS13,PKS13}. Due to rotational symmetry, we could of course without loss of generality  fix three out of four angles $\phi_{\rm L}=\phi_{\rm L}=\theta_{\rm R}=0$ and keep only the relative twisting angle $\theta_{\rm R}$ between the
source and sink measurement axes \cite{PKS13}. Nevertheless, it is perhaps illuminating to write out the full rotationally symmetric parametrisation of the twisted boundary driven solution for possible further generalisations or deformations.

As we have seen, general analytic solutions for twisted boundary driving are limited to the isotropic case $\Delta=1$ only. However, interesting nontrivial properties of the spin current under twisted driving have been observed in the 
anisotropic case $\Delta\neq 1$ in Ref.~\cite{popkov} by exact analysis of short chains. The question of non-equilibrium integrability in such a case remains open.
 
\subsubsection{Perturbative driving with source\&sink processes on each end.}

The explicit forms of NESS of boundary driven $XXZ$ chains that we discussed so far were all characterized with simple, ultra local, rank-one dissipators, which can
be considered as a source and a sink of spin excitations with respect to some measurement bases. In the classical integrable locally interacting Markov chains, e.g. ASEP with open boundaries \cite{BE07}, one can analytically solve more complicated boundary conditions with in and out incoherent processes on each side. In the
framework of Lindblad equation, these would be described by four Lindblad channels,  with four non-negative rates $\Gamma^\pm_{\rm L/R} \ge 0$:
\begin{equation}
L_1 = \sqrt{\Gamma^+_{\rm L}}\sigma^+_1,\quad
L_2 = \sqrt{\Gamma^-_{\rm L}}\sigma^-_1,\quad
L_3 = \sqrt{\Gamma^+_{\rm R}}\sigma^+_n,\quad
L_4 = \sqrt{\Gamma^-_{\rm R}}\sigma^-_n,\;\;
\label{eq:L4}
\end{equation}
generating the Liouvilllian $\LL = -\ii\,{\rm ad}\,H_{\rm b} + \sum_{\mu=1}^4 \DD_{L_\mu}$ in terms of Hamiltonian (\ref{eq:XXZbf}) and canonical dissipators (\ref{eq:candis}).
In analogy to ASEP, one may hope that the NESS can be written with an ansatz generalizing (\ref{eq:Ransatz}),
 $R_\infty = \bra{\Psi_{\rm L}}(\vmbb{L}K)^{\otimes n}\ket{\Psi_{\rm R}}$, where $\ket{\Psi_{\rm L/R}} \in {\cal H}_\aa\otimes{\cal H}_\bbb$ are some free auxiliary states.
However, a straightforward calculation performed by the author showed that the above ansatz is {\em insufficient}, i.e. there is generally no solution for $ \ket{\Psi_{\rm L/R}}$ and parameters
$\varphi,\vartheta,s,t,\chi$. Therefore, finding an exact solution of NESS for $XXZ$ chain driven by four channel boundary dissipation remains a challenging open problem.

What one can do instead is to solve the problem perturbatively if all the coupling and driving rates are uniformly small (see Ref.~\cite{P11a}).
Writing 
\begin{equation}
\Gamma^\pm_{\rm L/R} = \varepsilon \gamma^\pm_{\rm L/R},\quad b_{\rm L/R} = \varepsilon \mu_{\rm L/R},
\end{equation} 
where $\varepsilon$ is a {\em formal small parameter} and expressing an un-normalized density operator as a power series 
\begin{equation}
R_\infty = \sum_{p=0}^\infty (\ii \varepsilon)^p \rho^{(p)}
\label{eq:pertseries}
\end{equation}
we get a recurrence relation connecting subsequent orders (where $H$ below now denotes the free-boundary $XXZ$ Hamiltonian (\ref{eq:HXXZ}))
\begin{eqnarray}
&&(\!\ad H)\rho^{(0)} = 0, \\
&&(\!\ad H)\rho^{(p)} = -\DD(\rho^{(p-1)}), \quad p=1,2\ldots,\quad {\rm where}\label{eq:pertp}\\
&&\DD=\gamma^+_{\rm L}\DD_{\sigma^+_1}+\gamma^-_{\rm L}\DD_{\sigma^-_1} + \gamma^+_{\rm R}\DD_{\sigma^+_n}+\gamma^-_{\rm R}\DD_{\sigma^-_n}-
\ii\mu_{\rm R}\ad\sigma^\z_n- \ii\mu_{\rm L}\ad\sigma^\z_1. \nonumber
\end{eqnarray}
From the uniqueness of NESS it follows that each term $\rho^{(p)}$ in the formal series expansion (\ref{eq:pertseries}) should also be determined uniquely.
This means that even though at each fixed order $p$, the solution of Eq.~(\ref{eq:pertp}) $\rho^{(p)}$ is undetermined up to addition of an arbitrary element from the {\em kernel}
of $\ad H$ (operator which commutes with $H$), there is always a unique element $\rho^{(p)}$ such that $\DD (\rho^{(p)})$ is in the {\em image} of $\ad H$ (it is Hilbert-Schmidt orthogonal $\tr\bigl\{X^\dagger \DD(\rho^{(p)})\bigr\} = 0$ to all operators $X$ which commute with $H$), so that the equation in the next order $p+1$ can have a solution.
Note that the map $\ad H$ is self-adjoint w.r.t. Hilbert-Schmidt inner product hence the orthogonal complement of its image is its kernel.

In the leading nontrivial order, the perturbative solution of NESS can be encoded compactly in terms of the $Z$-operator \cite{P11a}, a {\em strictly upper triangular matrix} $Z\in{\rm End}({\cal H}_n^{\otimes n})$ which satisfies a remarkable conservation law property
\begin{equation}
[H,Z] = -\sigma_1^\z + \sigma^\z_n.
\end{equation}
The operator whose time-derivative is composed of local operators at the chain boundaries has been termed as {\em almost conserved} \cite{IP13} and provides a useful tool to study the thermodynamic limit of quantum transport via the Lieb-Robinson bounds \cite{LR72}. Our operator $Z$ can be expressed in terms of a derivative of amplitude operator (\ref{eq:MPA1}) w.r.t. noise strength, or highest-weight transfer matrix w.r.t. representation parameter at $s=0$ \cite{PI13}
\begin{equation}
Z = -\ii \partial_\varepsilon \Omega_n|_{\varepsilon=0}= \frac{1}{2\eta\sin\eta} \partial_s \bra{0}\mm{L}^T\left(\frac{\pi}{2},s\right)^{\otimes n}\ket{0}|_{s=0}
\label{eq:Zs}
\end{equation}
and has a simple explicit MPA representation \cite{P11a}, related to (\ref{eq:AAA},\ref{eq:result1}) at $\varepsilon=0$ in an {\em extended} auxiliary space with a 
split-vacuum state
${\cal H}'_{\rm a} = {\rm lsp}\{ \ket{{\rm L}},\ket{{\rm R}},\ket{1},\ket{2}\ldots\}$
\begin{eqnarray}
&& Z = \sum_{\alpha_1,\ldots,\alpha_n\in\{ +,-,0\}} \bra{{\rm L}}\mm{A}'_{\alpha_1}\mm{A}'_{\alpha_2}\cdots \mm{A}'_{\alpha_n}\ket{{\rm R}} \sigma^{\alpha_1}\otimes \sigma^{\alpha_2} \otimes \cdots \sigma^{\alpha_n}, \label{eq:MPAZ} \\
&& \mm{A}'_0 = \ket{{\rm L}}\!\bra{{\rm L}} + \ket{{\rm R}}\!\bra{{\rm R}} + \sum_{k=1}^\infty \cos(k\eta) \ket{k}\!\bra{k},\nonumber\\
&& \mm{A}'_+ = \ket{{\rm L}}\!\bra{1} + \sum_{k=1}^\infty \sin(k\eta) \ket{k}\!\bra{k+1},\nonumber \\
&& \mm{A}'_- = \ket{1}\!\bra{{\rm R}} - \sum_{k=1}^\infty \sin((k+1)\eta) \ket{k+1}\!\bra{k}.\nonumber
\end{eqnarray}

To first order, up to ${\cal O}(\varepsilon^2)$, the following simple ansatz 
\begin{equation}
\rho^{(0)} = K(\chi)^{\otimes n},\quad \rho^{(1)} = \zeta \left(Z - Z^\dagger\right) K(\chi)^{\otimes n}
\end{equation}
solves Eqs.~(\ref{eq:pertp}), which results in a single condition 
\begin{equation}
\!\!\!\!\!\!\!\!\!\!\!\!\!\!\!\!\!\!\!\!\!\!\!\!\!\!\!\!\!\!\!\!\!\!\!\!
\zeta [H,Z-Z^\dagger]= -\gamma^+_{\rm L}K_1^{-1}\DD_{\sigma^+_1}K_1-\gamma^-_{\rm L}K^{-1}_1\DD_{\sigma^-_1}K_1 - \gamma^+_{\rm R}K_n^{-1}\DD_{\sigma^+_n}K_n-\gamma^-_{\rm R}K_n^{-1}\DD_{\sigma^-_n}K_n,\;
\end{equation}
imposing a vanishing linear combination of $\one, \sigma^\z_1,\sigma^\z_n$. Requiring the coefficients to vanish, results in a system of equations for
$\chi,\zeta$ with the unique solution:
\begin{equation}
\chi = \frac{\gamma^+_{\rm L}+\gamma^+_{\rm R}}{\gamma^-_{\rm L}+\gamma^-_{\rm R}},\quad
\zeta = \frac{1}{2} \frac{\gamma^+_{\rm L}+\gamma^+_{\rm R}+\gamma^-_{\rm L}+\gamma^-_{\rm R}}{(\gamma^+_{\rm L}+\gamma^+_{\rm R})(\gamma^-_{\rm L}+\gamma^-_{\rm R})}(\gamma^+_{\rm L}\gamma^-_{\rm R} - \gamma^-_{\rm L}\gamma^+_{\rm R}). \nonumber
\end{equation}
It is worth to remark that these leading order terms of NESS do not depend on coherent driving parameters $\mu_{\rm L/R}$.
Note that $(\chi-1)/(\chi+1)$ gives the net magnetization $\tr(\sigma^\z_x \rho_\infty)$ in NESS to leading order, while $\varepsilon \zeta$ is essentially the spin current $\tr[(\ii \sigma^+_x \sigma^-_{x+1} + {\rm h.c.})\rho_\infty]$ within the first order. 
In order to obtain the spatial modulation of the magnetization density profile, one needs to compute the second order $p=2$. With the tools at hand, this is only explicitly possible in the case
of symmetric incoherent driving 
\begin{equation}
\gamma^\pm_{\rm L} = \half (1\pm \mu),\quad\;
\gamma^\pm_{\rm R} = \half (1\mp \mu),\quad\; \mu_{\rm L/R} = 0,
\end{equation}
where the solution reads (for a simple proof see Ref.~\cite{P11a})
\begin{equation}
\rho^{(0)} = \one,\;\quad \rho^{(1)} = \mu (Z-Z^\dagger),\;\quad \rho^{(2)} = \frac{\mu^2}{2} (Z-Z^\dagger)^2 - \frac{\mu}{2}[Z,Z^\dagger].
\end{equation} 
In order to obtain a closed form expression for $\rho^{(2)}$ for general driving parameters one probably needs to include second derivative of
highest-weight transfer matrix w.r.t. $s$ at $s=0$ (extending (\ref{eq:Zs})) However, this has not been explicitly demonstrated yet.

One can also study asymptotics for a large coupling parameter $\varepsilon\to\infty$, in the so-called quantum Zeno regime, by writing a formal operator valued expansion of NESS in $1/\varepsilon$, 
$\rho_\infty = \sum_{p=0}^\infty \rho^{(p)} \varepsilon^{-p}$. In the case of fully anisotropic Heisenberg spin 1/2 chain ($XYZ$ model)  a remarkable effect has been demonstrated \cite{PSL15}, namely
engineering a transitions from equilibrium-like (uncorrelated) to genuine nonequilibrium (strongly correlated) steady state by applying local magnetic fields to spins near the boundary (at sites $x=2$ and $x=n-1$).  It is possible also to derive explicit asymptotic expansions in some other (large) parameters of the model, say in (spatially) modulated external magnetic field or anisotropy $\Delta$. Even in the generic, non-integrable situation of $XXZ$ spin $1/2$ chain with spatially modulated interactions an explicit asymptotic expression for the spin current has been derived  \cite{LP15}, exhibiting a strong rectification effect upon switching the direction of coherent driving in the presence of incoherent driving (see also Ref.~\cite{LOK15}).

\section{Nonequilibrium partition function and computation of observables}
\label{obs}

Here we shall elaborate on computation of physical observables using the standard `transfer-matrix' technique. For concrete examples, we work out spin-density profiles, spin currents, and two point spin-spin correlations. For most of our discussion we allow the NESS to be of the most general non-perturbative form (\ref{eq:Ransatz}), or equivalently (\ref{eq:nessXXZ1}), with parameters (\ref{eq:varphi}-\ref{eq:parXXZ1}), as driven by a general combination of asymmetric 
coherent and incoherent boundary couplings (\ref{eq:XXZbf},\ref{eq:asymL}). Only for some very specific explicit calculations for the $XXX$ chain discussed in the second part of this section we shall assume purely incoherent and left-right symmetric driving.
In parts this section follows Refs.~\cite{P11b,PKS13}, while some essential results are new and presented here for the first time.

\subsection{The case of generic $XXZ$ chain}

Let us define the nonequilibrium partition function as the trace of the un-normalized density operator of NESS (expressed as (\ref{eq:Ransatz}))
\begin{equation}
{\cal Z}_n = \tr R_\infty = \bra{0,0}\vmbb{T}^{n}\ket{0,0},
\label{eq:Z}
\end{equation}
where $\vmbb{T}=\vmbb{T}(\varphi,s,\chi) \in {\rm End}({\cal H}_{\rm a}\otimes{\cal H}_{\rm b})$ is the transfer-matrix
\begin{equation}
\vmbb{T}=\tr_{\rm p}(\vmbb{L} K)=2\vmbb{L}^0 \cosh\kappa +  2\vmbb{L}^\z \sinh\kappa,
\end{equation}
writing the magnetisation parameter $\chi$ in terms of another parameter $\kappa$ as
\begin{equation}
 \chi = e^{2\kappa}.
\end{equation}
Any physical observable $A\in{\rm End}({\cal H}_{\rm p}^{\otimes n})$ can be expressed in terms of a linear combination 
$A=\sum_{\ul{\alpha}} a_{\ul{\alpha}} O_{\ul{\alpha}}$ of tensor products $O_{\ul{\alpha}} = \sigma^{\alpha_1}\otimes\sigma^{\alpha_2}\otimes\cdots \sigma^{\alpha_n}$, $\alpha_x\in {\cal J}$.
For each product operator $O_{\ul{\alpha}}$ its NESS expectation value can again be conveniently expressed in terms of a matrix product
\begin{eqnarray}
&&\ave{O_{\ul{\alpha}}} = \tr(O_{\ul{\alpha}}\rho_\infty)= \frac{\tr(O_{\ul{\alpha}} R_\infty )}{\tr R_\infty} = {\cal Z}^{-1}_n \bra{0,0}\vmbb{V}^{\alpha_1}\vmbb{V}^{\alpha_2}\cdots\vmbb{V}^{\alpha_n}\ket{0,0}, \label{eq:O}
\end{eqnarray}
where $\vmbb{V}^{\alpha} = \tr_{\rm p}(\vmbb{L} K \sigma^\alpha)$,
or, explicitly: 
\begin{equation}
\vmbb{V}^\pm=e^{\pm\kappa} \vmbb{L}^{\mp},\qquad \vmbb{V}^0 = \vmbb{T},\qquad\vmbb{V}^\z=2\vmbb{L}^\z \cosh\kappa + 2\vmbb{L}^0\sinh\kappa.
\end{equation}
The operators $\vmbb{V}^\alpha$ shall sometimes be referred to as {\em vertex operators}.
Since the generators $\mm{s},\mm{s}^\pm,\mm{t},\mm{t}^\pm$ expressing the physical components of the double Lax operator (\ref{eq:LL2}) are tridiagonal matrices, 
one needs to consider, for a chain of $n$ sites, 
only auxiliary basis states $\ket{k}$ up to $k\le n/2$, and consequently the above expression (\ref{eq:O}) can be evaluated efficiently
within ${\cal O}(n^3)$ computer operations even without any further insight. 
We will anyway show bellow that expectation values of many observables in many situations can be evaluated fully analytically.

We shall particularly focus on three kinds of observables in NESS. The simplest and perhaps the most important one is the spin current 
\begin{equation}
j_{x,x+1} = \ii (\sigma^+_x \sigma^-_{x+1} - \sigma^-_x \sigma^+_{x+1})
\label{eq:j}
\end{equation}
which satisfies the continuity equation, or local conservation law for spin density $\sigma^\z_x$
\begin{equation}
\frac{\dd}{\dd t} \sigma^\z_x = \ii [H,\sigma^\z_x] = -4j_{x,x+1} + 4j_{x-1,x}.
\end{equation}
We shall use (\ref{eq:j}) as the spin-current operator in the following, although we note a trivial factor of $4$ which needs to be taken into account when comparing to physical units.
The above identity also implies that the expectation value of the current should be site independent in the steady state $\ave{j_{x,x+1}} = \ave{j_{1,2}} =: J$.
Similarly as in the case of ASEP, the current in NESS can be computed solely in terms of nonequilibrium partition function ${\cal Z}_n$.
To see that, we express the asymmetry parameter $\chi$ from the defining equations (\ref{eq:asymb}-\ref{eq:asyme}) as a simple function of four free complex variables $\varphi,\vartheta,s,t$
\begin{equation}
\chi(\varphi,\vartheta,s,t) = \sqrt{\frac{\sin(\vartheta-\varphi+\eta s - \eta t)}{\sin(\varphi-\vartheta+\eta s - \eta t)}}.
\label{eq:chi2}
\end{equation}
We note that arbitrary branch of the square-root can be chosen as it only affects the sign of the transfer matrix $\vmbb{T}$, whose explicit form we read from (\ref{eq:LL2}):
\begin{eqnarray}
\vmbb{T} &=& (\sin^2\eta) \chi^{1/2} \mm{s}^+\mm{t}^+ + (\sin^2\eta) \chi^{-1/2} \mm{s}^-\mm{t}^- \nonumber\\
&+& \chi^{1/2} \sin(\varphi-\eta \mm{s})\sin(\vartheta-\eta\mm{t}) + \chi^{-1/2} \sin(\varphi+\eta\mm{s})\sin(\vartheta+\eta\mm{t}),
\end{eqnarray}
or the sign of the un-normalized density operator $R_\infty$ for odd $n$, but not the observables (\ref{eq:O}) themselves.
One notes that a similar expression is obtained for the commutator of the off-diagonal elements of double-Lax operator
$(\csc\eta)[\vmbb{L}^-,\vmbb{L}^+] = 
(\csc\eta)[\vmbb{V}^+,\vmbb{V}^-]= (\sin^2\eta)\sin(\varphi-\vartheta+\eta\mm{s}-\eta\mm{t})\mm{s}^+\mm{t}^+ +
 (\sin^2\eta)\sin(\vartheta-\varphi+\eta\mm{s}-\eta\mm{t})\mm{s}^-\mm{t}^-
- \sin(2\eta\mm{s})\sin(\vartheta+\eta\mm{t})\sin(\vartheta-\eta\mm{t})
+ \sin(2\eta\mm{t})\sin(\varphi+\eta\mm{s})\sin(\varphi-\eta\mm{s})$.

Furthermore, let us identify a particularly important, {\em diagonal} subspace of the product auxiliary space ${\cal K} = {\rm lsp}\{ \ket{k,k},k=0,1,2\ldots\} \subset {\cal H}_\aa\otimes{\cal H}_\bbb$, where a compact Dirac notation $\ket{k,l}\equiv \ket{k}\otimes\ket{l}$ is used.
We note that any operator valued function of $f(\mm{s}-\mm{t})$ on ${\cal K}$ evaluates as $f(s-t)$. Henceforth, elementary trigonometry results in the following very useful relation for the orthogonal projection on ${\cal K}$
\begin{equation}
\!\!\!\!\!\!\!\!\!\!\!\!
\vmbb{P}([\vmbb{V}^+,\vmbb{V}^-] - \ii\zeta \vmbb{T}) = ([\vmbb{V}^+,\vmbb{V}^-] - \ii\zeta \vmbb{T})\vmbb{P} = 0,\quad \vmbb{P}:=\sum_{k=0}^\infty\ket{k,k}\!\bra{k,k},\quad
\end{equation}
where
\begin{equation}
\!\!\!\!\!\!\!\!\!\!\!\!\zeta(\varphi,\vartheta,s,t) = (\sin\eta)\sqrt{\sin(\varphi-\vartheta-\eta s + \eta t)\sin(\varphi-\vartheta+\eta s - \eta t)}.
\label{eq:zeta}
\end{equation}
Now using the definitions (\ref{eq:j},\ref{eq:O}), together with the facts $[\vmbb{T},\vmbb{P}]=0$, $\vmbb{P}\ket{0,0}=\ket{0,0}$, one shortly arrives at a compact expression for the steady-state spin current
\begin{equation}
J = \zeta \frac{{\cal Z}_{n-1}}{{\cal Z}_n},
\label{eq:current}
\end{equation}
which is similar to the expression of a particle current in the classical simple exclusion processes \cite{BE07}.
Note that the parameter $\zeta$ can be conveniently expressed also in terms of physical driving parameters, via (\ref{eq:asymb}-\ref{eq:asymc}), namely
\begin{equation}
\!\!\!\!\!\!\!\!\!\!\!\!\!\!\zeta = \frac{(\Gamma_{\rm L}\Gamma_{\rm R})^{1/2} \sin^2\eta}{\left((\quart\Gamma_{\rm L}^2-b_{\rm L}^2 -\sin^2\eta)^2 + \Gamma^2_{\rm L} b^2_{\rm L}\right)^{1/4}\left((\quart\Gamma_{\rm R}^2-b_{\rm R}^2 - \sin^2 \eta)^2+ \Gamma^2_{\rm R} b^2_{\rm R}\right)^{1/4}}.
\end{equation}
It the simplest case of symmetric incoherent driving, $\Gamma_{\rm L}=\Gamma_{\rm R} =\varepsilon$, $b_{\rm L}=b_{\rm R}=0$, it amounts to $\zeta = \varepsilon\sin^2\eta/|\quart\varepsilon^2-\sin^2\eta|$.

Other simple and interesting physical observables that we consider in some detail are the spin-density and connected transverse spin-spin correlation function of NESS
\begin{eqnarray}
\!\!\!\!\!\!\!\!\!\!{\cal M}_x &=& \ave{\sigma^\z_x} = {\cal Z}^{-1}_n \bra{0,0}\vmbb{T}^{x-1}\vmbb{V}^\z\vmbb{T}^{n-x}\ket{0,0}, \label{eq:defMx} \\
\!\!\!\!\!\!\!\!\!\!{\cal C}_{x,y} &=& \ave{\sigma^\z_x \sigma^\z_y} - \ave{\sigma^\z_x}\ave{\sigma^\z_y} = {\cal Z}^{-1}_n \bra{0,0}\vmbb{T}^{x-1}\vmbb{V}^\z\vmbb{T}^{y-x-1}\vmbb{V}^\z\vmbb{T}^{n-y}\ket{0,0} - {\cal M}_x{\cal M}_y,\nonumber
\end{eqnarray}
where in the last line we have assumed, without loss of generality, that $x < y$. Since $\vmbb{V}^\z$ also conserves the spin-difference $\mm{s}-\mm{t}$ and hence leaves the diagonal subspace ${\cal K}$ invariant i.e., 
$[\vmbb{V}^\z,\vmbb{P}]=0$, one can identify the diagonal subspace with the auxiliary subspace ${\cal K} \leftrightarrow {\cal H}_\aa$ via the mapping $\ket{k,k}\leftrightarrow \ket{k}$, and write
the diagonally projected transfer matrices $\mm{T} := \vmbb{T}|_{\cal K}$, $\mm{V} := \vmbb{V}^\z|_{\cal K}$ as
\begin{eqnarray}
\mm{T} &=& \mm{V}^{00} + \mm{V}^{11},\quad \mm{V} = \mm{V}^{00} - \mm{V}^{11}, \label{eq:TVdef} \\
\mm{V}^{00} &=& \chi^{-1} \sum_{k=0}^\infty \Bigl(|\sin((k+1)\eta)|^2 \ket{k}\!\bra{k+1} + |\sin(\varphi-\eta(s-k))|^2 \ket{k}\!\bra{k}\Bigr),\nonumber \\
\mm{V}^{11} &=& \chi \sum_{k=0}^\infty \Bigl(  |\sin((2s-k)\eta)|^2 \ket{k+1}\!\bra{k} + |\sin(\varphi+\eta(s-k))|^2 \ket{k}\!\bra{k}\Bigr).\nonumber 
\end{eqnarray}
We explicitly used the complex conjugation property $\eta t=\overline{\eta s}$, $\vartheta=\bar{\varphi}$, which one has for physical values of the driving parameters (\ref{eq:varphi}).
In terms of the projected transfer matrices, the nonequilibrium partition function and the transverse spin observables read
\begin{eqnarray}
{\cal Z}_n &=& \bra{0}\mm{T}^n\ket{0}, \label{eq:Zn} \\ 
{\cal M}_x &=& {\cal Z}^{-1}_n \bra{0}\mm{T}^{x-1} \mm{V}\mm{T}^{n-x}\ket{0},  \\
{\cal C}_{x,y} &=& {\cal Z}^{-1}_n \bra{0}\mm{T}^{x-1}\mm{V}\mm{T}^{y-x-1}\mm{V}\mm{T}^{n-y}\ket{0} - {\cal M}_x{\cal M}_y.
\end{eqnarray}

For the massless $XXZ$ model, $|\Delta| < 1$, one can always approximate $\eta=\arccos\Delta \in\RaR$ to an arbitrary accuracy, for fixed $n$,
with a rational $\eta=\pi l/m$, with $l,m \in \ZZ$, $m> 0$. This corresponds to $q=e^{\ii\eta}$ being a (generically non-primitive) $2m-$th root of unity.
In such a case, the transfer matrix $\mm{T}$ can be truncated to an $m$-dimensional sub-space, ${\cal H}'_\aa = {\rm lsp}\{ \ket{k}, k = 0,1\ldots,m-1\}$, since the transition between states $\ket{m}$ and $\ket{m-1}$ is forbidden (see the first summation term of $\mm{V}^{00}$ in (\ref{eq:TVdef})). Hence the transfer and vertex matrices $\mm{T}$, $\mm{V}$, can be replaced, respectively, by $m\times m$ matrices, 
$\mm{T}' = \mm{T}|_{{\cal H}'_\aa}$, $\mm{V}' = \mm{V}|_{{\cal H}'_\aa}$. Specifically, $\mm{T}^x\ket{0} = (\mm{T}')^x\ket{0}$, $\forall x$, so in TL, $n\to\infty$,
observables are essentially given by eigenvalue decomposition of $\mm{T}'=\mm{U} {\,\rm diag}(\tau_1,\tau_2,\ldots \tau_m)\mm{U}^{-1}$, with eigenvalues ordered as $|\tau_1| \ge |\tau_2|\ge\ldots |\tau_m|$, so the steady-state current is ballistic (i.e., $n$-independent) 
\begin{equation}
J = \frac{\zeta}{\tau_1}.
\end{equation}
Similarly, using the fact that in the eigenbasis of $\mm{T}'$, the transformed vertex operator 
$\mm{U}^{-1}\mm{V}'\mm{U}$ have vanishing diagonal elements
\footnote{This follows from a simple
observation that there exist a gauge transformation $\ket{k}\to \phi_k \ket{k}$, $\bra{k}\to (\phi_k)^{-1} \bra{k}$ for appropriate weights $\phi_k\neq 0$ such that $\mm{T}'$ becomes a symmetric and $\mm{V}'$ an anti-symmetric matrix, $\mm{T}'^T=\mm{T}', \mm{V}'^T=-\mm{V}'$.},
i.e. $\bra{\psi'_j}\mm{V}'\ket{\psi_j}=0$ for $\mm{T}'\ket{\psi_j} = \tau_j \ket{\psi_j}$, $\bra{\psi'_j}\mm{T}' = \tau_j\bra{\psi'_j}$,  eigenvalue decomposition gives thermodynamically vanishing spin density
with exponentially damped profiles near the ends (Fig.~\ref{fig:observ}a),  namely 
\begin{equation} 
\!\!\!\!\!\!\!\!\!\!\!\!\!\!\!\!\!\!\!\!\!\!\!\!\!\!\!\!\!\!{\cal M}_x = \sum_{j=2}^m c_j\left\{ \left(\frac{\tau_j}{\tau_1}\right) ^{x-1}\!\!- \left(\frac{\tau_j}{\tau_1}\right)^{n-x}\right\},\;\;{\rm with}\; c_j =\braket{0}{\psi_j}\bra{\psi'_j}\mm{V}'\ket{\psi_1}\braket{\psi'_1}{0}.
\label{eq:masslessMx}
\end{equation} 
Similarly, it can be shown that far away from the edges, $1\ll x,y \ll n$, spin-spin correlations decay exponentially 
\begin{equation}
\!\!\!\!\!\!\!\!\!\!\!\!\!\!\!\!\!\!\!\!\!\!\!{\cal C}_{x,y} \approx \sum_{j=2}^m c'_j \left(\frac{\tau_j}{\tau_1}\right)^{|x-y|},\;\;{\rm with}\; c'_j=\braket{0}{\psi_1}\bra{\psi'_1}\mm{V}'\ket{\psi_j}\bra{\psi'_j}\mm{V}'\ket{\psi_1}\braket{\psi'_1}{0}.
\label{eq:masslessCxy}
\end{equation}  
Note that such exponential decay of correlations in the steady state in this regime is qualitatively reminiscent of equilibrium behaviour at a finite temperature $T>0$, namely long-range order is absent, even though $\rho_\infty$ is highly non-thermal (non-Gibbsian).
Let us now work out an explicit example
for $\Delta=1/2=\cos(\pi/3)$, $m=3$, and for symmetric incoherent driving $\Gamma_{\rm L}=\Gamma_{\rm R} =\varepsilon$, $b_{\rm L}=b_{\rm R}=0$, implying $\varphi=0$ and $\tan \eta s = 2\ii\sin\eta/\varepsilon=\ii\sqrt{3}/\varepsilon$.
Up to trivial similarity (gauge) transformation $\ket{k} \to \left(\frac{1}{8}|\varepsilon^2-3|\right)^{-k} \ket{k}$, $\bra{k} \to \left(\frac{1}{8}|\varepsilon^2-3|\right)^{k} \bra{k}$, the transfer and vertex matrices read
\begin{eqnarray}
\mm{T}' &=& \frac{2\zeta}{\varepsilon} 
\pmatrix{1 & \quart\varepsilon^2 & 0 \cr
1 & \quart(\varepsilon^2+1) & \frac{1}{64}(1+\varepsilon^2)(9+\varepsilon^2) \cr
0 & 1 & \quart (1+\varepsilon^2)
},\\
\mm{V}' &=& \frac{2\zeta}{\varepsilon} 
\pmatrix{0 & \quart\varepsilon^2 & 0 \cr
-1 & 0 & \frac{1}{64}(1+\varepsilon^2)(9+\varepsilon^2) \cr
0 & -1 & 0
},\
\end{eqnarray}
The eigenvalues of $\mm{T}'$ are $\tau_{1,3}= \frac{\zeta}{8\varepsilon}(7 + 3 \varepsilon^2 \pm\sqrt{ 81 + 74 \varepsilon^2 + 9 \varepsilon^4})$, $\tau_2 = \frac{\zeta}{4}(5+\varepsilon^2)$, yielding the
spin current
$J = 8\varepsilon/(7 + 3 \varepsilon^2 + \sqrt{ 81 + 74 \varepsilon^2 + 9 \varepsilon^4})$ (see Fig.~\ref{fig:observ}b). 
Spin-profiles (see Fig.~\ref{fig:observ}a) and spin-spin correlation are described by Eqs.~(\ref{eq:masslessMx},\ref{eq:masslessCxy}) with explicit, but lengthy expressions for the coefficients $c_{2,3}(\varepsilon),c'_{2,3}(\varepsilon)$.

On the other hand, for the massive $XXZ$ model, $|\Delta| > 1$, parameter $\eta$ is complex, namely $\eta = \ii \eta'$, with $\eta'={\rm arcosh}\,|\Delta|$ for $\Delta > 1$ and $\eta'={\rm arcosh}\,|\Delta| + \ii\pi$ for $\Delta < -1$. 
In such a case, the tridiagonal transfer matrix $\mm{T}$ (\ref{eq:TVdef}) is of genuinely infinite dimension, with exponentially growing (in state index $k$) transition amplitudes. In WGS picture [see Eq.~(\ref{eq:WGSXXZ})] the walk that gives a dominating
contribution to the partition sum ${\cal Z}_n$, for large even $n$, is composed of $n/2$ steps forward $(0,1),(1,2),\ldots,(n/2-1,n/2)$ followed by $n/2$ steps backward $(n/2,n/2-1),\ldots,(1,0)$,
\begin{equation}
{\cal Z}_n \simeq \prod_{k=1}^{n/2} |\sinh((2s-k+1)\eta') \sinh(k\eta')|^2.
\end{equation}
We note that the contribution from this extremal walk relatively overweights the sum of all other contribitions for large $n$, as it grows faster than any exponential in $n$, so it is super-exponentially larger than contributions of exponentially many ($\le 3^n$) typical terms. The spin-current, for large $n$, can then be computed by applying (\ref{eq:current}) twice, namely $(\zeta/J)^2 = {\cal Z}_n/{\cal Z}_{n-2} = |\sinh((2s-n/2+1)\eta')|^2 |\sinh(n\eta'/2)|^2$, yielding asymptotically $J \simeq \zeta/|\sinh(n\eta'/2)|^2 \simeq \zeta |e^{n\eta'}|$, or (see Fig.~\ref{fig:observ}b for comparison with exact numerical results from transfer matrix computation)
\begin{equation}
J \simeq \frac{\zeta}{(|\Delta| + \sqrt{\Delta^2-1})^n}.
\end{equation}
Similarly, one can compute the spin density ${\cal M}_x$, obtaining a kink profile from the {\em dominating} walk, namely first $n/2$ spins pointing up and last $n/2$ spins pointing down, ${\cal M}_{x} = {\rm sign}((n+1)/2-x)$,
where ${\rm sign}(x)=-1,0,1$ for $x <,=,> 0$, respectively, while connected correlations thermodynamically vanish ${\cal C}_{x,y}=0$. 
In fact, from the dominating walk one obtains the leading asymptotics for the entire NESS density operator 
\begin{eqnarray}
\rho_{\infty} &\simeq& (\sigma^+_1 \sigma^-_1) \cdots (\sigma^+_{n/2}\sigma^-_{n/2})(\sigma^-_{n/2+1}\sigma^+_{n/2+1})\cdots(\sigma^-_{n}\sigma^+_{n}) \nonumber \\
&=& \ket{00\ldots 011\ldots 1}\bra{00\ldots 011\ldots 1}.
\end{eqnarray}
For finite $n$, the kink spin density profile attains a finite width (see e.g., Fig.~\ref{fig:observ}a), which can be quantified to be of order $\log n$ \cite{BCPRZ}.

\subsection{Isotropic ($XXX$) chain -- nonequilibrium partition function.}
\label{subsect:neZ}

At the end, let us turn to perhaps the most interesting case of SU(2) symmetric $XXX$ or isotropic Heisenberg chain. Here the double Lax operator have to be defined with respect to a scaling limit
$\vmbb{L}(\lambda,\mu,s,t) \longleftarrow
\lim_{\eta\to\infty} \eta^{-2} \vmbb{L}(\eta\lambda,\eta\mu,s,t)$, reading
\begin{equation}
\vmbb{L} = \pmatrix{
 (\lambda - \mm{s})(\mu - \mm{t}) + \mm{s}^+\mm{t}^+ &
(\lambda-\mm{s})\mm{t}^-+(\mu+\mm{t})\mm{s}^+ \cr
(\lambda + \mm{s})\mm{t}^+ + (\mu-\mm{t})\mm{s}^- &
(\lambda + \mm{s})(\mu + \mm{t}) + \mm{s}^-\mm{t}^-
},
\end{equation}
while the projected transfer operator $\mm{T}(\lambda,s) \longleftarrow \lim_{\eta\to\infty} \eta^{-2} \mm{T}(\eta\lambda,s)$
 is again manifestly infinite-dimensional, but with amplitudes growing only quadratically with the state index $k$
\begin{eqnarray}
\mm{T} &=& \sum_{k=0}^\infty \Bigl(  \chi^{1/2} (k+1)^2 \ket{k+1}\!\bra{k} + \chi^{-1/2} |k-2s|^2 \ket{k}\!\bra{k+1} + \nonumber\\
&&+ \left(\chi^{1/2}|k-s+\lambda|^2+\chi^{-1/2} |k-s-\lambda|^2\right)\ket{k}\!\bra{k}\Bigr).
\end{eqnarray}
We shall now present a simple scaling argument which allows to analytically compute large $n$ asymptotics of the nonequilibrium partition function ${\cal Z}_n$ and consequently the spin current, for arbitrary driving parameters,
deriving the result announced already in Ref.~\cite{P11b}.
Let us define a tridiagonal operator $T$ on the space $\ell_\infty$ of sequences of coefficients $\ul{\psi}=(\psi_0,\psi_1,\psi_2\ldots)$, with $ \sum_{k=0}^\infty (T\ul{\psi})_k \ket{k}
= \mm{T}\sum_{k=0}^\infty \psi_k\ket{k} $, namely
\begin{equation}
\!\!\!\!\!\!\!\!\!\!\!\!\!\!\!\!\!\!\!\!\!\!\!\!\!\!\!\!(T\ul{\psi})_k =  \chi^{1/2} k^2 \psi_{k-1} + \chi^{-1/2} |k-2s|^2 \psi_{k+1} + \bigl(\chi^{1/2} |k-s+\lambda|^2 + \chi^{-1/2}|k-s-\lambda|^2\bigr)\psi_k.\quad
\end{equation}
The partition function (\ref{eq:Zn}) can be written as ${\cal Z}_n = \psi^{(n)}_0$ where $\ul{\psi}^{(n)} = T^n (1,0,0,\ldots)$. We shall however compute large $n$ asymptotics of the entire sequence $\psi^{(n)}_k$.
Since $T$ is a tridiagonal operator the vector $\psi^{(n)}$ is supported on exactly $n+1$ sites, i.e. $\psi^{(n)}_k = 0$ for all $k > n$.
We thus propose the following scaling ansatz
\begin{equation}
\psi^{(n)}_k \simeq F_n \exp\left(n f(k/n)\right),
\end{equation}
where $F_n$ is a sequence of real numbers and $f(\xi)$ some smooth (differentiable) function on $\xi\in[0,1]$.
The consistency of the ansatz is demonstrated,  and difference and differential equations for $F_n$ and $f(\xi)$ are, respectively, derived, from expanding both-sides of local scaling relation
$\psi^{(n+1)}_k =F_{n+1} e^{(n+1) f\left(k/(n+1)\right)} = (T \ul{\psi}^{(n)})_k$ in $1/n$, namely
\begin{equation}
\!\!\!\!\!\!\!\!\!\!\!\!\!\!\!\!\!\!\!\!\!\!\!\!\!\!\!\!F_{n+1} e^{(n+1) f(k/n)-(k/n)f'(k/n)}  \simeq F_n  k^2 e^{n f(k/n)} (e^{f'(k/n)-\kappa}+e^{-f'(k/n)+\kappa}+e^{\kappa}+e^{-\kappa}).
\label{eq:LHSRHS}
\end{equation}
The dependence on parameters $\lambda$ and $s$ can be neglected to leading orders in $1/n$, namely they yield smaller or comparable correction than neglecting the second derivatives due to shifts $k\to k\pm 1$ or scaling $k/n\to k/(n+1)$ in the exponentials
on the RHS or LHS of (\ref{eq:LHSRHS}), respectively. Introducing a scaling variable $\xi$, via $k=n \xi$ and dividing by $F_n e^{n f(\xi)}$ we finally obtain
\begin{equation}
\frac{F_{n+1}}{n^2 F_n} \exp(f(\xi) - \xi f'(\xi)) = 2\xi^2 (\cosh(f'(\xi)-\kappa) + \cosh\kappa).
\label{eq:fdif}
\end{equation}
As RHS does not depend on $n$, neither must the LHS, i.e., $F_{n+1}/(n^2 F_n) = C$, while without loss of generality one may fix $C=1$ by suitably adjusting $f(\xi)$ by an additive constant.
Hence we arrive to
\begin{equation}
F_n = F_1 (n!)^2,
\end{equation}
and a curiously-looking {\em implicit} differential equation for $g(\xi):=f(\xi)-\kappa \xi$
\begin{equation}
g (\xi) - \xi g'(\xi) = 2\log \xi  + \log(2 \cosh g'(\xi) + 2\cosh \kappa).
\label{eq:gxi}
\end{equation}
Using a substitution for a new {\em independent} variable
\begin{equation}
t = -\dd g(\xi)/\dd\xi
\label{eq:t}
\end{equation}
and writing the parametric dependences as $g_t$ and $\xi_t$, the equation (\ref{eq:gxi}) transforms to
\begin{equation}
g_t = 2\log \xi_t - t \xi_t + \log(2\cosh t + 2\cosh \kappa).
\label{eq:g}
\end{equation}
An {\em explicit} differential equation for $\xi_t$ is obtained by differentiating with respect to $t$ and using (\ref{eq:t})
\begin{equation}
\frac{2}{\xi_t}\frac{\dd \xi_t}{\dd t} = \xi_t - \frac{\sinh t}{\cosh t + \cosh \kappa}.
\end{equation}
This equation can be linearised by substitution $y(t) = 1/\xi_t$ and solved explicitly in terms of elliptic integral of the first kind $F(\phi,k) = \int_0^\phi\dd \theta (1-k^2\sin^2\theta)^{-1/2}$.
Fixing the integration constant by incorporating the boundary condition
\begin{equation}
\xi_{t\to-\infty} = 0,\quad \xi_{t\to\infty} = 1,
\end{equation}
which expresses the obvious fact that the scaling variable $\xi$ has to span the entire interval $[0,1]$,
one obtains an explicit result
\begin{equation}
\xi_t = \sqrt{\frac{1+\cosh\kappa}{\cosh t + \cosh\kappa}}\left(K\left(\tanh^2\frac{\kappa}{2}\right) + \ii F\left(\frac{\ii t}{2}, {\rm sech}^2 \frac{\kappa}{2}\right)\right)^{-1},
\label{eq:xtgt}
\end{equation}
where $K(k) = F(\pi/2,k)$ is the complete elliptic integral of the first kind, which together with Eq.~(\ref{eq:g}) yields the complete scaling profile.
From (\ref{eq:xtgt}) we obtain the key information 
\begin{equation}
f(0) = \lim_{t\to-\infty} f_t = 2 \log\frac{\cosh(\kappa/2)}{K(\tanh^2(\kappa/2))},
\end{equation}
which yields the asymptotics of ${\cal Z}_n$, and consequently, the spin-current (\ref{eq:current})
\begin{equation}
\!\!\!\!\!\!\!\!\!\!\!\!\!\!\!\!\!\!\!\!\!\!\!\!\!\!\!\!
{\cal Z}_n \simeq F_n e^{n f(0)} = F_1 (n!)^2 \left(\frac{\cosh(\kappa/2)}{K(\tanh^2(\kappa/2))}\right)^{2n}\!\!,\;\; J = \left(\frac{\cosh(\kappa/2)}{K(\tanh^2(\kappa/2))}\right)^{2}\frac{\zeta}{n^2}, \label{eq:Zsc}
\end{equation}
where the $XXX$-scaled current parameter $\zeta$ (\ref{eq:zeta}) reads
\begin{equation}
\zeta \longleftarrow \lim_{\eta\to 0}\eta^{-2}\zeta(\eta\lambda,\eta\mu,s,t) = \sqrt{(\lambda-\mu-s+t)(\lambda-\mu+s-t)},
\end{equation} 
or in terms of driving parameters
\begin{equation}
\zeta=\frac{\sqrt{\Gamma_{\rm L}\Gamma_{\rm R}}}{\sqrt{(\quart\Gamma_{\rm L}^2+b_{\rm L}^2 )(\quart\Gamma_{\rm R}^2+b_{\rm R}^2)}}.
\end{equation}
In the special case of symmetric driving $\chi=1$, $\kappa=0$, the scaling profile simplifies (noting that $K(0)=\pi/2$)
\begin{equation}
\xi_t = \frac{2{\rm sech}(t/2)}{\pi-4\arctan\tanh(t/4)},\quad f_t \equiv g_t = 2\log\left(2\xi_t\cosh\frac{t}{2}\right)- t \xi_t,
\label{eq:xtgt0}
\end{equation}
and $f(0) = 2\log\frac{2}{\pi}$, yielding the spin current \cite{P11b}
\begin{equation}
J = \frac{\pi \zeta}{4 n^2}.
\end{equation}

One finds an excellent agreement of the whole scaling profile 
\begin{equation}
\bra{k}\mm{T}^n\ket{0} \simeq F_1 (n!)^2 \exp\left(n f(k/n) \right)
\label{eq:profile}
\end{equation}
with numerical iteration of the transfer operator $\mm{T}$ even for relatively small $n$ ($n \sim 100$ for $\Gamma_{\rm L,R},b_{\rm L,R}\sim 1$), where the undeterminable constant $F_1$ quickly becomes irrelevant due to the super-exponential growth (see Fig.~\ref{fig:fx}). One can repeat our scaling analysis for the transpose of the transfer operator to show the same leading
order in $1/n$ asymptotics
\begin{equation}
\bra{k}\mm{T}^n\ket{0} \simeq \bra{0}\mm{T}^n\ket{k}.
\label{eq:profsym}
\end{equation}

\begin{figure}
\hbox{\hspace{1.4cm} \includegraphics[scale=1.63]{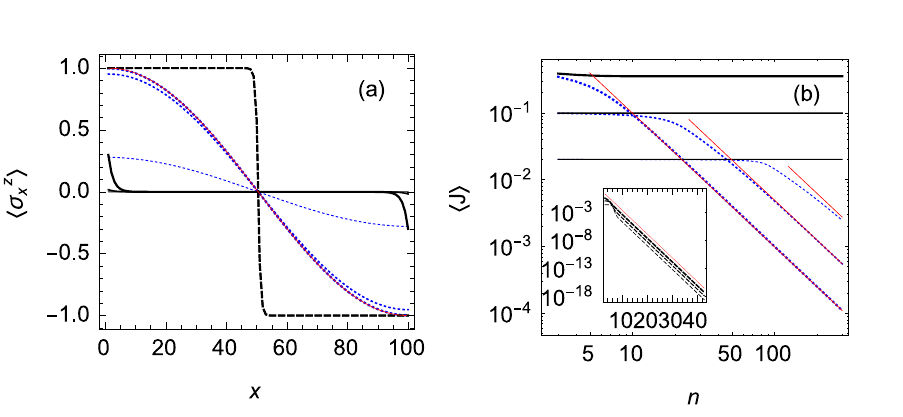}
}	
	\caption{ 
	(from \cite{P11b})
	Spin profiles ${\cal M}_x=\ave{\sigma^{\rm z}_x}$ at $n=100$ (a), and spin currents $J$ vs. size $n$ (b), of boundary driven $XXZ$ spin 1/2 chain for $\Delta=3/2$ (dashed), $\Delta=1$ (dotted/blue), 
	$\Delta=1/2$ (full curves), all for three different incoherent spin source/sink rates $\Gamma_{\rm L}=\Gamma_{\rm R}=\varepsilon=1,1/5, 1/25$ using thick, medium, thin curves, respectively. 
	Red full curves show closed-form asymptotic results  [see text]:
	${\cal M}_x=\cos \pi \frac{x-1}{n-1}$, $J =\pi^2 \varepsilon^{-1} n^{-2}$ for $\Delta=1$ in the main panels (a,b),
	and  $J \propto e^{-n \,{\rm arcosh}\Delta}$ in (b)-inset.}
 	\label{fig:observ}
\end{figure}

Note that the asymptotics of the partition function (\ref{eq:Zsc}) is unique and depends only on the asymmetry parameter $\kappa$ and not on the spectral and the representation parameters, $\lambda,s$ separately. In fact, an analogous asymptotics should be obtained, following essentially the same derivation, when starting from an arbitrary local state $\ket{l}$, $l\in\{0,1,\ldots,k\}$, hence giving the scaling of an arbitrary matrix element of $\mm{T}^n$, namely
\begin{equation}
\bra{k}\mm{T}^{n-l}\ket{l} \simeq F \frac{(n!)^2}{(l!)^2} \exp\left(n f(k/n)\right),\quad l \le k,
\label{eq:profile1}
\end{equation}
where $F$ is a constant independent of $k,l,n$. Using expressions (\ref{eq:profile},\ref{eq:profsym},\ref{eq:profile1}) one can control the asymptotic $n\to\infty$ behaviour
of any $\bra{\psi_{\rm L}}\mm{T}^n\ket{\psi_{\rm R}}$ if at least one of the auxilliary states $\ket{\psi_{\rm L/R}}\in{\cal K}$ has a {\em finite support}, i.e., in can be expanded in 
finitely many basis states $\ket{k}$.

As we shall see later, the same universal scaling of nonequilibrium partition function and the corresponding canonical current ${\cal Z}_{n-1}/{\cal Z}_n \propto J\propto n^{-2}$ applies to several other models with intrinsic (and undeformed) $SU(2)$ symmetry, such as the nonequilibrium Hubbard model or even spin-$1$ Lai-Sutherland chain.

\begin{figure}[!t]
\begin{center}
\includegraphics[scale=1.27]{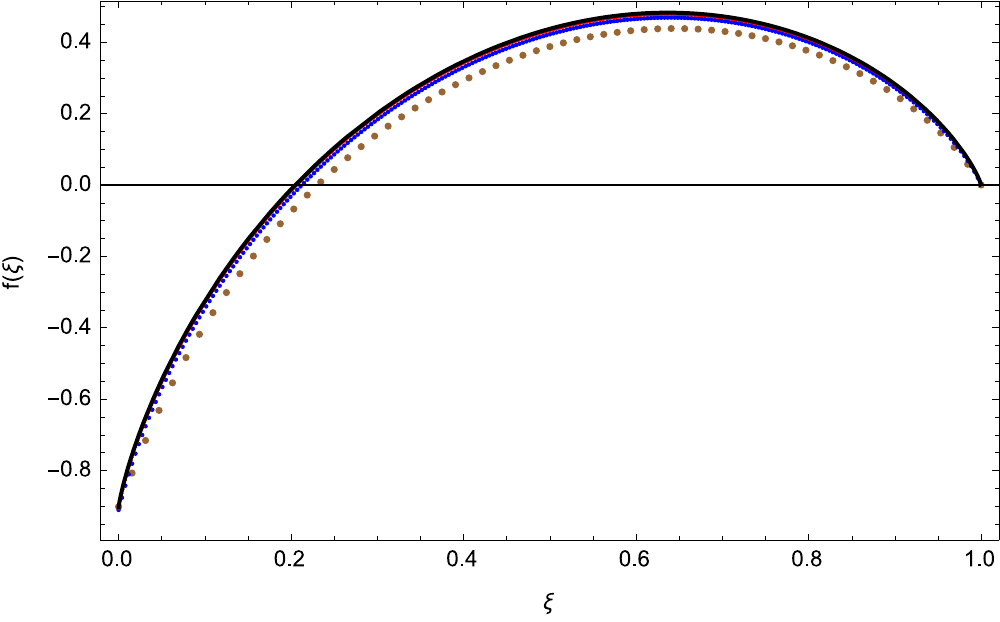}
\caption{
The universal scaling profile $f(\xi=k/n) \simeq \frac{1}{n}\log(\bra{k}\mm{T}^n\ket{0}/(n!)^2)$, generating the nonequilibrium partition function, for the boundary driven $XXX$ spin 1/2 chain with symmetric driving $\Gamma_{\rm L}=\Gamma_{\rm R}=1,b_{\rm L/R}=0$ (corresponding to $\kappa=0$).
The points show numerical data for $n=64$ (brown points), $n=256$ (blue points), and $n=1024$ (red points), compared to the universal analytical result (\ref{eq:xtgt0}) depicted with a black curve.
\label{fig:fx}}
\end{center}
\end{figure}

\subsection{Isotropic ($XXX$) chain -- spin-density profiles and correlations.}

\label{subsect:density}

In the isotropic case one can also calculate the spin-density profile and spin-spin correlations analytically, at least in the case of symmetric (and un-twisted) driving, i.e. in the absence of spectral parameters, $\mu=\lambda=0$ and $\chi=1$. This corresponds to driving with $\Gamma_{\rm L}=\Gamma_{\rm R}=\varepsilon$, $b_{\rm L}=-b_{\rm R}=b$, yielding the representation  (spin) parameter
\begin{equation}
s = \frac{2\ii}{\varepsilon-2\ii b}.
\label{eq:s}
\end{equation} In such a case, one finds a remarkable algebraic relation between the transfer and vertex operators,
$\vmbb{T}= \mm{s}^+\mm{t}^+ + \mm{s}^-\mm{t}^- + 2\mm{s}\mm{t}$, $\vmbb{V}^\z = \mm{s}^+\mm{t}^+ - \mm{s}^-\mm{t}^-$,
\begin{equation}
[\vmbb{T},[\vmbb{T},\vmbb{V}^\z]] + 2\{\vmbb{T},\vmbb{V}^\z\} = 4(s(s+1) + t(t+1)) \vmbb{V}^\z,
\label{eq:TTV}
\end{equation}
which can be derived straightforwardly using only $SU(2)$ commutation relations and our complex spin representation. Note that the relation holds even for a tensor product of two abstract $SU(2)$ algebras, where $s(s+1)$ and $t(t+1)$ have to be replaced by the corresponding Casimir operators $\mm{s}^-\mm{s}^++\mm{s}(\mm{s}+1)$ and $\mm{t}^-\mm{t}^++\mm{t}(\mm{t}+1)$, respectively.
Multiplying Eq.~(\ref{eq:TTV}) by $\bra{0,0}\vmbb{T}^{x-2}$ from the left and by $\vmbb{T}^{n-x-1}\ket{0,0}$ from the right, using the definition (\ref{eq:defMx}) of ${\cal M}^{(n)}_x$ (adding explicit notation of dependence on the chain length $n$), and noting that $t=\bar{s}$, one gets
\begin{equation}
\!\!\!\!\!\!\!\!\!\!\!\!\!\!\!\!\!\!\!\!\!\!\!\!\!\!\!\!\!\!\!\!\!\!
{\cal M}^{(n)}_{x-1} - 2 {\cal M}^{(n)}_{x} + {\cal M}^{(n)}_{x+1} + 2({\cal M}^{(n-1)}_{x-1} + {\cal M}^{(n-1)}_x)\frac{{\cal Z}_{n-1}}{{\cal Z}_n} = 8{\rm Re\,}s(s+1) {\cal M}^{(n-2)}_{x-1}\frac{{\cal Z}_{n-2}}{{\cal Z}_n}.\;\;
\label{eq:Mxbulk}
\end{equation}
This is a closed form difference equation for ${\cal M}^{(n)}_x$, knowing ${\cal Z}_{n-1}/{\cal Z}_{n} \simeq \pi^2/(4n^2)$, which has a unique solution once we specify the boundary conditions 
${\cal M}^{(n)}_1$ and ${\cal M}_n^{(n)}$. These are givien by the following trivially satisfied boundary equations
\begin{equation}
\bra{0,0}(\vmbb{T}-\vmbb{V}^\z)=2 |s|^2 \bra{0,0},\quad (\vmbb{T}+\vmbb{V}^\z)\ket{0,0} = 2|s|^2 \ket{0,0},
\label{eq:bb}
\end{equation}
namely multiplying them, respectively, by $\vmbb{T}^{n-1}\ket{0,0}$ and $\bra{0,0}\vmbb{T}^{n-1}$, we obtain
\begin{equation}
1 + {\cal M}^{(n)}_1 = 2|s|^2\frac{{\cal Z}_{n-1}}{{\cal Z}_n} \simeq \frac{\pi^2|s|^2}{2n^2},\;\;
1 - {\cal M}^{(n)}_n = 2|s|^2\frac{{\cal Z}_{n-1}}{{\cal Z}_n} \simeq \frac{\pi^2|s|^2}{2n^2}.\quad
\label{eq:Mxb}
\end{equation}
We can take TL $n\to\infty$ of equations (\ref{eq:Mxbulk},\ref{eq:Mxb}) obtaining the differential equations for the scaled spin-density profile 
\begin{equation}
{\cal M}\left(\xi=\frac{x-1}{n-1}\right) \simeq {\cal M}^{(n)}_x,
\label{eq:contans}
\end{equation}
specifically
\begin{equation}
{\cal M}''(\xi) = -\pi^2  {\cal M}(\xi),\quad  {\cal M}(0) = 1,\quad  {\cal M}(1) = -1.
\label{eq:difM0}
\end{equation}
The bulk and boundary conditions are all correct to order ${\cal O}((\varepsilon^2+4b^2)^{-1} n^{-2})$.
The cosine-shaped solution of the spin-density profile 
\begin{equation}
{\cal M}(\xi) = \cos(\pi \xi),
\end{equation}
should be universally valid for any fixed $\varepsilon > 0$, in the limit $n\to\infty$, or $\varepsilon \gg \varepsilon^*\sim1/n$.

We can make a similar computation for the two point spin-spin correlation function ${\cal C}_{x,y}$, however here we need to keep the first {\em two} leading orders in the $1/n$ expansion. In principle we should now need also $1/n$-corrected scaling of the partition function
\begin{equation}
\frac{{\cal Z}_{n-1}}{{\cal Z}_n} = \frac{\pi^2}{4 (n-\alpha)^2} (1 + {\cal O}(n^{-2})).
\end{equation}
It will turn out that the final result for connected correlator ${\cal C}_{x,y}$ does not depend on the value of $\alpha$ as it cancels out from our calculation, so we may leave it as a free, unspecified parameter.
Nevertheless, numerical simulations suggest clearly that $\alpha=3/4$ \cite{P11b}. 
We start by upgrading the accuracy of the 1-point function ${\cal M}(\xi)$.
Expanding Eq.~(\ref{eq:Mxbulk}) via (\ref{eq:contans}) to ${\cal O}(n^{-2})$ results in the following differential equation correcting (\ref{eq:difM0})
\begin{equation}
{\cal M}''(\xi) + \pi^2 {\cal M}(\xi) = \frac{\pi^2}{2n}\left(\beta {\cal M}(\xi) + (1-2\xi){\cal M}'(\xi)\right),\quad
\beta := 4(1-\alpha).\;\;
\end{equation}
Writing the solution as ${\cal M}(\xi) = \cos\pi\xi + n^{-1}\widetilde{\cal M}(\xi) + {\cal O}(n^{-2})$, one finds inhomogeneous equation for the first order term
\begin{equation}
\widetilde{\cal M}''(\xi) + \pi^2 \widetilde{\cal M}(\xi) = \frac{\pi^2}{2}\left(\beta\cos\pi\xi - \pi (1-2\xi)\sin\pi\xi\right),
\end{equation}
with boundary conditions $\widetilde{\cal M}(0)=\widetilde{\cal M}'(0)=\widetilde{\cal M}(1)=\widetilde{\cal M}'(1)=0$, following from further expanding Eq.~(\ref{eq:Mxb}),
with a unique solution
\begin{equation}
\widetilde{\cal M}(\xi) = \frac{\pi}{4}\left(\pi \xi(1-\xi)\cos\pi\xi + ((1+\beta)\xi-1)\sin\pi\xi\right).
\end{equation}
Going next to 2-point function we start with the difference system [following from (\ref{eq:TTV})] for its unconnected part
${\cal M}^{(n)}_{x,y}= \bra{0,0}\vmbb{T}^{x-1}\vmbb{V}^\z \vmbb{T}^{y-x-1}\vmbb{V}^\z\vmbb{T}^{n-1-y}\ket{0,0}$, $x<y$:
\begin{eqnarray}
&\!\!\!\!\!\!\!\!\!\!\!\!\!\!\!\!\!\!\!\!\!\!\!\!\!\!\!\!\!\!\!\!\!\!(n-\alpha)^2\left({\cal M}^{(n)}_{x-1,y} - 2 {\cal M}^{(n)}_{x,y} + {\cal M}^{(n)}_{x+1,y}\right) + \frac{\pi^2}{2}\left({\cal M}^{(n-1)}_{x-1,y-1}+{\cal M}^{(n-1)}_{x,y-1}\right) = {\cal O}(n^{-2}), \label{eq:bulk2} \\
&\!\!\!\!\!\!\!\!\!\!\!\!\!\!\!\!\!\!\!\!\!\!\!\!\!\!\!\!\!\!\!\!\!\!
{\cal M}^{(n)}_{1,x}-{\cal M}^{(n)}_x={\cal O}(n^{-2}),\; {\cal M}^{(n)}_{2,x}-{\cal M}^{(n)}_{x}={\cal O}(n^{-2}). \label{eq:bb2}
\end{eqnarray}
The second boundary equation (\ref{eq:bb2}) is derived straightforwardly from a relation analogous to (\ref{eq:bb}) using explicit representation of $\vmbb{T}$ and $\vmbb{V}$.
Writing the scaling function ${\cal M}^{(n)}_{x,y} = {\cal M}(\frac{x-1}{n-1},\frac{y-1}{n-1})$ we expand (\ref{eq:bulk2}) to ${\cal O}(n^{-2})$, in particular keeping the order $1/n$ 
coming from the anti-commutator of (\ref{eq:TTV}). Omitting straightforward details, we obtain a differential equation which fully determines ${\cal M}(\xi_1,\xi_2)$, for $\xi_1 < \xi_2$
\begin{eqnarray}
&&(\partial_1^2 + \pi^2){\cal M}(\xi_1,\xi_2) = \frac{\pi^2}{2n}\left(\beta+(1\!-\!2\xi_1)\partial_1 + 2(1\!-\!\xi_2)\partial_2\right){\cal M}(\xi_1,\xi_2), \label{eq:sys} \\
&& {\cal M}(0,\xi_2) ={\cal M}(\xi_2),\quad \partial_1 {\cal M}(0,\xi_2)= 0.
\end{eqnarray}
This system is solved with an ansatz ${\cal M}(\xi_1,\xi_2) = {\cal C}(\xi_1,\xi_2) + {\cal M}(\xi_1){\cal M}(\xi_2)= \cos(\pi\xi_1)\cos(\pi\xi_2) + {\cal C}(\xi_1,\xi_2) +
n^{-1}\cos(\pi\xi_1){\widetilde{\cal M}}(\xi_2) + n^{-1}{\widetilde{\cal M}}(\xi_1)\cos(\pi\xi_2) + {\cal O}(n^{-2})$, resulting in an inhomogeneous system for the connected correlator
${\cal C}(\xi_1,\xi_2)$, with $\beta$ cancelling out,
\begin{eqnarray}
&&(\partial_1^2 + \pi^2){\cal C}(\xi_1,\xi_2) = \frac{\pi^3}{n}(\xi_2-1)\cos(\pi\xi_1)\sin(\pi\xi_2), \label{eq:sysC} \\
&& {\cal C}(0,\xi_2) = 0,\quad \partial_1 {\cal C}(0,\xi_2) = 0,
\end{eqnarray}
with a solution\footnote{Note a typo in the expression for ${\cal C}(\xi_1,\xi_2)$ in Ref.~\cite{P11b}.}, for $\xi_1 < \xi_2$: ${\cal C}(\xi_1,\xi_2) = -\frac{\pi^2}{2n} \xi_1 (1-\xi_2)\sin(\pi\xi_1)\sin(\pi\xi_2)$.
For $\xi_1 > \xi_2$, the solution is obtained from the symmetry ${\cal C}(\xi_1,\xi_2)={\cal C}(\xi_2,\xi_1)$, or generally (see Fig.~\ref{fig:Cxy})
\begin{equation}
{\cal C}(\xi_1,\xi_2) = -\frac{\pi^2}{2n} {\rm min}(\xi_1,\xi_2)(1-{\rm max}(\xi_1,\xi_2))\sin(\pi\xi_1)\sin(\pi\xi_2).
\label{eq:C2}
\end{equation}
Note a qualitative resemblance to a 2-point function in classical SSEP (see e.g. Ref.~\cite{S01}), apart from a trigonometric factor $\sin(\pi\xi_1)\sin(\pi\xi)2)$ which seems to be of genuinely quantum nature. Our result establishes {\em anti-correlation} ${\cal C} < 0$  between arbitrary pair of spins and the hydrodynamic scaling ${\cal C} \propto 1/n$.
Using this strategy one could derive further all the higher $k$-point transverse spin correlation functions.

\begin{figure}
          \centering	
	\hspace{1.5cm}\includegraphics[scale=0.65]{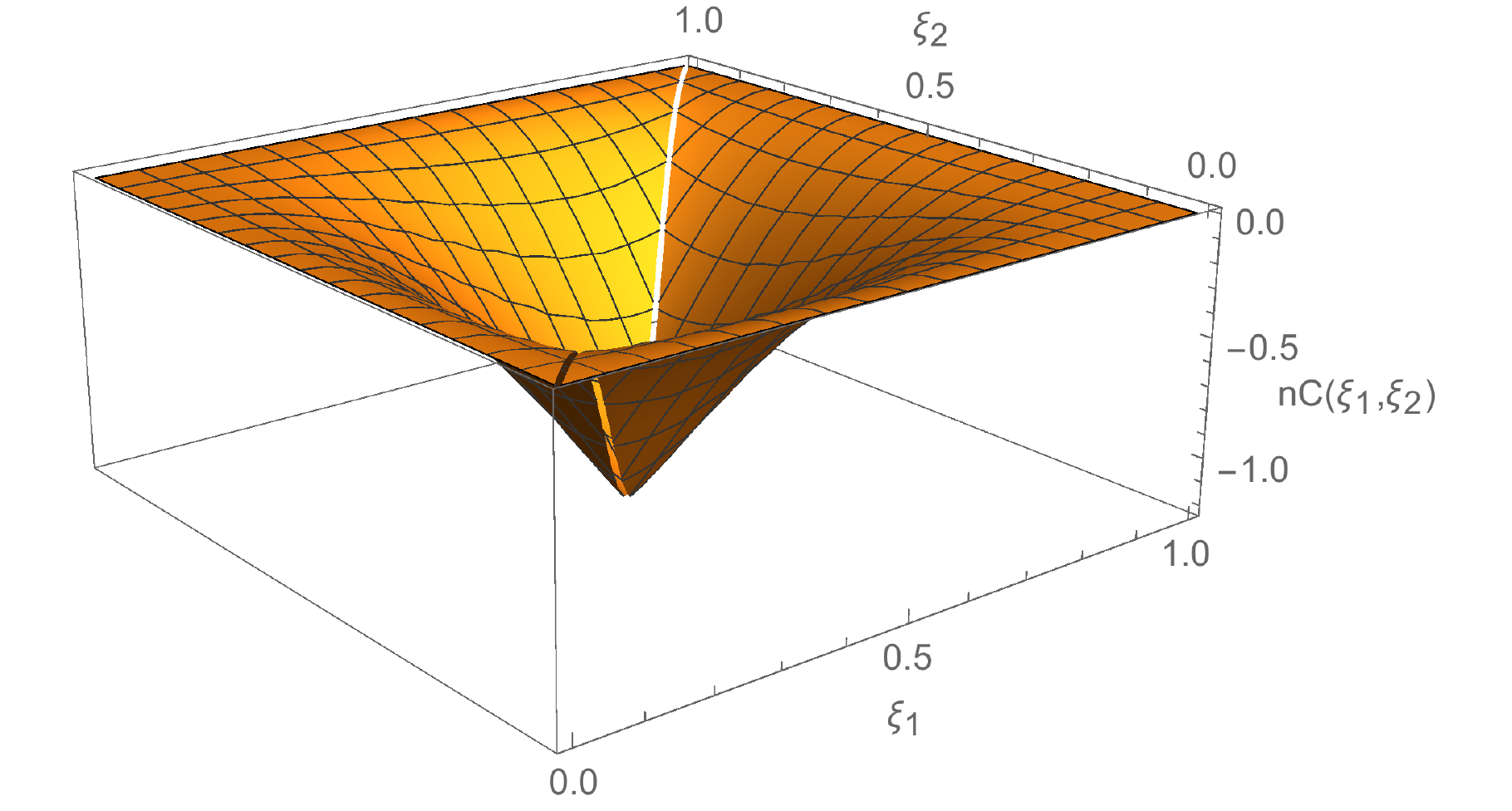}
	\vspace{-3mm}
	\caption{ 
	Scaled connected 2-point spin-spin correlation (\ref{eq:C2}) function $n \times {\cal C}(\xi_1,\xi_2)$ in NESS of $XXX$ chain ($\Delta=1$).}
 	\label{fig:Cxy}
\end{figure}

\subsection{Isotropic ($XXX$) chain -- $SU(2)-$twisted boundary driving}

Here we briefly discuss some details of explicit computation of a more general nonequilibrium partition function 
\begin{equation}
{\cal Z}_n=(\bra{\psi_{\rm L}}\otimes\bra{\overline{\psi_{\rm L}}})\vmbb{T}^n(\ket{\psi_{\rm R}}\otimes\ket{\overline{\psi_{\rm R}}}),
\end{equation}
generalizing the expression (\ref{eq:Z}), and observables in the $XXX$ case with twisted boundary driving as described in subsect.~\ref{subsect:twist} and treated originally in Ref.~\cite{PKS13}. 
As using an arbitrary pair of twists requires full knowledge of the transfer operator $\vmbb{T}$ beyond the diagonal subspace ${\cal K}$ we shall immediately utilize the rotational invariance, and choose a coordinate system in which the source axis is un-twisted while the sink axis is tilted in the $x-z$ plane, specifically 
$\phi_{\rm L}=\phi_{\rm R}=\theta_{\rm L}=0$, $\theta_{\rm R}\equiv \theta \in [0,\pi)$, where $\theta$ is the angle between the source/sink measurement axes. In such a case one can again work with the diagonal auxiliary subspace ${\cal K}$ and the projected transfer operator $\mm{T}$, rewriting ${\cal Z}_n$ as (using asymptotic scaling of local transfer matrix elements (\ref{eq:profile})) 
\begin{eqnarray}
{\cal Z}_n 
&\simeq& \sum_{k=0}^n  \left(\tan\frac{\theta}{2}\right)^{2k} \bra{0}\mm{T}^n\ket{k} \nonumber\\
&\simeq& F_1 (n!)^2\sum_{k=0}^n \exp\left(n f(k/n)+2k \log|\!\tan\theta/2|\right) \nonumber \\
&\simeq& F_1 (n!)^2 n \int_{-\infty}^\infty\!\dd t\,(\dd{\xi}_t/\dd t) \exp\left(n (g_t +  2\xi_t \log|\!\tan\theta/2|)\right).
\end{eqnarray}
Note that $\kappa=0$ ($\chi=1$) as required by solvability of twisted-driving boundary conditions (see subsect.~\ref{subsect:twist}).
This integral can be explicitly evaluated asymptotically ($n\to\infty$) by means of the saddle point method, namely expanding around the extremum of the exponential at $t=\log\tan(\theta/2)$, yielding
\begin{equation}
\!\!\!\!\!\!\!\!\!\!\!\!\!\!\!\!\!\!\!\!\!\!\!\!\!\!\!\!{\cal Z}_n \simeq F_1 (n!)^2 \frac{\sin\theta}{\pi-\theta}\sqrt{\pi n(1+(\pi-\theta)\cot\theta)} \left(\frac{2}{\pi-\theta}\right)^{2n},
\quad
\frac{{\cal Z}_{n-1}}{{\cal Z}_n} \simeq \frac{{(\pi-\theta)}^2}{4n^2}.
\end{equation}
This reproduces the leading $1/\varepsilon$-order result of Ref.~\cite{PKS13} and generalises it to any fixed value of $\varepsilon > 0$ in the large $n$ asymptotics.
Note that ${\cal Z}_n(\theta)$ is not continuous at $\theta=0$, where one should take instead the expression (\ref{eq:Zsc}), as the limits $n\to\infty$ and $\theta\to 0$ do not commute.

Note that for the computation of ${\cal Z}_n(\theta)$ described above we could still allow for some left-right driving asymmetry, and hence non-vanishing value of spectral parameter $\lambda$, as long as $\chi=1$. However, bellow we report, following Ref.~\cite{PKS13}, a simple calculation of vectorial spin-currents and spin-densities which is based on simple closed-form algebraic identities among $\vmbb{T},\vmbb{V}^\alpha$ which are only possible for fully symmetric driving, i.e. $\lambda=0$ and (\ref{eq:s}), so we assume this to be the case for the rest of this discussion. 
For $SU(2)-$symmetric $XXX$ model, one can write the local conservation law for the full spin density vector $\vec{\sigma}_x=(\sigma^\x_x,\sigma^\y_x,\sigma^\z_x)$
and spin-current vector $\vec{j}_{x,x+1}$ satisfying 
\begin{equation}
\frac{\dd}{\dd t} \vec{\sigma}_x = [H,\vec{\sigma}_x] = \vec{j}_{x-1,x} - \vec{j}_{x,x+1},\quad {\rm where}\quad
\vec{j}_{x,x+1}:=2 \vec{\sigma}_x \times \vec{\sigma}_{x+1},
\label{eq:ce}
\end{equation}
and $j_{x,x+1} = \quart j^\z_{x,x+1}$ is the current discussed earlier.
The expectation for the current components $J^\alpha = \ave{j_{x,x+1}^\alpha}$, which due to continuity equation (\ref{eq:ce}) has to be site-independent, can be expressed in terms of the commutators
\begin{equation}
J^\alpha = \frac{2}{{\cal Z}_{n}(\theta)} \sum_{\beta,\gamma\in\{\x,\y,\z\}} \epsilon_{\alpha\beta\gamma} \bra{0,0}[\vmbb{V}^\beta,\vmbb{V}^\gamma]\vmbb{T}^{n-1}\ket{\psi_{\rm R},\overline{\psi_{\rm R}}}
\end{equation}
where $\epsilon_{\alpha\beta\gamma}$ is Levi-Civita symbol.
Facilitating algebraic identities, which hold for $\lambda=0$,
\begin{eqnarray}
\left[\vmbb{V}^\x,\vmbb{V}^\y\right] &=& 2\ii(\mm{t}-\mm{s})\vmbb{T}  = 2\ii\vmbb{T} (\mm{t}-\mm{s}), \\
\left[\vmbb{V}^\y,\vmbb{V}^\z\right] &=& (\mm{s}^+ - \mm{s}^- + \mm{t}^+ - \mm{t}^-)\vmbb{T}= \vmbb{T} (\mm{s}^+ - \mm{s}^- + \mm{t}^+ - \mm{t}^-), \\
\left[\vmbb{V}^\z,\vmbb{V}^\x\right] &=& (\mm{s}^+ + \mm{s}^- - \mm{t}^+ -\mm{t}^-) \vmbb{T}=\vmbb{T} (\mm{s}^+ + \mm{s}^- - \mm{t}^+ -\mm{t}^-),
\end{eqnarray}
and elementary properties of the coherent states (\ref{eq:cs}) we find the scalings of the {\em in-plane} current components
\begin{equation}
\!\!\!\!\!\!\!\!\!\!\!\!\!\!\!\!\!\!\!\!\!\!\!\!\!\!\!\!
J^\z = 4\ii (t-s) \frac{{\cal Z}_{n-1}}{{\cal Z}_n} = \frac{\zeta (\pi-\theta)^2}{n^2},\quad
J^\x = 4\ii (s-t) \frac{{\cal Z}_{n-1}}{{\cal Z}_n} \tan\frac{\theta}{2} = -J^\z \tan\frac{\theta}{2},
\end{equation}
where $\zeta = 4\varepsilon/(\varepsilon^2 + 4 b^2)$. 
The transverse component behaves drastically differently though, it asymptotically scales in an `Ohmic' fashion $n^{-1}$ \cite{PKS13}
\begin{equation}
J^\y \simeq \frac{2(\pi-\theta)}{n},
\end{equation}
and behaves discontinuously at $\theta=0$ where it vanishes ($J^\y|_{\theta=0}=0$).

Vectorial spin-density profiles $\vec{\cal M}_x = \ave{\vec{\sigma}_x}$ are found in full analogy to computation described in subsect.~\ref{subsect:density}, by extending the
algebraic identity (\ref{eq:TTV}) to arbitrary components $[\vmbb{T},[\vmbb{T},\vmbb{V}^\alpha]] + 2\{\vmbb{T},\vmbb{V}^\alpha\} = 4(s(s+1) + t(t+1)) \vmbb{V}^\alpha$ due to $SU(2)$ invariance.
Taking the continuum limit one again arrives to harmonic differential equation 
\begin{equation}
\left(\frac{\dd^2}{\dd\xi^2} + (\pi-\theta)^2\right){\cal M}^\alpha(\xi) = 0
\end{equation} 
for all three components of the continuous spin-density ${\cal M}^\alpha((x-1)/(n-1)) = {\cal M}^\alpha_x$, 
where appropriate boundary conditions follow from explicit representation of $\vmbb{T},\vmbb{V}^\alpha$ and the properties of coherent states (\ref{eq:cs}), resulting in asymptotic harmonic profiles
\begin{equation}
\!\!\!\!\!\!\!\!\!\!\!\!\!\!\!\!\!\!\!\!\!\!\!\!\!\!\!\!
{\cal M}^\z_x \simeq \cos\left((\pi-\theta)\frac{x-1}{n-1}\right),\quad {\cal M}^\x_x \simeq \sin\left((\pi-\theta)\frac{x-1}{n-1}\right),\quad {\cal M}^\y_x \simeq 0.
\end{equation}

\section{Hubbard chain}
\label{hubb}

Here we turn to a different, two-species quantum model, the fermionic Hubbard chain \cite{hubbardbook}. The Hubbard model is the fundamental model of strongly correlated electrons on regular lattices.
Even though the model on a 1D chain has been solved by the coordinate Bethe ansatz a while ago \cite{LW}, it still poses many deep fundamental questions, in particular regarding its dynamical and 
nonequilibrium properties. Here we describe an explicit MPA solution of the corresponding nonequilibrium steady state of the Hubbard chain for diagonal (untwisted) boundary driving. We shall discuss graph theoretic interpretation of the solution and identify key elements of both approaches: IDO method (following Ref.~\cite{P14a}) and local operator divergence (or Lax operator) method (following Ref.~\cite{PP15}).

The Hubbard Hamiltonian for an open chain of $n$ sites, with canonical fermi operators $c_{s,x}$, $x\in\{1\ldots n\}$, $s\in\{\ua,\da\}$, reads
\begin{eqnarray}
H&=&
-2\sum_{s,x} (c^\dagger_{s,x}c_{s,x+1}+c^\dagger_{s,x+1}c_{s,x} ) +  u \sum_{x} (2n_{\uparrow,x}-1)(2n_{\downarrow,x}-1) \nonumber\\
&& + \mu_{\rm L}(n_{\uparrow,1}+n_{\downarrow,1}-1)+\mu_{\rm R}(n_{\uparrow,n}+n_{\downarrow,n}-1), 
\label{eq:HubHfermi}
\end{eqnarray}
where $n_{s,x}=c^\dagger_{s,x} c_{s,x}$.
The nondimensional interaction parameter 
$u=U/(2t_{\rm h})$ contains standard Hubbard interaction $U$ and hopping amplitude $t_{\rm h}$, while $\mu_{\rm L/R}$ are non-dimensional chemical potentials at the boundary sites which shall produce the coherent part of the boundary driving.
The incoherent boundary driving is provided by four Lindblad channels which manifest a pure source/sink for electrons at rates $\Gamma_{\rm L/R}$
\begin{equation}
L_1 = \sqrt{\Gamma_{\rm L}}c^\dagger_{\ua,1},\;\;L_2 = \sqrt{\Gamma_{\rm L}}c^\dagger_{\da,1},\;\;
L_3 = \sqrt{\Gamma_{\rm R}}c_{\ua,n},\;\;L_4 = \sqrt{\Gamma_{\rm R}}c_{\da,n}.
\end{equation}
We shall again be interested in the density operator $\rho_\infty$ of NESS defined as the solution of the stationary Lindblad equation (\ref{eq:lindeq}), $\LL\rho_\infty=0$.
Before proceeding, we shall reformulate the problem in terms of a spin-$1/2$ ladder, namely implementing the Wigner-Jordan transformation which expresses the anticommuting fermi variables
\begin{equation}
c_{\uparrow,x}=P^{(\sigma)}_{x-1} \sigma_x^-,\quad c_{\downarrow,x}=P^{(\sigma)}_n P^{(\tau)}_{x-1} \tau_x^-
\end{equation}
where
$P^{(\sigma)}_x:=\sigma_1^{\rm z}\sigma_2^\z \cdots \sigma_x^{\rm z}$, $P^{(\tau)}_x:=\tau_1^{\rm z}\tau_2^\z \cdots \tau_x^{\rm z}$, in terms of two sets
of independent spins-1/2, $\sigma^s_x,\tau^t_x\; x\in\{1,\ldots,n\},s,t\in{\cal J}=\{+,-,0,\z\}$, $\sigma^0_x\equiv \tau^0_x\equiv\one$, which can be considered as operators over ${\cal H}^{\otimes n}_{\rm p}$. The local physical space is now four dimensional ${\cal H}_{\rm p} = \CC^2\otimes\CC^2$ so that $\sigma^s\tau^t$ span\footnote{Note that here we use letters $s,t$ to name indices denoting physical space components in contradistinction to previous sections where they denoted complex spin parameters.} the complete basis of ${\rm End}({\cal H}_{\rm p})$.
The Hubbard Hamiltonian (\ref{eq:HubHfermi}) then maps to
\begin{eqnarray}
H &=& \sum_{x=1}^{n-1} h_{x,x+1} + h_{\rm L} + h_{\rm R}, \label{eq:Hhub}\\
h_{1,2} &:=& h^\sigma_{1,2} + h^\tau_{1,2} + \frac{u}{2}\left(\sigma^\z_1 \tau^\z_1 + \sigma^\z_{2} \tau^\z_{2}\right), \label{eq:h12} \\
h^\sigma_{1,2}&:=&2 \sigma^{+}_1 \sigma^-_2 + 2 \sigma^{-}_1 \sigma^+_2,\quad h^\tau_{1,2}:=2 \tau^{+}_1 \tau^-_2 + 2 \tau^{-}_1 \tau^+_2 \\
h_{\rm L/R} &:=& \frac{u}{2} \sigma^\z_{1/n} \tau^\z_{1/n} + \frac{\mu_{\rm L/R}}{2}\left(\sigma^\z_{1/n} + \tau^\z_{1/n}\right), \label{eq:hLR}
\end{eqnarray}
while the Lindblad jump operators map to
\begin{equation}
\!\!\!\!\!\!\!\!\!\!\!\!\!\!\!\!\!\!\!\!\!\!\!\!\!\!\!\!\!\!\!L_1 = \sqrt{\Gamma_{\rm L}}\sigma^+_1,\;L_2 = \sqrt{\Gamma_{\rm L}}  P^{(\sigma)}_n \tau^+_1,\;L_3 = -\sqrt{\Gamma_{\rm R}}  P^{(\sigma)}_n \sigma^-_n,\;L_4 = -\sqrt{\Gamma_{\rm R}}  P^{(\sigma)}_n P^{(\tau)}_n \tau^-_n.
\end{equation}
However, since the Hamiltonian and the dissipator $\DD=\sum_{\mu=1}^4 \DD_{L_\mu}$ conserve the numbers of spin-up and spin-down electrons, 
$N_{\sigma}=\sum_{x=1}^n \half(\sigma^\z_x + \one)$, $N_\tau = \sum_{x=1}^n \half(\tau^\z_x + \one)$, $[H,N_{\sigma/\tau}]=0$, $[N_{\sigma/\tau},\DD(\rho)]=\DD([N_{\sigma/\tau},\rho])$,
the {\em unique} steady state $\rho_\infty$, should \cite{BP12} also conserve $N_{\sigma/\tau}$ and their parities $P^{(\sigma/\tau)}_n$, i.e., $[\rho_\infty,P^{(\sigma/\tau)}_n]=0$.
Therefore, $\rho_\infty$ should also be a fixed point of $\LL = -\ii \ad H + \sum_{\alpha=1}^4 \DD_{L_\alpha}$, $\LL\rho_\infty=0$, where $L_\mu$ are replaced by 
\begin{equation}
L_1 = \sqrt{\Gamma_{\rm L}}\sigma^+_1,\;\;L_2 = \sqrt{\Gamma_{\rm L}} \tau^+_1,\;\;
L_3 = \sqrt{\Gamma_{\rm R}} \sigma^-_n ,\;\;L_4 = \sqrt{\Gamma_{\rm R}}  \tau^-_n .
\label{eq:Lhub}
\end{equation}
i.e., with all unitary conserved (non-local) operators removed (noting that $(P^{(\sigma/\tau)})^2=\one$). This remarkable fact teaches us that non-locality of Wigner-Jordan transformation has no effect in the (nonequilibrium) steady state but potentially affects the nature of relaxation. Uniqueness of NESS can be again proved by straightforward application of the Evans--Frigeiro theorem \cite{evans,frigeiro} trivially extending the argument of
subsect.~\ref{unique} to two species of spins.

An important $\ZZ^2-$symmetry of the Hubbard model, analogous to (\ref{eq:parity}) for the Heisenberg model, is generated by the spin-flip operator $G$, i.e. permutation operator between
$\sigma$ and $\tau$ spins (or fermion species), defined as $G \sigma^s G = \tau^s$, $G^2=\one$. Clearly,
\begin{equation} 
G h^{\sigma}_{1,2}G = h^{\tau}_{1,2},\;\,
G h_{1,2} G = h_{1,2}, \;\,
G H G = H, \;\,
G\DD(\rho)G = \DD(G\rho G).\quad
\end{equation}

\subsection{Walking graph state representation of NESS}

In the absence of previously known non-Hermitian Lax operators with enough free complex parameters for the Hubbard model (note that the Hermitian Shastry's Lax matrix \cite{S88} would not work as it lacks a free representation 
parameter), we shall again start with a constructive approach of IDO method \cite{P14a}, while impatient reader is welcome to jump right away to a more elegant formulation of subsect.~\ref{subsect:LaxHubb}.
 
 A useful technical result which can be implemented to establish NESS is a trivial extension of Lemma 1 to a symmetrically boundary driven Hubbard ladder:
 
 \begin{lem}
Let  $\Omega_n \in {\rm End}({\cal H}_{\rm p}^{\otimes n})$ satisfy the following conditions (defining relations):\\ 
(i) a recursion identity for the bulk, setting $\Omega_1 := \one_4$,
\begin{equation}
\!\!\!\!\!\!\!\!\!\!\!\!\!\!\!\!\!\!\!\!\!\!\!\!\!\!\!\!\!\!\!\!
[H,\Omega_n]= -\ii\varepsilon\!\!\sum_{s\in\{0,+\}}\!\!\bigl(\sigma^\z \tau^s \otimes P^{0,s}_{n-1} + \sigma^s \tau^\z \otimes P^{s,0}_{n-1}  
- Q^{0,-s}_{n-1}\otimes \sigma^\z \tau^{-s} - Q^{-s,0}_{n-1}\otimes \sigma^{-s}\tau^\z\bigr),\quad \label{eq:hubdefrel}
\end{equation}
introducing the operators $P^{s,t}_{n-1},Q^{s,t}_{n-1} \in {\rm End}({\cal H}_{\rm p}^{\otimes (n-1)})$
\begin{equation}
P^{s,t}_{n-1} = \frac{\tr_{1}\{(\sigma_1^s\tau_1^t)^\dagger \Omega_n\}}{\tr (\{\sigma^s\tau^t)^\dagger \sigma^s\tau^t\}},\quad
Q^{s,t}_{n-1} = \frac{\tr_{n}\{(\sigma_n^s\tau_n^t)^\dagger \Omega_n\}}{\tr\{(\sigma^s\tau^t)^\dagger \sigma^s\tau^t\}}
\end{equation}
and (ii) the boundary conditions (rendering $\Omega_n$ upper-triangular with unit-diagonal)
\begin{eqnarray}
&& P^{s,t}_{n-1}=0 \quad {\rm for}\; s\in \{ -,\z\}\;{\rm or}\; t \in\{-,\z\}, \nonumber \\
&& Q^{s,t}_{n-1}=0 \quad {\rm for}\; s\in \{ +,\z\}\;{\rm or}\; t \in\{+,\z\}.
\end{eqnarray}
Then, the density operator
\begin{equation}
\rho_\infty = \frac{R_\infty}{\tr R_\infty},\quad R_\infty= \Omega_n \Omega^\dagger_n,
\label{eq:rhoinfH}
\end{equation}
satisfies the fixed point (NESS) condition 
\begin{equation}
\ii [H,\rho_\infty] = \sum_{\alpha=1}^4 \DD_{L_\alpha}(\rho_\infty)
\label{eq:fp}
\end{equation}
for symmetric, totally incoherent driving $\varepsilon = \Gamma_{\rm L}=\Gamma_{\rm R}$, $\mu_{\rm L}=\mu_{\rm R}=0$.
\end{lem}
A straightforward proof along  extension of Eqs.~(\ref{eq:stlind}-\ref{eq:endproof}) is left to the reader. 

The following constructive strategy for obtaining exact NESS solution has been devised \cite{P14a} which is based purely on empirical data about the model. One starts by computing numerical NESS density operators for small systems, feasible for $n\le 6$, and determine their Cholesky factors $\Omega_n$. Operators $\Omega_n$ posses $U(1)\times U(1)$ symmetry, namely they commute separately with the species number operators $[\Omega_n,N_{\sigma/\tau}]=0$ hence all non-vanishing terms of a general operator expansion
$\Omega_{n} = \sum_{\ul{s},\ul{t}} c_{\ul{s},\ul{t}} \bigotimes_{x=1}^n \sigma^{s_x}\tau^{t_x}$ should satisfy $\sum_{x=1}^n d(s_x) = 0$ and $\sum_{x=1}^n d(t_x) = 0$ where the shift-function $d : {\cal J}\to\ZZ$ is defined as 
$d(\pm)=\pm 1, d(0)=d(\z)=0$. Thus each sequence $(s_1,t_1,\dots,s_n,t_n)\in{\cal J}^{2n}$ with $c_{\ul{s},\ul{t}}\neq 0$ can be considered as an $n-$step recurrent walk on a 2-dimensional cartesian grid $\ZZ \times \ZZ$ originating from site 
$(0,0)$, visiting a point $\sum_{y=1}^x (d(s_y),d(t_y))$ after step $x$. However, empirical evidence suggests that the set of non-vanishing terms is much more restricted and can be compactly encoded by a directed graph $({\cal V},{\cal E})$ depicted in Fig.~\ref{fig:HD}.

\begin{figure}
\centering	
\vspace{-1mm}
\includegraphics[width=0.55\columnwidth]{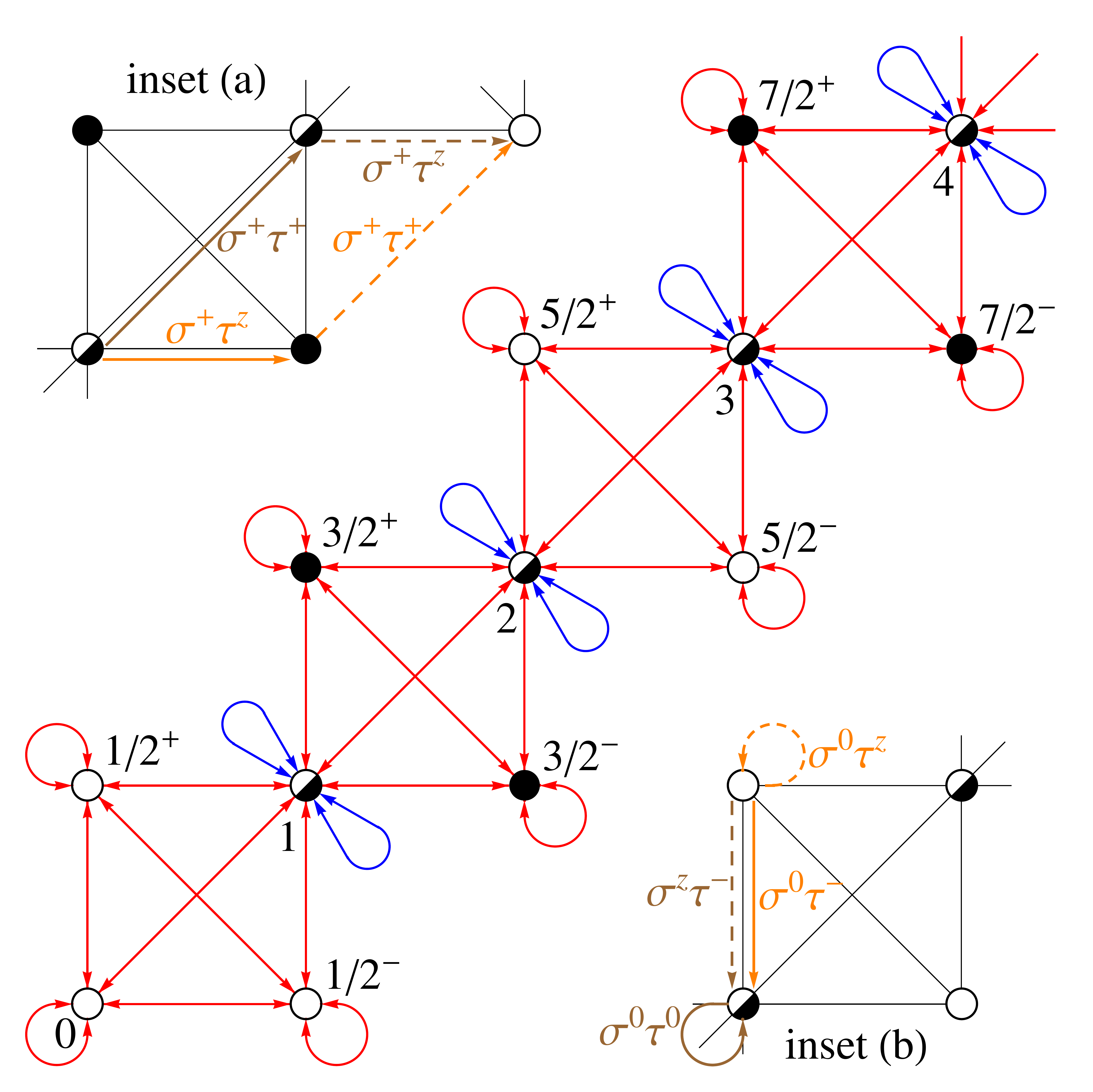}
\vspace{-1mm}
\caption{(from \cite{P14a}) Diagram of a semi-infinite graph $G$ (structure repeating periodically beyond the upper-right corner) showing the {\em allowed} transitions 
for building up the MPA form of NESS for the Hubbard chain. 
Nodes in black, edges with multiplicity 1 in red, and edges with multiplicity 2 in blue.
Each edge $e$ is associated with a physical product-operator $\omega(e)=\sigma^{b^1}\tau^{b^2}$ 
where $b^\nu=0$ ($b^\nu=\z$) for edges connecting white (black) nodes, where
$\nu$ is the Cartesian component which does not change along such $e$ in the diagram.
Degenerate edges correspond to operators $\sigma^0 \tau^0$ ($\mu=+1$) and $\sigma^\z\tau^\z$ ($\mu=-1$). Insets indicate all possible terms (two in each, orange and brown) for two examples of $[h,\omega(e)\otimes\omega(f)]$,
specifically $[h,\sigma^+\tau^+\otimes\sigma^+\tau^0]$ (a), and $[h,\sigma^0\tau^-\otimes\sigma^0\tau^0]$ (b). Full arrows denote valid edge factors, while dashed arrows correspond to {\em defect} operators.}
\label{fig:HD}
\end{figure}

The set of all visited nodes (or vertices) ${\cal V} \subset \ZZ\times\ZZ$ is composed of: the origin $0\equiv (0,0)$, the {\em diagonal} nodes $k \equiv (k,k)$, and upper-, and lower-diagonal nodes,
$(k-\frac{1}{2})^+ \equiv (k-1,k)$, and $(k-\frac{1}{2})^- \equiv (k,k-1)$, for $k \in \{1,2,\ldots\}$. Cartesian coordinates of a node $v\equiv (v^1,v^2)$ will be written as $v^\nu$, $\nu\in\{1,2\}$, in general.
The set of directed edges ${\cal E}(G)$ contains: vertical, horizontal, diagonal, skew-diagonal, and self-connections, as indicated in Fig.~\ref{fig:HD}, where only self-connections
of diagonal nodes are degenerate with multiplicity two. Edges may also be identified with triples $e\equiv (p(e),q(e);\mu(e))$, pointing from node $p(e)$ to node $q(e)$ and having degeneracy label $\mu(e)$, where $\mu=1$ for all edges except diagonal self-connections $(k,k;\mu)$ where $\mu \in\{\pm 1\}$. 

There are two crucial concepts of the IDO method generalising the concept of defect operator $\sigma^\z$ in the $XXZ$ model. The first is the {\em index function} $\omega : {\cal E} \to {\rm End}({\cal H}_{\rm p})$ associating a local physical operator $\omega(e)$ to each edge $e$ of the graph, which can be fully determined by a careful inspection of empirical data, i.e., it should match $\sigma^{s_x}\tau^{t_x}$ for the edge corresponding to $x-$th step of all walks generated by nonvanishing terms of $\Omega_n$. Painting the nodes of the graph as black, white and black\&white (see Fig.~\ref{fig:HD}) we encode the empirical data suggesting the index function 
\begin{equation}
\omega(e) = \sigma^{b^1(e)} \tau^{b^2(e)},\quad {\rm where} \quad b^\nu : {\cal E} \to {\cal J},
\end{equation}
as follows:
$b^\nu (e)=\pm$ if $q^\nu(e)-p^\nu(e)=\pm 1$, while for $q^\nu(e)=p^\nu(e)$, $b^\nu(e)=0$, if $e$ connects {\em white} nodes, and
$b^\nu(e)=\z$, if $e$ connects {\em black} nodes. 
For diagonal self-connections (on {\em black\&white} nodes), the index functions are determined by the degeneracy index, 
$b^\nu(k,k;1)=0, b^\nu(k,k;-1)=\z$.

The second key concept is the {\em defect edge} of the graph.
Let us consider an arbitrary walk of length 2, i.e., a pair of subsequent edges $e,f\in {\cal E}$, with $q(e)=p(f)$.
Writing a Hamiltonian density on a pair of sites as 
$h_{1,2}=h^\sigma_{1,2}+h^\tau_{1,2} +u_1 \sigma^\z_1 \tau^\z_1 + u_2 \sigma^\z_2\tau^\z_2$, which can represent either the bulk or boundary part of 
$H$ (\ref{eq:Hhub}), one finds the following general form of the local commutator of 
$h$ with a tensor product of two {\em valid} edge factors for a pair of consecutive edges ($2-$walks) $e,f\in{\cal E}$, $q(e)=p(f)$
\begin{equation}
\!\!\!\!\!\!\!\!\!\!\!\!\!\!\!\!\!\!\!\!\!\!\!\!\!\!\!\!\!\!\!\![h,\omega(e)\otimes\omega(f)] = \!\!\!\!\!\!\sum_{s,t \in{\cal J},e'\in{\cal E}}^{p(e')=p(e),q(f)-q(e')=d(s,t)}\!\!\!\!\!\!\!\!\!\!\!\!\!\!\!X^{s,t}_{e,f}\,\omega(e')\otimes \sigma^s \tau^t+\!\!\!\!\!\!\!\!\sum_{s,t \in{\cal J},f'\in{\cal E}}^{q(f')=q(f),p(e)-p(f')=d(s,t)}\!\!\!\!\!\!\!\!\!\!\!\!\!\!\!Y^{s,t}_{e,f}\,\sigma^s \tau^t \otimes \omega(f'),\quad \label{local}
\end{equation}
with suitable structure constants $X^{s,t}_{e,f}(u_1,u_2),Y^{s,t}_{e,f}(u_1,u_2)$. We define a displacement vector associated with a pair of Pauli indices as $d(s,t) \equiv (d(s),d(t))$.
Eq.~(\ref{local}), in analogy to identity (\ref{eq:struct}) for $XXZ$ model, has the following crucial property: Any tensor factor $\sigma^s \tau^t$ in the first (or second) sum on RHS of (\ref{local}) is (i) 
neither of the form $\omega(f')$ (or $\omega(e')$), for {\em any} edge $f'$ (or $e'$) which would complete the 2-walk $(e',f')$ to connect the same nodes as $(e,f)$, 
(ii) nor is the missing link $d(s,t)$ between $q(e')$ and $q(f)$ (or $p(e)$ and $p(f')$) provided by {\em any} edge of the graph. We shall call such a factor {\em a defect operator}, or {\em defect edge} if referring to the graph. See insets of Fig.~\ref{fig:HD} for two characteristic examples. 

To each vertex $v\in {\cal V}$ we associate a vector space ${\cal H}_v$, such that the entire auxiliary space is a direct sum ${\cal H}_{\rm a}=\bigoplus_{v\in{\cal V}} {\cal H}_v$, and associate a transition amplitude to each edge $e\in{\cal E}$ as a linear operator $a_e \in {\rm Lin}({\cal H}_{p(e)},{\cal H}_{q(e)})$. Writing a WGS ansatz (\ref{eq:WGSXXZ}) for the amplitude operator of NESS 
\begin{equation}
\Omega_n = \sum_{(e_1,\ldots,e_n)\in{\cal W}_n(0,0)} a_{e_1}a_{e_2}\cdots a_{e_n} \omega(e_1)\otimes \omega(e_2)\otimes \cdots \omega(e_n).
\label{eq:wg}
\end{equation}
one notes that, since the Hamiltonian (\ref{eq:Hhub}) is a sum of local terms, the entire commutator $[H,\Omega_n]$ written in the tensor product expansion (like (\ref{eq:wg})) is composed of terms which 
correspond to $n$-walks over a defective graph with exactly {\em one} defect operator. As the RHS of (\ref{eq:hubdefrel}) has only boundary defects, in the first or last factor, all the terms with defects in the bulk should therefore identically vanish. Picking any pair of nodes, $v,r\in{\cal V}$, which can be connected with at least one $3-$walk, it is then sufficient that the following local conditions are satisfied
\begin{equation}
\!\!\!\!\!\!\!\!\!\!\!\!\!\!\!\!\!\!\!\!\!\!\!\!\!\!\!\!\!\!\!\!\!\sum_{(e,f,g)\in {\cal W}_3(v,r)} a_e a_f a_g  \tr\!\left\{ \left(\omega(e')\otimes \sigma^s \tau^t\otimes\omega(g')\right)^\dagger [H_3,\omega(e)\otimes\omega(f)\otimes\omega(g)]\right\}
 = 0,
 \label{bulk}
\end{equation}
for any pair of edges $e',g'\in{\cal E}(G)$ for which $p(e')=v,q(g')=r$, and any defect component $s,t\in{\cal J}$. 
Here and below $H_k\in {\rm End}({\cal H}^{\otimes k}_{\rm p})$ denotes the Hamiltonian (\ref{eq:Hhub}) for a small cluster of $n=k$ sites with open boundaries $\mu_{\rm L/R}=0$.
Of course, for many combinations $(v,r,e',g',s,t)$ the above equation is trivial, i.e. always satisfied, e.g., when $\sigma^s\tau^t=\omega(f')$ for some 
valid edge $f'$ between $q(e')$ and $p(g')$. The remaining local equations which need to be satisfied to yield (\ref{eq:hubdefrel}) are those for which the defect operator sits at the first $x=1$ or the last $x=n$ tensor factor.
Again, one can factor out sufficient local conditions, which can now be formulated on two sites, in terms of $2-$walks, namely
\begin{eqnarray}
\!\!\!\!\!\!\!\!\!\!\!\!\!\!\!\!\!\!\!\!\!\!\!\!\!\!\!\!\!\!\!\!\!\!\!\!\sum_{(e,f)\in{\cal W}_2(0,v)} a_e a_f \tr\!\left\{\left(\sigma^s\tau^t \otimes \omega(f')\right)^\dagger\left([H_2,\omega(e)\otimes\omega(f)]+\ii \varepsilon \hat{\cal P}(\omega(e))\otimes \omega(f)\right)\right\}=0,\quad \label{eq:lbc} \\
\!\!\!\!\!\!\!\!\!\!\!\!\!\!\!\!\!\!\!\!\!\!\!\!\!\!\!\!\!\!\!\!\!\!\!\!\sum_{(e,f)\in{\cal W}_2(v,0)} a_e a_f \tr\!\left\{\left(\omega(e')\otimes\sigma^s\tau^t\right)^\dagger\left([H_2,\omega(e)\otimes\omega(f)]- \ii \varepsilon \omega(e)\otimes \hat{\cal P}(\omega(f))\right)\right\}=0,\quad \label{eq:rbc}
\end{eqnarray}
for all $e',f' \in{\cal E}$, with $q(f')=v$, and $p(e')=v$. 
 $\hat{\cal P}$ is a map over ${\rm End}({\cal H}_{\rm p})$ defined as $\hat{\cal P}(\rho) := \frac{1}{2}\sigma^\z\otimes \tr_{\!\sigma}(\rho) + \frac{1}{2}\tr_{\!\tau}(\rho)\otimes \tau^{\rm z}$ where $\tr_{\!\sigma}$ (or $\tr_{\!\tau}$) 
denotes the partial trace over $\sigma$ (or $\tau$) qubit of ${\cal H}_{\rm p}$. Here, the set of possible defect operators is quite limited, specifically, to $(s,t)\in\{(0,\z),(\z,0),(+,\z),(\z,+)\}$ for the left boundary conditions (\ref{eq:lbc}), or to $(s,t)\in\{(0,\z),(\z,0),(-,\z),(\z,-)\}$ for the right boundary condition (\ref{eq:rbc}).

Summarizing, finding $a_e$ obeying the three-point recurrences in the bulk (\ref{bulk}) with the two-point boundary conditions (\ref{eq:lbc},\ref{eq:rbc}) is {\em sufficient} for establishing validity of 
Eq.~(\ref{eq:hubdefrel}) with ansatz (\ref{eq:wg}) together with the conjectured structure of the graph $({\cal V},{\cal E})$ and its index function $\omega$ and hence, according to Lemma 2, exactly solving NESS for symmetric incoherent driving $\Gamma_{\rm L}=\Gamma_{\rm R}=\varepsilon$, $\mu_{\rm L/R}=0$, for {\em any} $n$. The solution, unique up to gauge transformations, has been found \cite{P14a} by means of a computer program in Mathematica, requiring the local auxiliary 
spaces of uni-color nodes $0,(k-\half)^\pm$ to be one-dimensional
and those of black\&white nodes $k, k\ge 1$, to be two-dimensional:
\begin{eqnarray}
&&a_{(0,0;+1)} = 1, \quad a_{(0,0;-1)} = 0, \quad a_{(0,1/2^\pm)} = -\varepsilon, \quad a_{(1/2^\pm,0)} = -\ii, \nonumber \\
&& a_{((k-1/2)^\pm,(k-1/2)^\pm)} = -(-1)^k \ihalf\varepsilon,\quad a_{((k-1/2)^\pm,(k-1/2)^\mp)} = \ii \varepsilon,
\nonumber \\
&&a_{(0,1)} = \pmatrix{-2\ii \varepsilon & 0},
\quad a_{(1,0)} = 
\pmatrix{
-\ihalf\varepsilon - u \cr 
-1
},
\nonumber \\
&& a_{(k,(k+1/2)^\pm)} = \pmatrix{
-\varepsilon \cr 0},
\quad
a_{(k,(k-1/2)^\pm)} = \pmatrix{
-(-1)^k (ku + \ihalf \varepsilon){\textstyle\frac{\varepsilon}{2}} \cr
(-1)^{\lfloor\frac{k-1}{2}\rfloor}{\textstyle\frac{\varepsilon}{2}}
}, \nonumber \\ 
&& a_{((k-1/2)^\pm,k)} = \pmatrix{
-\varepsilon & 0
}, \;\,
a_{((k+1/2)^\pm,k)} = \pmatrix{
-(-1)^k(\ii + {\textstyle\frac{k}{2}} \varepsilon u) & (-1)^{\lfloor\frac{k-1}{2}\rfloor}{\textstyle\frac{\varepsilon}{2}}
},
\nonumber\\
&&a_{(k,k+1)} = -(-1)^k 2\ii\varepsilon \pmatrix{
1 & 0\cr
0 & 0
},\nonumber\\
&&
a_{(k+1,k)} = \pmatrix{
(-1)^k (\ii + {\textstyle\frac{k}{2}} \varepsilon u)(\ii (k\!+\!1) u - {\textstyle\frac{\varepsilon}{2}}) & (-1)^{\lfloor\frac{k-1}{2}\rfloor}({\textstyle\frac{\varepsilon}{2}}-\ii(k\!+\!1)u){\textstyle\frac{\varepsilon}{2}},\cr
(-1)^{\lfloor\frac{k}{2}\rfloor}(\ihalf k \varepsilon u - 1) & \ihalf\varepsilon
},
\nonumber\\
&& a_{(k,k;(-1)^k)} =
\pmatrix{
(-1)^k (1- \ihalf k \varepsilon u) & (-1)^{\lfloor\frac{k-1}{2}\rfloor} \ihalf \varepsilon \cr
0 & 0 
}, 
\nonumber\\
&&a_{(k,k;-(-1)^k)} =
\pmatrix{
(-1)^k ({\textstyle\frac{\varepsilon}{2}}- \ii ku)\frac{\varepsilon}{2} & 0 \cr
(-1)^{\lfloor \frac{k-1}{2}\rfloor} \ihalf\varepsilon & 0
}.  \label{givea} \end{eqnarray}

\subsection{Lax representation of NESS}

\label{subsect:LaxHubb}

Having an exact NESS solution for the Hubbard model at hand, one can now explore its mathematical properties more deeply. A strong motivation for that comes from observation that MPA formulation of the WGS ansatz (\ref{eq:wg},\ref{givea}) $\Omega_n = \sum_{\ul{s},\ul{t}} \mm{A}_{s_1,t_1}\cdots \mm{A}_{s_n,t_n} \sigma^{s_1}_1\tau^{t_1}_1\cdots \sigma^{s_n}_n\tau^{t_n}_n$
with $\mm{A}_{s,t} = \bigoplus_{e\in{\cal E}} \delta_{s,b^1(e)} \delta_{t,b^2(e)} a_e$ allows an explicit factorization $\mm{A}_{s,t} = \mm{S}_s \mm{T}_t \mm{X}$, with $\mm{S}_s,\mm{T}_t,\mm{X}\in{\rm End}({\cal H}_{\rm a})$, $[\mm{S}_s,\mm{T}_t]=0$,
which is, in spirit, very close to Shastry's form of the Lax operator \cite{S88}. We show in this section how a general Lax form of the NESS amplitude operator $\Omega_n$ can be derived satisfying a generalised Sutherland relation (closely following Ref.~\cite{PP15}) which reproduces and generalizes the result of the previous section, namely it solves the boundary driven Hubbard chain for arbitrary rates $\Gamma_{\rm L/R}$ and chemical potentials $\mu_{\rm L/R}$.
IDO technique could then be re-interpreted merely as a graph theoretical representation of the Sutherland condition formulated locally between adjacent vertices of the graph.

It turns advantageous here to choose a particular basis of auxiliary sub-spaces for diagonal (black\&white) vertices ${\cal H}_{k\ge 1}={\rm lsp}\{ \ket{k^-}, \ket{k^+}\}$, ${\cal H}_0 = {\rm lsp}\{ \ket{0^+} \}$, and hence to identify the nodes of the graph
with unique labels of individual auxiliary basis states ${\cal V}=\{0^{+},\frac{1}{2}^{+},\frac{1}{2}^{-},1^{-},1^{+},\frac{3}{2}^{+},\frac{3}{2}^{-},2^{-},2^{+}\ldots\}$, so that the entire infinite-dimensional auxiliary space is a simple linear span ${\cal H}_{\rm a}={\rm lsp}\{\ket{v}; v\in{\cal V}\}$.
We extend the definition of the spin-flip $\mm{G}$ over ${\cal H}_{\rm a}$ as a diagonal reflection of the graph, $\mm{G}\ket{k^\pm} = \ket{k^\pm}$, $\mm{G} \ket{k\!+\!\half^\pm} = \ket{k\!+\!\half^\mp}$, $k \in \ZZ^+$.
We begin our analysis with a simple observation:

\begin{lem} \cite{PP15}
Assume there exist operators $\mm{S},\acute{\mm{S}},\grave{\mm{S}},\mm{T},\acute{\mm{T}},\grave{\mm{T}} \in {\rm End}({\cal H}_{\rm a}\otimes {\cal H}_{\rm p})$, and $\mm{X},\mm{Y}\in {\rm End}({\cal H}_{\rm a})$ (acting as scalars over ${\cal H}_{\rm p}$),  satisfying
\begin{eqnarray}
&& [h^\sigma_{1,2},\mm{S}_1 \mm{X} \mm{S}_2]  = \acute{\mm{S}}_1 \mm{X} \mm{S}_2 - \mm{S}_1\mm{X} \grave{\mm{S}}_2, \label{eq:id1}\\
&& [h^\tau_{1,2},\mm{T}_1 \mm{X} \mm{T}_2] = \acute{\mm{T}}_1 \mm{X} \mm{T}_2 - \mm{T}_1\mm{X}\grave{\mm{T}}_2, \label{eq:id2}\\
&& \mm{S}\acute{\mm{T}} + \mm{T}\acute{\mm{S}} - \grave{\mm{S}}\mm{T} - \grave{\mm{T}}\mm{S} = [\mm{Y}-u \sigma^\z \tau^\z,\mm{S}\mm{T}], \label{eq:id3}\\
&& [\mm{S},\mm{T}] = 0, \label{eq:id4}\\
&&[\mm{X},\mm{Y}] = 0. \label{eq:id5}
\end{eqnarray}
Subscripts, like in $\mm{S}_x$, indicate independent local physical spaces pertaining to sites $x$ in the embeded representation ${\rm End}({\cal H}_{\rm a}\otimes {\cal H}^{\otimes n}_{\rm p})$. Then, one can define a Lax operator and its `derivative'
$\mm{L},\widetilde{\mm{L}}\in {\rm End}({\cal H}_{\rm a}\otimes {\cal H}_{\rm p})$ as
\begin{eqnarray}
&&\mm{L} = \mm{S}\mm{T}\mm{X}, \label{eq:L} \\
&&\widetilde{\mm{L}} = \half(\mm{S}\acute{\mm{T}} + \mm{T}\acute{\mm{S}} + \grave{\mm{S}}\mm{T} + \grave{\mm{T}}\mm{S} - \{\mm{Y},\mm{S}\mm{T}\})\mm{X}, \label{eq:Lt}
\end{eqnarray}
such that the following Sutherland-Shastry relation (or generalized local operator divergence condition) holds
\begin{equation}
[h_{1,2},\mm{L}_1 \mm{L}_2] = (\widetilde{\mm{L}}_1 + \mm{Y} \mm{L}_1)\mm{L}_2 - \mm{L}_1(\widetilde{\mm{L}}_2 + \mm{L}_2\mm{Y}).
\label{eq:gLOD}
\end{equation}
\end{lem}
The proof is a straightforward insertion of (\ref{eq:L},\ref{eq:Lt}) into Eq.~(\ref{eq:gLOD}) followed by subsequent application of identities (\ref{eq:id1}-\ref{eq:id5}) observing the definition (\ref{eq:h12}).

 We continue by {\em deriving} an explicit closed form representation of algebraic identities (\ref{eq:id1}-\ref{eq:id5}). Assuming the spin-flip symmetry
 \begin{equation}
 \mm{G} \mm{S} \mm{G} = \mm{T},\,
 \mm{G} \acute{\mm{S}} \mm{G} = \acute{\mm{T}},\,
 \mm{G} \grave{\mm{S}} \mm{G} = \grave{\mm{T}},\, [\mm{G},\mm{X}]=[\mm{G},\mm{Y}]=0,
 \label{eq:sym}
 \end{equation}
 and writing out the components $\mm{S}=\sum_{s\in{\cal J}} \mm{S}^s \sigma^s$, $\mm{T}=\sum_{t\in{\cal J}} \mm{T}^t \tau^t$, and similarly for $\acute{\mm{S}},\grave{\mm{S}},\acute{\mm{T}},\grave{\mm{T}}$, we find that Eqs. (\ref{eq:id1}) and (\ref{eq:id2}) are equivalent, Eq.~(\ref{eq:id3}) is invariant under $\mm{G}$, while Eq.~(\ref{eq:id4}) implies $[\mm{S}^s,\mm{T}^t]\equiv 0$.
Eqs. (\ref{eq:id1},\ref{eq:id2}) are in fact just a particularly `decorated' 6-vertex Yang-Baxter equations for free fermion (or $XX$) chains. We shall thus make an ansatz for $\mm{S}^s,\mm{T}^t$
in which each square plaquette $\{ k^+, k\!+\!\half^+,k\!+\!\half^-,k\!+\!1^-\}$ of the graph spans a pair of representations of a free fermion algebra (see Fig.~\ref{ST}), namely requiring that
$\{\mm{S}^+,\mm{S}^-\}$ (and similarly for $\{\mm{T}^+,\mm{T}^-\}$ via (\ref{eq:sym})) is  in the center of the algebra generated by $\mm{S}^s,\mm{T}^t$
\begin{equation}
[\{\mm{S}^+,\mm{S}^-\},\mm{S}^s] = [\{\mm{S}^+,\mm{S}^-\},\mm{T}^t] = 0,\quad s,t\in{\cal J}.
\end{equation}
One finds that these conditions are fulfilled by an ansatz
\begin{eqnarray}
&&\mm{S}^+ =\sqrt{2} \sum_{k=0}^\infty \left(\ket{k^+}\bra{k\!+\!\half^+} + \ket{k\!+\!\half^-}\bra{k\!+\!1^-}\right),\label{eq:Sp} \\
&&\mm{S}^-  = \sqrt{2} \sum_{k=0}^\infty (-1)^k \left(\ket{k\!+\!\half^+}\bra{k^+} + \ket{k\!+\!1^-}\bra{k\!+\!\half^-}\right),\nonumber\\
&&\mm{S}^0 = \sum_{k=0}^\infty \bigl(\ket{2k^+}\bra{2k^+} + \ket{2k\!+\!\half^+}\bra{2k\!+\!\half^+} \nonumber \\
&&\qquad+ \ket{2k\!+\!1^-}\bra{2k\!+\!1^-}+ \ket{2k\!+\!\half^-}\bra{2k\!+\!\half^-}\bigr)\nonumber\\
&&\quad+ \lambda\sum_{k=1}^\infty \left( \ket{2k\!-\!\half^+}\bra{2k\!-\!\half^+} + \ket{2k^-}\bra{2k^-}\right), \nonumber \\
&&\mm{S}^\z = \sum_{k=1}^\infty \bigl(\ket{2k\!-\!1^+}\bra{2k\!-\!1^+} + \ket{2k\!-\!\half^+}\bra{2k\!-\!\half^+} \nonumber \\
&&\qquad + \ket{2k^-}\bra{2k^-} + \ket{2k\!+\!\half^-}\bra{2k\!+\!\half^-}\bigr)  \nonumber\\
&&\quad + \lambda\sum_{k=0}^\infty\left(\ket{2k\!+\!\half^+}\bra{2k\!+\!\half^+} + \ket{2k\!+\!1^-}\bra{2k\!+\!1^-}\right), \nonumber
\end{eqnarray}
where $\lambda\in\CC$ is a free parameter.
Eqs. (\ref{eq:sym}) imply definition of another set of auxilliary fermi operators $\mm{T}^t=\mm{G}\mm{S}^t\mm{G}$, such that Eq.~(\ref{eq:id4}) is satisfied.

Furthermore, one can write a consistent ansatz for the `interaction' operator $\mm{X}$ coupling the neighbouring plaquettes:
\begin{eqnarray}
\mm{X} &=& \ket{0^+}\bra{0^+} +
 \sum_{k=1}^\infty (-1)^{k}\!\!\!\sum_{\nu,\nu'\in\{-,+\}}   \ket{k^{\nu}} X^{\nu,\nu'}_k  \bra{k^{\nu'}} \nonumber\\
&+& w \sum_{k=0}^\infty (-1)^k \left(\ket{k\!+\!\half^+}\bra{k\!+\!\half^+} +
\ket{k\!+\!\half^-}\bra{k\!+\!\half^-}\right),  \label{eq:Xansatz}
\end{eqnarray}
where $X_k=\{ X^{\nu,\nu'}_k \}_{\nu,\nu'\in\{-,+\}}$ are still unknown $2\times 2$ matrices and $w\in\CC$ is another free parameter.
Namely, Eq.~(\ref{eq:id1}) yields a system of linear equations for auxiliary operators $\acute{\mm{S}}^s\mm{X},\mm{X}\grave{\mm{S}}^s$,
with a unique solution parametrised by $X_k,w,\lambda$:
\begin{eqnarray}
&&\acute{\mm{S}}^+\mm{X}  =  -2\sqrt{2}\sum_{k=1}^{\infty}(-1)^k X^{+-}_k \ket{k^-}\bra{k\!+\!\half^+}, \label{eq:Sb}\\
&&\acute{\mm{S}}^-\mm{X} =   -2\sqrt{2}\sum_{k=1}^{\infty} X^{-+}_k\ket{k^+}\bra{k\!-\!\half^-}, \nonumber\\
&&\mm{X}\grave{\mm{S}}^{+}  =  2\sqrt{2}\sum_{k=1}^{\infty}
(-1)^k X^{+-}_k\ket{k\!-\!\half^-}\bra{k^+} \nonumber\\
&&\mm{X}\grave{\mm{S}}^- =  -2\sqrt{2}\sum_{k=1}^{\infty}
X^{-+}_k \ket{k\!+\!\half^+}\bra{k^-}, \nonumber\\
&&\acute{\mm{S}}^0\mm{X}=\mm{X}\grave{\mm{S}}^0=
2\sum_{k=1}^{\infty}\bigl(w\ket{2k\!-\!1^+}\bra{2k\!-\!1^+}-w\ket{2k^-}\bra{2k^-} \nonumber \\
&&\quad - X^{++}_{2k-1}\ket{2k\!-\!\half^+}\bra{2k\!-\!\half^+}-X^{--}_{2k}\ket{2k\!-\!\half^-}\bra{2k\!-\!\half^-}\bigr) \nonumber \\
&&+2\lambda
\sum_{k=0}^{\infty}
\bigl(-w \ket{2k^+}\bra{2k^+}+X^{--}_{2k+1}\ket{2k\!+\!\half^-}\bra{2k\!+\!\half^-}\bigr), \nonumber\\
&&\acute{\mm{S}}^\z\mm{X}=\mm{X}\grave{\mm{S}}^\z=
 2\sum_{k=0}^{\infty}\bigl(w \ket{2k\!+\!1^-}\bra{2k\!+\!1^-}-w\ket{2k^+}\bra{2k^+} \nonumber \\
 &&\quad + X^{++}_{2k}\ket{2k\!+\!\half^+}\bra{2k\!+\!\half^+} + X^{--}_{2k+1}\ket{2k\!+\!\half^-}\bra{2k\!+\!\half^-}\bigr) \nonumber \\
 && + 2\lambda
\sum_{k=1}^{\infty}\bigl(w\ket{2k\!-\! 1^+}\bra{2k\!-\! 1^+}-X^{--}_{2k} \ket{2k\!-\!\half^-}\bra{2k\!-\!\half^-}\bigr). \nonumber
\end{eqnarray}
Assuming $\mm{X}$ to be invertible (i.e., $w\neq 0$, $\det X_k \neq 0$) and
plugging expressions (\ref{eq:Sb}) to the remaining identity (\ref{eq:id3}) result in (i) a unique consistent expression for the `spectral' operator $\mm{Y}$
\begin{equation}
\mm{Y}=-2\lambda u\sum_{k=0}^{\infty}\bigl(\ket{k^+}\bra{k^+} + \ket{k\!+\!1^-}\bra{k\!+\!1^-}\bigr),
\end{equation}
which clearly commutes with $\mm{X}$, as required by (\ref{eq:id5}),
and (ii) recurrence relations for the matrix elements of $X_k$:
$X^{--}_{k+1}=X^{--}_{k} - u w$, $X^{++}_{k+1}=X^{++}_{k} - u w (1-\lambda^2)$,
and $\det X_k = -w^2$, while also fixing the initial condition $X^{++}_{0}=1$, $X^{--}_{0}=-w^2$, yielding
\begin{equation}
X_k(\lambda,w)=\pmatrix{
-(w+k u)w & 1-(w+k u)w(1-\lambda^2)\cr
 -k u w & 1- k u w (1-\lambda^2)
}.
\label{Xk}
\end{equation}
Note that $X^{-+}_k/X^{+-}_k$ can be chosen freely exploring a gauge freedom
$\ket{k^\pm}\to \xi^{\pm 1} \ket{k^\pm}$, $k=1,2\ldots$
\begin{figure}
\hbox{
\hspace{1.35cm}

\includegraphics[width=0.9\columnwidth]{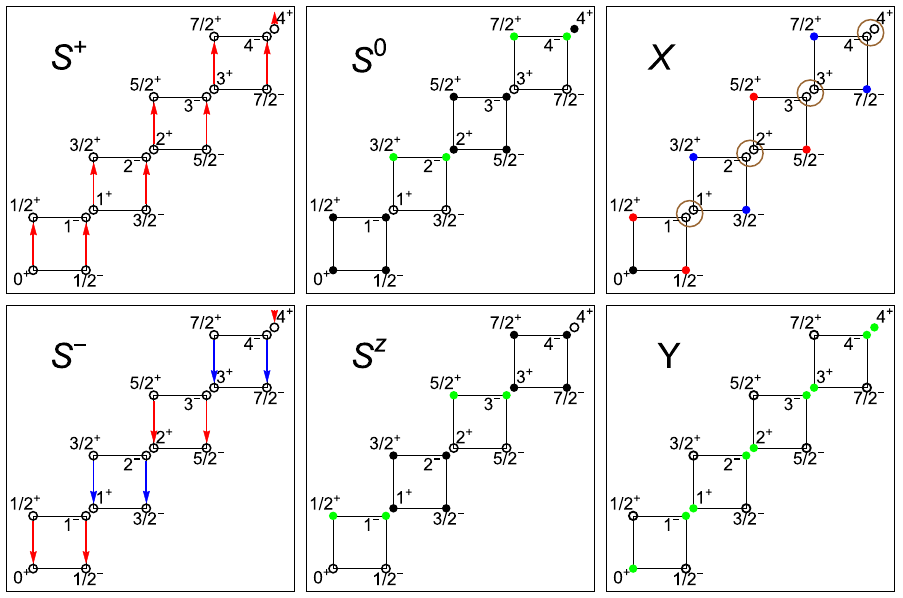}
}
\caption{Diagrammatic representation of factors of the Lax operator where auxiliary states are labelled by vertices ${\cal V}$. Diagrams for $\mm{T}^s$ are obtained by
reflection of those of $\mm{S}^s$ across the diagonal.
Red/blue arrows indicate offdiagonal transitions with amplitude $\pm \sqrt{2}$.
Red, blue, green, black, open points represent diagonal multiplications by $w,-w,\propto\lambda,1,0$, respectively, and brown circles represent multiplications by $2\times 2$ matrices $X_k$.}
\label{ST}
\end{figure}
We have thus constructed two-parameter representation of the Lax matrix $\mm{L}(\lambda,w)=\mm{S}(\lambda)\mm{T}(\lambda)\mm{X}(\lambda,w)$ satisfying Sutherland-Shastry relation (\ref{eq:gLOD}). We propose to call $\lambda$ a {\em spectral parameter} and $w$ a {\em representation parameter}. 
Remarkably, our representation is generically of infinite-dimension, for any nonzero $u$, and seems to be
essentially different from Shastry's \cite{S88} which, including appropriate auxiliary operators $\acute{\mm{S}},\grave{\mm{S}},\ldots$, forms a $4$-dimensional representation of the algebra (\ref{eq:id1}-\ref{eq:id5}).

As an application of such novel Lax operator we consider a markovian master equation and demonstrate how to obtain NESS fixed point (\ref{eq:fp}) for general Hamiltonian (\ref{eq:Hhub}) and boundary dissipation (\ref{eq:Lhub}):
\smallskip
\noindent
\begin{thm} \cite{PP15}
Unique fixed point $\LL\rho_\infty=0$ of boundary driven Hubbard chain reads
\begin{equation}
\rho_\infty = (\!\tr R_\infty)^{-1} R_\infty,\quad R_\infty = \Omega\,\Omega^\dagger K,
\end{equation}
where $\Omega=\Omega_n(\lambda,w)$ is a {\em highest-weight} transfer matrix
\begin{equation}
\Omega = \bra{0^+} \mm{L}_1(\lambda,w)\mm{L}_2(\lambda,w)\cdots  \mm{L}_n(\lambda,w)\ket{0^+}
\label{eq:OmHub}
\end{equation}
and $K$ is a diagonal operator
\begin{equation}
K = K_1 K_2\cdots K_n,\quad K_n = \exp\left(\kappa (\sigma^\z_j+\tau^\z_j)\right)
\end{equation}
with $\kappa = \half\log \Gamma_{\rm L}/\Gamma_{\rm R}$ and parameters $\lambda,w$ are related to coherent and incoherent biases
\begin{equation}
\lambda = \frac{\Gamma_{\rm L}-\Gamma_{\rm R} - \ii (\mu_{\rm L}+\mu_{\rm R})}{\Gamma_{\rm L}+\Gamma_{\rm R} -
\ii (\mu_{\rm L}-\mu_{\rm R})},\;
w =\frac{1}{4}\left(\mu_{\rm L}-\mu_{\rm R} + \ii\left(\Gamma_{\rm L}+\Gamma_{\rm R}\right)\right).
\label{eq:pars}
\end{equation}
\end{thm}
\begin{proof}
Let us now invoke two copies of the auxiliary space and define operators $\mm{S},\mm{T},\mm{\un{S}},\mm{\un{T}}\in{\rm End}({\cal H}_{\rm a}\otimes {\cal H}_{\rm a}\otimes {\cal H}_{\rm p})$ as
\begin{eqnarray*}
&&\mm{S} = \sum_s \mm{S}^s \otimes \one_{\rm a} \otimes \sigma^s,\;\,\qquad \mm{T} = \sum_t \mm{T}^t \otimes \one_{\rm a} \otimes \tau^t,\quad {\rm and}\\
&&\mm{\un{S}} = \sum_s  \one_{\rm a} \otimes \mm{\bar{S}}^s \otimes (\sigma^s)^T,\quad \mm{\un{T}} = \sum_t \one_{\rm a} \otimes \mm{\bar{T}}^t \otimes (\tau^t)^T.
\end{eqnarray*}
$()^T$ denotes the matrix transposition and $\bar{\mm{S}}$ the complex conjugation, i.e. replacement $\lambda,w,\to\bar{\lambda},\bar{w}$, and similarly for
$\acute{\mm{S}},\grave{\mm{S}},\acute{\mm{\un{S}}},\grave{\mm{\un{S}}},\acute{\mm{T}},\grave{\mm{T}},\acute{\mm{\un{T}}},\grave{\mm{\un{T}}}$,
and $\mm{X},\mm{\un{X}},\mm{Y},\mm{\un{Y}}\in{\rm End}({\cal H}_{\rm a}\otimes{\cal H}_{\rm a})$.
In fact, the primed operators
$ \mm{\un{S}},\acute{\mm{\un{S}}},\grave{\mm{\un{S}}},\mm{\un{T}},\acute{\mm{\un{T}}},\grave{\mm{\un{T}}},\mm{\un{X}},\mm{\un{Y}} $ generate a {\em conjugate} representation of the algebra (\ref{eq:id1}-\ref{eq:id5}). Noting $[h_{1,2},K_1 K_2]=0$ and the Jacobi identity one finds that the following double auxiliary operators
\begin{equation}
\vmbb{L}_j = \mm{L}_j\mm{\un{L}}_{\!\!j} K_j,\;
\widetilde{\vmbb{L}}_j = (\widetilde{\mm{L}}_j \mm{\un{L}}_{\!\!j} - \mm{L}_j\widetilde{\mm{\un{L}}}_{\!\!j})K_j,\;
\vmbb{Y} = \mm{Y}-\mm{\un{Y}},
\end{equation}
also respect  Sutherland-Shastri relation (\ref{eq:gLOD}), resulting in the telescoping series
\begin{eqnarray}
\sum_{j=1}^{n-1} \left[h_{j,j+1}, \vmbb{L}_1\vmbb{L}_2\cdots \vmbb{L}_n\right] &=&
(\widetilde{\vmbb{L}}_1+\{\vmbb{Y},\vmbb{L}_1\})\vmbb{L}_2\cdots\vmbb{L}_{n} \nonumber \\
&-&\vmbb{L}_1\cdots\vmbb{L}_{n-1}(\widetilde{\vmbb{L}}_n+\{\vmbb{Y},\vmbb{L}_n\}).\qquad \label{eq:telescop}
\end{eqnarray}
Double Lax operator expresses NESS in a compact form
\begin{equation}
R_\infty = \bra{0^+,0^+} \vmbb{L}_1\vmbb{L}_2 \cdots \vmbb{L}_n \ket{0^+,0^+},
\end{equation}
hence the fixed point condition $\LL R = 0$ becomes, after applying (\ref{eq:telescop}) to $[H,R]$, equivalent
to a pair of equations for ultralocal operators at the boundary physical sites
\begin{eqnarray}
&& \bra{0^+,0^+}\left(\ii \Gamma_{\rm L}(\hat{\cal D}_{\sigma^+}+ \hat{\cal D}_{\tau^+})\vmbb{L} + \widetilde{\vmbb{L}}+ \vmbb{L}\vmbb{Y}+[h_{\rm L},\vmbb{L}]\right)=0,\quad
\nonumber\\
&&\left(\ii \Gamma_{\rm R}(\hat{\cal D}_{\sigma^-}+ \hat{\cal D}_{\tau^-})\vmbb{L} - \widetilde{\vmbb{L}} - \vmbb{Y}\vmbb{L}+[h_{\rm R},\vmbb{L}]\right)\ket{0^+,0^+}=0,\quad
\label{eq:bc}
 \end{eqnarray}
 where boundary interactions with fields, $h_{\rm L/R}$, are defined in (\ref{eq:hLR}).
 Using explicit forms (\ref{eq:Sp}-\ref{Xk}) and in particular $\mm{X}\ket{0^+,0^+}=\mm{\un{X}}\ket{0^+,0^+}=\ket{0^+,0^+}$, each of Eqs.~(\ref{eq:bc})
  results in $\dim{\cal H}_{\rm p}\times \dim{\cal H}_{\rm p}=16$ equations for (bra/ket) vectors from ${\cal H}_{\rm a}\otimes{\cal H}_{\rm a}$, most of them
  trivially satisfied, whereas the non-trivial ones are equivalent to conditions (\ref{eq:pars}).
\end{proof}

We have presented infinite-dimensional irreducible representation of Lax operator and Sutherland-Shastry compatibility condition and shown how it can be employed to yield exact NESS of asymmetrically boundary driven Hubbard chain with arbitrary boundary chemical potentials. We shall discuss later in subsection~\ref{subsect:undeformed} how this Lax operator can be used to define new conserved operators of the Hubbard model which break particle-hole or spin-reversal symmetry, in full analogy with the situation in the $XXZ$ model. One could now also embark on computation of local observables in the Hubbard model. For example, linear dependence of the amplitudes (\ref{Xk}) on auxiliary state $k$ immediately yields, 
following the same asymptotic analysis as for computing the nonequilibrium partition function in the case of $XXX$ chain presented in subsect.~\ref{subsect:neZ}, a universal scaling of the spin/charge currents $J\sim n^{-2}$ and cosine-shaped
spin/charge density profile, as observed in numerical simulations of Ref.~\cite{PZ12}.

\section{Lai-Sutherland spin-1 chain}
\label{lai}

As the third characteristic example we outline (following Ref.~\cite{IP14}) exact solution of a boundary driven spin-1 Lai-Sutherland chain which, due to the three-state nature of the model with only a pair of states being dissipatively transformed at the boundary, allows for macroscopically degenerate NESS.

Consider a finite chain of $n$ sites with $3-$state local physical space ${\cal H}_{\rm p}=\CC^3$. 
Using the Weyl matrix basis $\{ e^{ij} = \ket{i}\!\bra{j}; i,j=1,2,3\}$ of $\End({\cal H}_{\rm p}) = \frak{gl}_3$, we define a full set of local generators of the full matrix algebra $\mathfrak{F}=\End({\cal H}_n)$, where ${\cal H}_n={\cal H}^{\otimes n}_{\rm p}$ denotes $3^n$-dimensional Hilbert space of $n$-site chain, as
\begin{equation}
e^{ij}_x = \one_3^{\otimes (x-1)} \otimes e^{ij} \otimes \one_3^{\otimes (n-x)},
\end{equation}
satisfying the Lie algebra relations
\begin{equation}
[e_{x}^{ij},e_{x'}^{kl}]=(\delta_{j\,k}e_{x}^{i\,l}-\delta_{i\,l}e_{x}^{kj})\delta_{x,x'}.
\end{equation}
The spin-$1$ Lai--Sutherland model \cite{S75} for a chain of $n$ sites is given by the Hamiltonian  
\begin{equation}
H=\sum_{x=1}^{n-1}h_{x,x+1}, \quad h_{x,x+1}= \vec{s}_{x}\cdot \vec{s}_{x+1}+(\vec{s}_{x}\cdot \vec{s}_{x+1})^{2} - \one,
\label{Lai--Sutherland}
\end{equation}
where $\vec{s}_x=(s^{1}_x,s^{2}_x,s^{3}_x)$, with 
\begin{equation}
\!\!\!\!\!\!\!\!\!\!\!\!\!\!
s_x^1 = \frac{1}{\sqrt{2}}(e_x^{12}+e_x^{21}+e_x^{23}+e_x^{32}),\; s^2_x = \frac{\ii}{\sqrt{2}}(e_x^{21}-e_x^{12}+e_x^{32}-e_x^{23}),\; s^3_x=e_x^{11}-e_x^{33},
\quad
\end{equation}
form independent spin-$1$ variables (local $s=1$ representations of $\mathfrak{su}_2$) satisfying
\begin{equation}
[s_{x}^{i},s_{x'}^{j}]=\ii\sum_k \epsilon_{ijk}s_{x}^{k}\delta_{x,x'}.
\end{equation}
Straightforward inspection shows that the local Hamiltonian $h_{x,x+1}$ -- the interaction -- is in fact just the permutation operator between neighbouring sites 
\begin{equation}
h_{x,x+1} = \sum_{i,j=1}^3 \one_3^{\otimes(x-1)}\otimes \ket{i,j}\bra{j,i}\otimes \one_3^{\otimes (n-x-1)}= \sum_{i,j=1}^3 e^{ij}_x e^{j\,i}_{x+1}.
\end{equation}
The local Hilbert state basis is therefore given by a triple of states $\ket{1}\equiv \ket{\uparrow},\ket{2}\equiv \ket{0},\ket{3}\equiv \ket{\downarrow}$,
which can be interpreted as three different particle components (or colours); respectively, as {\em spin-up} (red) particles,  {\em spin zero} or {\em holes} (green), and {\em spin-down} (blue) particles.

\begin{figure}
 \centering	
\vspace{-1mm}
\includegraphics[width=0.7\columnwidth]{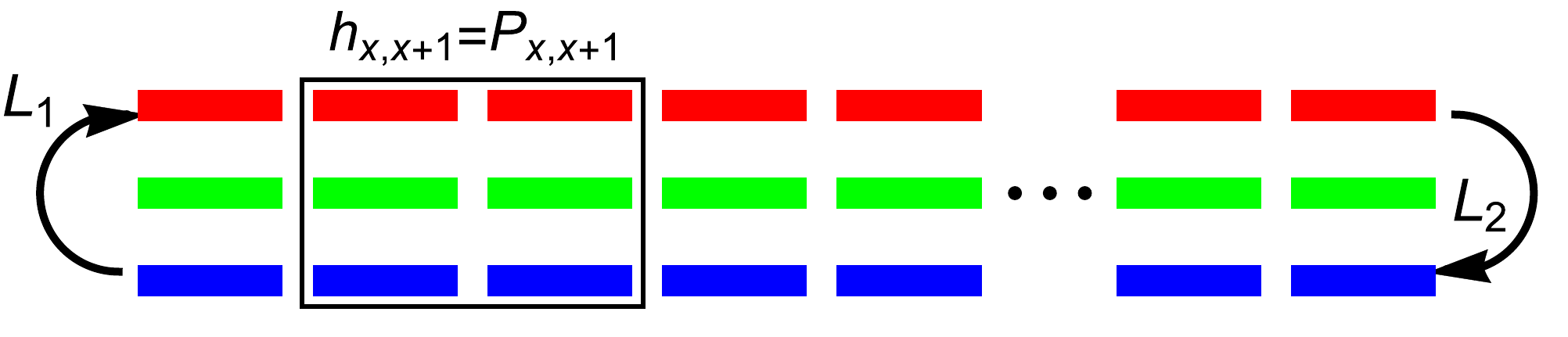}
\vspace{-1mm}
\caption{Schematic illustration of the degenerate boundary driven 3-color quantum chain, or spin-1 Lai-Sutherland chain.
The incoherent boundary jump processes transform just a pair of colour states (red and blue), while the green states remain invariant and hence the number of green particles is a constant of motion.}
\label{fig:LS}
\end{figure}

Since Lai--Sutherland chain is a multi-colour quantum model one may associate with it a skew-symmetric tensor of particle currents, with two-site density
\begin{equation}
J^{ij} = \ii (e^{ij}\otimes e^{j\,i} - e^{j\,i}\otimes e^{ij}),\quad J^{ij}_x = \one^{\otimes (x-1)}_3 \otimes J^{ij} \otimes \one^{\otimes(n-1-x)}_3 = -J^{j\,i}_x,
\quad
\end{equation}
which, by construction, satisfies the following continuity equation
\begin{equation}
\frac{\dd}{\dd t} (e^{i\,i}_x - e^{jj}_x) = \ii [H,e^{i\,i}_x - e^{jj}_x] = J^{ij}_{x-1,x} - J^{ij}_{x,x+1}.
\label{continuity_equation}
\end{equation}
$J^{ij}$ can be considered as a partial current of the particles of colour $i$ into particles of colour $j$.
The total current of component (colour) $i$,
\begin{equation}
J^i = \sum_{j=1}^3 J^{ij},
\label{totalJ}
\end{equation}
then fulfills the continuity equation
\begin{equation}
\frac{\dd}{\dd t} e^{i\,i}_x = J^i_{x-1,x}  - J^i_{x,x+1},
\end{equation}
where $e^{i\,i}_x$ can be considered as the operator of particle density of colour $i$.

We shall now open the Lai--Sutherland chain and couple it to the environment via Markovian processes which act only on local quantum spin spaces at the boundary, via the following Lindblad jump operators
\begin{equation}
\!\!\!\!\!\!\!\!\!\!\!\!\!\!\!\!\!\!\!\!\!\!\!
L_{1}= \sqrt{\varepsilon}e_{1}^{13} = \frac{\sqrt{\varepsilon}}{2} (s^+_1)^2,\quad L_{2}= \sqrt{\varepsilon}e_{n}^{31} = \frac{ \sqrt{\varepsilon}}{2} (s^-_n)^2,\;\;{\rm where}\;\;
s^\pm_x :=  s^1_x \pm \ii s^2_x.
\label{driving}
\end{equation} 
Two dissipation channels, interpreted as the left and right magnetization bath, perform the incoherent processes 
$\ket{\uparrow}\to\ket{\downarrow}$ and $\ket{\downarrow}\to\ket{\uparrow}$, respectively, with the rates $\varepsilon$. Both processes keep the hole state $\ket{0}$ unaffected.
Since also the bulk dynamics generated by $\LL_0$ conserves the number of particles of each colour, it follows that the whole Liouvillian dynamics (master equation) preserves the number of holes.
More precisely, defining the hole-number operator $N_0 \in \mathfrak{F}$ as
\begin{equation}
N_0 \ket{i_1,i_2,\ldots,i_n} = \left(\sum_{x=1}^n \delta_{i_x,2}\right) \ket{i_1,i_2,\ldots,i_n},
\end{equation}
we have that the set of all, Hamiltonian and jump operators, commute with $N_0$
\begin{equation}
[H,N_0] = 0,\quad [L_{1,2},N_0] = 0,
\end{equation}
which implies that $N_0$ generates a {\em strong} \cite{BP12} $U(1)$ symmetry of the Liouvillian flow (\ref{eq:lindeq}).
$N_0$ foliates the physical Hilbert space into $n+1$ orthogonal $N_0$-eigenspaces, ${\cal H}_n = \bigoplus_{\nu=0}^n {\cal H}^{(\nu)}_{n}$,
$N_0 {\cal{H}}^{(\nu)}_n= \nu {\cal H}^{(\nu)}_n$. The theorem A.1 of Ref.~\cite{BP12} then guarantees that the full Lindblad dynamics (\ref{eq:lindeq})
is closed on $\mathfrak{F}^{(\nu)}=\End({\cal H}^{(\nu)}_n)$, $\LL^{(\nu)} = \LL|_{\mathfrak{F}^{(\nu)}}$,
and that a {\em fixed point} $\rho^{(\nu)}_\infty = \lim_{t\to \infty} \exp(t \LL^{(\nu)}) \rho^{(\nu)}_0$  (NESS) exists for each symmetry subspace flow
\begin{equation}
\LL^{(\nu)} \rho^{(\nu)}_\infty = -\ii [H, \rho^{(\nu)}_{\infty}] + (\DD_{L_1}+\DD_{L_2})\rho^{(\nu)}_{\infty}= 0. \label{fixed_point}
\end{equation}
The theorem by Evans and Frigeiro \cite{evans,frigeiro} can then again be used to show uniqueness of NESS $\rho^{(\nu)}_\infty$ for each fixed $\nu$.
In the following we shall outline a simple algebraic procedure for actual explicit construction of density operators $\rho^{(\nu)}_\infty$.

\subsection{Degenerate matrix product solution}

Let $\PP^{(\nu)}\in\End(\mathfrak{F})$ be a  projector to $\mathfrak{F}^{(\nu)}$, orthogonal with respect to Hilbert-Schmidt inner product $(A|B)=\tr A^\dagger B$, with respect to which the Weyl basis $e^{i_1 j_1}\otimes e^{i_2 j_2}\cdots$ is orthonormal. We define a {\em grand} density matrix of NESS as a direct sum of non-trivial solutions of (\ref{fixed_point}) for all $\nu$,
\begin{equation}
\rho_\infty = \sum_{\nu=0}^{n} \rho^{(\nu)}_\infty,\quad{\rm with}\quad\rho^{(\nu)}_\infty= \PP^{(\nu)}\rho_\infty \neq 0,
\label{universal}
\end{equation}
being solution of the fixed point equation (\ref{fixed_point}) as well. The grand state $\rho_\infty$ shall be sought for in terms of Cholesky factorization
(in analogy to previous solutions of $XXZ$  and Hubbard models)
\begin{equation}
\rho_{\infty}(\varepsilon)=\Omega_n(\varepsilon)\Omega^{\dagger}_n(\varepsilon),
\label{factorization}
\end{equation}
where $\Omega_{n}(\varepsilon)\in \End({\cal H}_n)$ is some yet unknown operator which is represented by an upper triangular matrix in the computational basis $\ket{i_1,\ldots,i_n}$.
Introducing an auxiliary Hilbert space ${\cal H}_{\rm a}$ -- separable, but of infinite-dimensionality as will become clear later -- 
we define the monodromy operator
$\bb{M}(\varepsilon)\in \End{({\cal H}_n\otimes{\cal H}_\aa)}$ as a spatially-ordered product of some local Lax operators
$\bb{L}_{x}(\varepsilon)\in \End{({\cal H}_n\otimes{\cal H}_\aa)}$,
\begin{equation}
\bb{M}(\varepsilon) =\bb{L}_{1}(\varepsilon)\bb{L}_{2}(\varepsilon)\cdots \bb{L}_{n}(\varepsilon).
\end{equation}
Index free Lax operator is defined as $\bb{L}(\varepsilon) \in \End({\cal H}_{\rm p}\otimes{\cal H}_{\rm a})$ so that one writes 
$\bb{L}_{x}(\varepsilon)=\one_{3}^{\otimes (x-1)} \otimes \bb{L}(\varepsilon)\otimes \one_{3}^{\otimes (n-x)}$. 
Furthermore, we define the components of Lax matrix $\bb{L}^{ij}(\varepsilon) \in  \End({\cal H}_{\rm a})$, such that
\begin{equation}
\bb{L}_{x}(\varepsilon)=\sum_{i,j=1}^{3}e_{x}^{ij}\otimes \bb{L}^{ij}(\varepsilon),\qquad
\bb{L}(\varepsilon)=\sum_{i,j=1}^{3}e^{ij}\otimes \bb{L}^{ij}(\varepsilon).
\label{L_operator}
\end{equation}
We further assume existence of a special state $\ket{\rm vac}\in{\cal H}_{\rm a}$, such that Cholesky factor writes as the auxiliary expectation value of monodromy operator, or equivalently, as MPA
\begin{equation}
\!\!\!\!\!\!\!\!\!\!\!\!\!\!
\Omega_n = \bra{\rm vac}\bb{M}\ket{\rm vac} = \sum_{i_1,j_1\ldots i_n,j_n} \bra{\rm vac}\bb{L}^{i_1 j_1}\cdots \bb{L}^{i_n j_n}\ket{\rm vac} e^{i_1 j_1}\otimes \cdots \otimes e^{i_n j_n}.
\label{MPO}
\end{equation}
Fixing an arbitrary, fixed orthonormal basis $\{\ket{\psi_k}\}$ of ${\cal H}_{\rm a}$ we define the conjugate Lax matrices $\ol{\bb{L}}(\varepsilon)$ by 
$\bra{\psi_k}\ol{\bb{L}}^{ij}(\varepsilon)\ket{\psi_l} := \ol{\bra{\psi_k}\bb{L}^{ij}(\varepsilon)\ket{\psi_l}}$.
For notational convenience we denote the second copy of auxiliary space carrying conjugate representation of $\ol{\bb{L}}^{ij}$ as $\ol{{\cal H}}_{\rm a}$.
One can then write MPA for NESS density operator $\rho_{\infty}$ directly, by introducing {\em two-leg Lax matrices} $\vmbb{L}^{ij}(\varepsilon)\in\End({\cal H}_{\rm a}\otimes \ol{{\cal H}}_{\rm a})$,
and $\vmbb{L}_x(\varepsilon) \in \End({\cal H}_n\otimes{\cal H}_\aa\otimes \ol{{\cal H}}_{\rm a})$ as
\begin{equation}
\vmbb{L}^{ij}(\varepsilon) = \sum_{k} \bb{L}^{ik}(\varepsilon)\otimes \ol{\bb{L}}^{jk}(\varepsilon),\quad \vmbb{L}_x(\varepsilon) = \sum_{i,j} e^{ij}_x \otimes \vmbb{L}^{ij}(\varepsilon),
\label{double_L}
\end{equation}
namely
\begin{equation}
\rho_{\infty}(\varepsilon)=\llvac\vmbb{M}(\varepsilon)\rrvac.
\label{rhoMM}
\end{equation}
Note the transposition in the quantum space of the conjugated factor of \eref{double_L}.
Here a two-leg monodromy operator 
\begin{equation}
\vmbb{M}(\varepsilon) = \vmbb{L}_1(\varepsilon)\cdots  \vmbb{L}_n(\varepsilon) \in \End({\cal H}_n\otimes{\cal H}_{\rm a}\otimes \ol{{\cal H}}_{\rm a}),
\label{MM}
\end{equation}
and a product of a pair of vacua $\llvac=\lvac \otimes \lvac$, $\rrvac=\rvac \otimes \rvac$ have been introduced, so that (\ref{rhoMM}) is merely a formal rewriting of (\ref{factorization}).
These definitions become particularly handy when we consider evaluation of expectation values of local observables with respect to NESS $\rho_{\infty}(\varepsilon)$.

Let $\eta:=\ii \varepsilon$ be a complex-rotated coupling parameter and let us (for convenience) relabel the quantum space matrix elements of the $\bb{L}$-operator as
\begin{eqnarray}
\label{relabeling}
\bb{L}=\pmatrix{
\bb{l}^{\ua} & \bb{t}^{+} & \bb{v}^{+} \cr
\bb{t}^{-} & \bb{l}^{0} & \bb{u}^{+} \cr
\bb{v}^{-} & \bb{u}^{-} & \bb{l}^{\da}
}.
\end{eqnarray}
Explicit MPA structure of the grand density operator of NESS is then established by the following result \cite{IP14}:
\begin{thm} \cite{IP14} (i) Suppose that 9 matrix elements $\{\bb{L}^{ij}\}$ generate the Lie algebra $\mathfrak{g}$ defined by commutation relations,
\begin{eqnarray}
\label{algebra}
&& [\bb{u}^{+},\bb{t}^{\pm}]=[\bb{u}^{-},\bb{t}^{\pm}]=[\bb{u}^{\pm},\bb{v}^{\pm}]=[\bb{t}^{\pm},\bb{v}^{\pm}]=0,\cr
&& [\bb{l}^{\ua},\bb{u}^{\pm}]=[\bb{l}^{\da},\bb{t}^{\pm}]=[\bb{l}^{\ua},\bb{l}^{\da}]=0,\cr
&& [\bb{l}^{\ua},\bb{t}^{\pm}]=\mp \eta \bb{t}^{\pm},\qquad [\bb{l}^{\da},\bb{u}^{\pm}]=\mp \eta \bb{u}^{\pm},\cr
&& [\bb{u}^{\pm},\bb{v}^{\mp}]=\pm \eta \bb{t}^{\mp},\qquad [\bb{t}^{\pm},\bb{v}^{\mp}]=\pm \eta \bb{u}^{\mp}, \cr
&&  [\bb{l}^{\ua},\bb{v}^{\pm}]=[\bb{l}^{\da},\bb{v}^{\pm}]=\mp \eta \bb{v}^{\pm},\quad [\bb{v}^{+},\bb{v}^{-}]=\eta(\bb{l}^{\ua}+\bb{l}^{\da}),\cr
&& [\bb{t}^{+},\bb{t}^{-}]=[\bb{u}^{+},\bb{u}^{-}]=\eta \bb{l}^{0},\cr
&& [\bb{l}^{\ua,\da},\bb{l}^{0}] = [\bb{u}^\pm,\bb{l}^{0}] = [\bb{v}^\pm,\bb{l}^{0}] = [\bb{t}^\pm,\bb{l}^{0}] = 0,
\end{eqnarray}
with a representation over the Hilbert space ${\cal H}_{\rm a}$ satisfying the following conditions
\begin{eqnarray}
\label{boundary_conditions}
&\bb{l}^{\ua}\rvac = \bb{l}^{0}\rvac = \bb{l}^{\da}\rvac = \rvac,\cr
&\lvac \bb{l}^{\ua} = \lvac \bb{l}^{0} = \lvac \bb{l}^{\da} = \lvac,\cr
&\bb{t}^+ \rvac = \bb{u}^+ \rvac = \bb{v}^+ \rvac = 0,\cr
&\lvac \bb{t}^- = \lvac \bb{u}^- = \lvac\bb{v}^- = 0.
\end{eqnarray}
Then, the grand state solution (\ref{universal}) to NESS fixed point condition \eref{fixed_point} is given via Cholesky factorization \eref{factorization} with explicit MPA \eref{MPO} for $\Omega_n(\varepsilon)$ with $\eta=\ii\varepsilon$.

(ii) A possible irreducible explicit representation of Lie algebra $\mathfrak{g}$ \eref{algebra} satisfying \eref{boundary_conditions} is given as 
\begin{eqnarray}
&& \bb{t}^{+}=\bb{b}_{\ua},\quad \bb{t}^{-}=\eta \bb{b}^\dagger_{\ua}, \cr
&& \bb{u}^{+}=\eta \bb{b}_{\da},\quad \bb{u}^{-}=\bb{b}^\dagger_{\da},\cr
&& \bb{v}^{+}=\eta(\bb{b}_{\ua}\bb{b}_{\da}+\bb{s}^{+}),\quad \bb{v}^{-}=\eta(\bb{b}^\dagger_{\ua}\bb{b}^\dagger_{\da}-\bb{s}^{-}),\cr
&& \bb{l}^{\ua,\da}=\eta\left(\bb{b}^\dagger_{\ua,\da}\bb{b}_{\ua,\da}+\frac{1}{2}-\bb{s}^{\rm z}\right),\quad \bb{l}^{0}=\one,
\label{algebra_generators}
\end{eqnarray}
in terms of three auxiliary degrees of freedom with a three dimensional lattice $\{\ket{j,k,l},\; j,k,l\in\ZZ^+\}$ forming a basis of ${\cal H}_{\rm a}$, namely, two bosonic modes $\bb{b}_{\ua,\da}$
\begin{eqnarray}
\bb{b}^\dagger_{\ua}\ket{j,k,l}= \sqrt{j+1}\ket{j+1,k,l},\quad \bb{b}_{\ua}\ket{j,k,l}=\sqrt{j}\ket{j-1,k,l},\cr
\bb{b}^\dagger_{\da}\ket{j,k,l}= \sqrt{k+1}\ket{j,k+1,l},\quad \bb{b}_{\da}\ket{j,k,l}=\sqrt{k}\ket{j,k-1,l},
\label{bosons}
\end{eqnarray}
and a complex spin (Verma module of $\mathfrak{sl}_2$)
\begin{eqnarray}
\bb{s}^{+} \ket{j,k,l} &=& l \ket{j,k,l-1},\cr 
\bb{s}^-\ket{j,k,l} &=& (2p-l) \ket{j,k,l+1},\cr 
\bb{s}^{\rm z}\ket{j,k,l} &=& (p-l)\ket{j,k,l}.
\label{verma}
\end{eqnarray}
with $\rvac = \ket{0,0,0}$ being the highest-weight-state.
The complex spin parameter $p$ should be linked to dissipation parameter via
\begin{equation}
p = \frac{1}{2} - \frac{1}{\eta} = \frac{1}{2} + \frac{\ii}{\varepsilon}.
\label{pfix}
\end{equation}
\end{thm}

\begin{proof}
The proof is based on verifying that the Lie algebra $\mathfrak{g}$, given by \eref{algebra}, can be equivalently defined by means of an appropriate Sutherland relation  
\begin{equation}
[h_{x,x+1},\bb{L}_{x}\bb{L}_{x+1}]=B_{x}\bb{L}_{x+1} -\bb{L}_{x}B_{x+1},
\label{Sutherland}
\end{equation}
with the-so-called {\em boundary operator} $B_{x}(\varepsilon) \in \End{({\cal H}_n\otimes{\cal H}_\aa)}$ -- operating non-trivially only in the local quantum space
\begin{equation}
B_{x}=\eta \left(e_{x}^{33} \otimes \one_{\rm a} - e_{x}^{11}\otimes \one_{\rm a}\right)= b_x \otimes \one_{\rm a},  
\label{B}
\end{equation}
where $b_x(\varepsilon) = -\ii\varepsilon s^{3}_x \in \mathfrak{F}$.
Identification of \eref{algebra} with the Sutherland relation \eref{Sutherland} is straightforward, based solely on the permutation action of Hamiltonian density 
\begin{equation}
[h_{x,x+1},e^{ij}_x e^{kl}_{x+1}] = e^{kj}_x e^{i\,l}_{x+1} - e^{i\,l}_x e^{kj}_{x+1}.
\end{equation}
Multiplying the Sutherland relation by a string $\bb{L}_1\cdots \bb{L}_{x-1}$ from the left and a string $\bb{L}_{x+2}\cdots \bb{L}_n$ from the right, summing over $x$
and taking vacuum expectation value yields the defining relation for the amplitude operator
\begin{equation}
[H,\Omega_{n}]=-\ii\varepsilon\left(s^{3}\otimes \Omega_{n-1}-\Omega_{n-1}\otimes s^{3}\right),
\label{defining_relation}
\end{equation}
where $s^3 = e^{11}-e^{33}$. Consequently, by expanding the unitary part of Liouvillian $\LL_{0}$,
\begin{equation}
-\LL_{0}(\rho_{\infty})\equiv \ii[H,\rho_{\infty}]=\ii[H,\Omega_{n}]S^{\dagger}_{n}-\ii \Omega_{n}[H,\Omega_{n}]^{\dagger},
\end{equation}
in conjunction with \eref{defining_relation}, and employing the definition \eref{double_L}, the steady state condition \eref{fixed_point} yields
a decoupled system of \textit{boundary equations} 
\begin{eqnarray}
& \llvac\left( \DD_{A_1}(\vmbb{L}_{1}) - \ii(\vmbb{B}^{(1)}_{1}-\vmbb{B}^{(2)}_{1})\right) = 0,\nonumber\\
& \left(\DD_{A_2}(\vmbb{L}_{n}) + \ii(\vmbb{B}^{(1)}_{n}-\vmbb{B}^{(2)}_{n})\right)\rrvac = 0,
\label{boundary_system}
\end{eqnarray}
where {\em two-leg boundary operators} $\vmbb{B}^{(1)}_{x},\vmbb{B}^{(2)}_{x}\in \End({\cal H}_n\otimes{\cal H}_\aa\otimes \ol{{\cal H}}_{\rm a})$, reading 
\begin{equation}
\vmbb{B}^{(1)}_{x}= \sum_{i,j=1}^3 b_x e^{ij}_x \otimes \one_{\rm a}\otimes \ol{\bb{L}}^{j\,i},\qquad 
\vmbb{B}^{(2)}_{x}= \sum_{i,j=1}^3 e^{ij}_x \ol{b}_x \otimes \bb{L}^{ij} \otimes \one_{\rm a},
\end{equation}
have been defined. Note that, due to \eref{B}, $\ol{b}_x = \ii \varepsilon s^3_x = -b_x$ for $\varepsilon\in\RaR$.

The last two lines of  \eref{algebra} indicate that pairs of auxiliary operators $(\bb{t}^{+},\bb{t}^{-})$ and $(\bb{u}^{+},\bb{u}^{-})$ span the Weyl-Heisenberg algebra.
In conjunction with the highest weight conditions \eref{boundary_conditions} this fixes the representation of  $(\bb{t}^{+},\bb{t}^{-})$ and $(\bb{u}^{+},\bb{u}^{-})$ to be that of a
 Fock space of two oscillator modes (bosons), specified by creation/annihilation operators,
$[\bb{b}_{\sigma},\bb{b}^{\dagger}_{\sigma'}]=\delta_{\sigma,\sigma'}$, $[\bb{b}_{\sigma},\bb{b}_{\sigma'}]=0$, $\sigma,\sigma'\in \{\ua,\da\}$, suggesting that the auxiliary space ${\cal H}_{\rm a}$ is perhaps just a two-mode boson Fock space.
While realization for all the other generators consistent with the bulk algebra $\mathfrak{g}$ is not difficult to construct (e.g. $\bb{v}^\pm$, $\bb{l}^\ua+\bb{l}^\da$ can be just the Schwinger boson representation of $\mathfrak{su}_2$ -- see 5th line of \eref{algebra}),
it turns out not to be consistent with the boundary conditions \eref{boundary_conditions}.Therefore the auxiliary space ${\cal H}_{\rm a}$ has to contain (at least) one additional degree of freedom.

Ultimately, in order to fulfil \eref{boundary_system}, a straightforward calculation shows that it is enough to add a Verma module ${\cal V}_p$ of complex spin representation \eref{verma} of $\mathfrak{sl}_2$ and consider a triple-product space ${\cal H}_{\rm a}\cong {\cal B}\otimes {\cal B}\otimes {\cal V}_p={\rm lsp}\{ \ket{j,k,l}; j,k,l\in\ZZ^+\}$,
and find a representation of the algebra \eref{algebra} which is compliant with conditions
\begin{eqnarray}
\label{Lax_boundary}
\bb{L}\rvac =
\pmatrix{
\rvac & 0 & 0 \cr
\eta \ket{1,0,0} & \rvac & 0 \cr
\eta(\ket{1,1,0}-\ket{0,0,1})+2\ket{0,0,1} & \ket{0,1,0} & \rvac},\\
\lvac \bb{L} =
\pmatrix{
\lvac & \bra{1,0,0} & \eta(\bra{1,1,0}+\bra{0,0,1}) \cr
0 & \lvac & \eta \bra{0,1,0} \cr
0 & 0 & \lvac},
\end{eqnarray}
with vacuum being given by the ground state $\rvac \equiv \ket{0,0,0}$.
These requirements are all satisfied by choosing representation (\ref{algebra_generators},\ref{bosons},\ref{verma}) with $p$ being fixed (\ref{pfix}) as
required by the conditions in the first two lines of (\ref{boundary_conditions}). The last two lines of (\ref{boundary_conditions}) are satisfied due to highest-weight-property of $\rvac$. As such a representation is clearly irreducible, this concludes the proof .\end{proof} 

\subsection{Grand canonical NESS and observables}
 
The formulae (\ref{factorization},\ref{MPO},\ref{relabeling}-\ref{pfix}) yield explicit construction of a many-body density matrix of a family of degenerate NESSes $\rho^{(\nu)}_\infty = \PP^{(\nu)}\rho_\infty$ for any number of holes $\nu\in\{0,1\ldots n\}$.
The computational complexity of obtaining any local information about the state $\rho_\infty$, say to compute its matrix elements
of the type $\bra{i_1,\ldots,i_n}\rho_\infty\ket{j_1,\ldots,j_n}$ or local observables, is at most {\em polynomial} in $n$.
Since the eigenspaces ${\cal H}^{(\nu)}$ of number-of-holes operator $N_0$ are orthogonal, one can also split decompose the Cholesky factors
$\Omega_n^{(\nu)}(\varepsilon) = \PP^{(\nu)} \Omega_n(\varepsilon)$ 
\begin{equation}
\rho^{(\nu)}_\infty(\varepsilon) = \Omega^{(\nu)}_n(\varepsilon)\,\Omega^{(\nu)\dagger}_n(\varepsilon),
\end{equation}
since $\Omega^{(\nu)}\Omega^{(\nu')\dagger}= 0\; {\rm if}\; \nu\neq\nu'$.
Projected Cholesky factor satisfies a {\em projected} defining relation (\ref{defining_relation})
\begin{equation}
[H,\Omega^{(\nu)}_{n}]=-\ii\varepsilon \left(s^{3}\otimes \Omega^{(\nu)}_{n-1}-\Omega^{(\nu)}_{n-1}\otimes s^{3}\right),
\end{equation}
and can be expressed in terms of a constrained or {\em microcanonical} MPA
\begin{equation}
\!\!\!\!\!\!\!\!\!\!\Omega^{(\nu)}_n  = \sum_{i_1,j_1\ldots i_n,j_n} \delta_{\left(\sum_x \delta_{i_x,2}\right),\nu} \bra{\rm vac}\bb{L}^{i_1 j_1}\cdots \bb{L}^{i_n j_n}\ket{\rm vac} e^{i_1 j_1}\otimes \cdots \otimes e^{i_n j_n}.
\label{mMPO}
\end{equation}
Note that since $[\Omega^{(\nu)},N_0]=0$, the Kronecker-$\delta$ constraint can just as well be replaced by $\delta_{\left(\sum_x \delta_{j_x,2}\right),\nu}$ as only operators 
$e^{i_1 j_1}\otimes \cdots \otimes e^{i_n j_n}$ for which $\sum_x\delta_{i_x,2}=\sum_x\delta_{j_x,2}$ appear in MPA (\ref{MPO}).

We note two limiting cases of our NESS solution. For zero hole sector $\nu=0$ one obtains exactly the fully polarized boundary driven isotropic ($XXX$) Heisenberg spin-$1/2$ chain and reproduces the solution reported in subsect.~\ref{subsect:XXX}. 
The other extreme case ($\nu=n$) is the so-called \textit{dark state}, i.e. a pure state  $\rho^{(\nu=n)}_{\infty}=(e^{22})^{\otimes n} = \ket{2,2\ldots 2}\bra{2,2,\ldots 2}$ which
is unaffected by the dissipation, i.e. it simultaneously annihilated by $\LL_{0}$ and $\DD$, $\LL_{0}\rho^{(n)}_{\infty}=\DD\rho^{(n)}_{\infty}=0$.

Any convex mixture of states $\rho_\infty = \sum_\nu c_\nu \rho^{(\nu)}_\infty$, $c_\nu\in\RaR^+$, is a valid NESS density operator as well, which factorizes (\ref{factorization}) with a Cholesky factor $\Omega_n = \sum_\nu \sqrt{c_\nu} \Omega^{(\nu)}_n$.
Microcanonical constraint in (\ref{mMPO}) seems cumbersome as it prevents facilitating transfer matrices for computation of local observables.
There seems to be a particularly attractive option which overcomes this problem. Namely, one may define a {\em grand canonical nonequilbrium steady state} (gcNESS) ensemble by taking a {\em hole chemical potential} $\mu$ with $c_\nu = \exp(\mu \nu)$:
\begin{equation}
\rho_{\infty}(\varepsilon,\mu)=\sum_{\nu=0}^{n}\exp{(\mu \nu)}\,\rho^{(\nu)}_{\infty}(\varepsilon).
\label{grand_canonical_ens}
\end{equation}
Note that the grand state corresponds to gcNESS with zero chemical potential $\rho_\infty(\varepsilon) = \rho_\infty(\varepsilon,\mu=0)$.
Clearly, the addition theorem for exponential function erases the constraint in MPO expansions:
\begin{eqnarray}
\!\!\!\!\!\!\!\!\!\!\!\!\!\!\!\!\!\!\!\!\!\!\!\!
\Omega_n(\varepsilon,\mu) = \sum_{i_1,j_1\ldots i_n,j_n} \bra{\rm vac}\bb{L}^{i_1 j_1}(\varepsilon,\mu)\cdots \bb{L}^{i_n j_n}(\varepsilon,\mu)\ket{\rm vac} e^{i_1 j_1}\otimes \cdots \otimes e^{i_n j_n},\\
\!\!\!\!\!\!\!\!\!\!\!\!\!\!\!\!\!\!\!\!\!\!\!\!
\rho_{\infty}(\varepsilon,\mu) = \sum_{i_1,j_1\ldots i_n,j_n}\llvac\vmbb{L}^{i_1 j_1}(\varepsilon,\mu)\cdots\vmbb{L}^{i_n j_n}(\varepsilon,\mu)\rrvac e^{i_1 j_1}\otimes \cdots \otimes e^{i_n j_n},
\end{eqnarray}
where the chemical potential only modifies the components of the Lax operators as
\begin{equation}
\!\!\!\!\!\!\!\!\!\!\!\!
\bb{L}^{ij}(\varepsilon,\mu) = \exp\left(\frac{\mu}{2}\delta_{i,2}\right)\bb{L}^{ij}(\varepsilon),\quad
\vmbb{L}^{ij}(\varepsilon,\mu) = \exp{\left(\frac{\mu}{2}(\delta_{i,2}+\delta_{j,2})\right)}\vmbb{L}^{ij}(\varepsilon).
\end{equation}
Moreover, introducing a {\em transfer operator}
\begin{equation}
\vmbb{T}(\varepsilon,\mu) = \sum_{i}\vmbb{L}^{i\,i}(\varepsilon,\mu)=\sum_{i,j}\bb{L}^{ij}(\varepsilon,\mu)\otimes \ol{\bb{L}}^{ij}(\varepsilon,\mu),
\label{TVO}
\end{equation}
we define the \textit{grand canonical nonequilibrium partition function} and express it via the transfer matrix method
\begin{equation}
{\cal Z}_{n}(\varepsilon,\mu) = \tr\left(\rho_{\infty}(\varepsilon,\mu)\right) = \llvac \left(\vmbb{T}(\varepsilon,\mu)\right)^n \rrvac.
\label{NPF}
\end{equation}
The hole chemical potential $\mu$ can be connected to the ensemble averaged filling factor (doping) $r$ via logarithmic derivative of the partition function
\begin{equation}
r := \frac{\ave{\nu}}{n} = \frac{\sum_{\nu=0}^n \nu \exp(\nu \mu) \tr\rho^{(\nu)}_\infty}{n\sum_{\nu=0}^n \exp(\nu \mu) \tr\rho^{(\nu)}_\infty} = n^{-1}\partial_{\mu}\log {\cal Z}_{n}(\varepsilon,\mu).
\label{q}
\end{equation}
As usual, we expect the fluctuations $(\delta r)^2 = \ave{\frac{\nu}{n}}^2 - r^2$ to be thermodynamically small.

Expectation values of  general (local) observables can again  be extracted by facilitating the auxiliary vertex operators.
Let $X_{[x,y]}=\one_3^{\otimes (x-1)}\otimes X\otimes \one_3^{\otimes (n-y)}$ be a generic local observable supported on a sublattice between sites $x$ and $y$. Then, a formal expression
\begin{equation}
\expect{X_{[x,y]}}={\cal Z}_{n}^{-1}(\varepsilon,\mu)\;\tr{(X_{[x,y]}\rho_{\infty}(\varepsilon,\mu))},
\end{equation}
can be calculated from the MPA of $\rho_{\infty}(\varepsilon,\mu)$ by \textit{tracing out} the physical space ${\cal H}_n$ and associating
to each observable $X_{[x,y]}$ a corresponding vertex operator via a mapping $\Lambda_{\ell}:{\rm End}({\cal H}^{\otimes \ell}_{\rm 1}) \rightarrow {\rm End}({\cal H}_{\rm a}\otimes \ol{{\cal H}}_{\rm a})$, where $\ell = y-x + 1$,
using the prescription 
\begin{equation}
\Lambda_{\ell}(X)=\vmbb{X}:=\sum_{i_1,j_1\ldots i_\ell,j_\ell} \tr\left((e^{i_1 j_1}\otimes\cdots\otimes e^{i_\ell j_\ell})X\right) \vmbb{L}^{i_1 j_1}\cdots \vmbb{L}^{i_\ell j_\ell}.  
\end{equation}
For a complementary part of the lattice, i.e. where $X_{[x,y]}$ operates trivially, one has the transfer vertex operator $\vmbb{T}=\Lambda_{1}(\one_{3})$, eq. \eref{TVO},
so the final expectation value reads
\begin{equation}
\expect{X_{[x,y]}} = {\cal Z}_{n}^{-1}\llvac \vmbb{T}^{x-1}\;\vmbb{X}\;\vmbb{T}^{n-y}\rrvac.
\end{equation}
For example, for on-site observables we have auxiliary vertex operators $\Lambda_{1}(e^{ij}) = \vmbb{L}^{j\,i}$, e.g. for magnetization density 
$\Lambda_{1}(s^3) = \vmbb{L}^{11}-\vmbb{L}^{33}$.

As for two point observables, we consider an interesting example of the current density tensor
\begin{equation}
\Lambda_2(J^{ij})=\vmbb{J}^{ij}= \ii\left(\vmbb{L}^{j\,i}\vmbb{L}^{ij}\ - \vmbb{L}^{ij}\vmbb{L}^{j\,i}\right) \label{current_vertex}
\end{equation}
Stationarity (time-independence) of NESS and continuity equation \eref{continuity_equation} imply spatial-independence of current expectation values. In auxiliary transfer matrix formulation \eref{TVO} this implies commutation of transfer vertex operator with current vertex operators when projected onto the subspace of states created upon action of $\vmbb{T}$ on the vacua, namely
\begin{equation}
\bbra{\varphi^{\rm{L}}_{k}}[\vmbb{T},\vmbb{J}^{ij}]\kket{\varphi^{\rm{R}}_{l}}=0,\quad
\bbra{\varphi^{\rm{L}}_{k}}:=\llvac \vmbb{T}^{k},\quad \kket{\varphi^{\rm{R}}_{k}}:=\vmbb{T}^{k}\rrvac.
\end{equation}
Additionally, using the representation given in Theorem 2 and highest weight nature of the vacuum state, one can with some effort express the expectation values of total current operators \eref{totalJ} in terms of the nonequilibrium partition function (\ref{NPF})
\begin{equation}
\expect{J^{1}}=2\varepsilon\frac{{\cal Z}_{n-1}}{{\cal Z}_{n}},\quad\expect{J^{2}}=0,\quad\expect{J^{3}}=-2\varepsilon\frac{{\cal Z}_{n-1}}{{\cal Z}_{n}}.
\label{current_recurrence}
\end{equation}
It would be challenging to attempt an analytic asmptotic computation of the grand canonical nonequilibrium partition function ${\cal Z}_n(\varepsilon,\mu)$ along the lines described in subsection \ref{subsect:neZ}, however this has not been accomplished yet.
On the other hand, numerical computation of the expression (\ref{NPF}) strongly suggests \cite{Med}, again, the universal scaling ${\cal Z}_{n-1}/{\cal Z}_n  \sim n^{-2}$ for all (generic) values of $\varepsilon,\mu$.
 
\section{Quasilocal conservation laws and linear response physics}
\label{qlcl}
 
In this section we discuss the main `spin-off' application of the exact solutions of boundary driven nonequilibrium master equations discussed so far, namely deriving novel, so-called {\em quasilocal} conserved quantities, and consequently deepening our understanding of the linear response physics of the corresponding closed (coherent, non-dissipative) models. The main exposition is focusing on the paradigmatic $XXZ$ model,
following Refs.~\cite{P14c,PI13},  while some comments with regard to other integrable chains will be given at the end.

\subsection{Universal $R$-matrix and exterior integrability of NESS density operator}

Using the concept of a universal $R$-matrix (see e.g. Refs.~\cite{KT91,F94,K95,K01,K02,klumper1,klumper2}), one may find and explicitly construct the solution of YBE with $U_q(\mathfrak{sl}_2)$ symmetry over an arbitrary triple tensor product of 
highest-weight Verma modules (\ref{verma})
${\cal V}_{s_1}\otimes {\cal V}_{s_2}\otimes {\cal V}_{s_3}$, 
\begin{equation}
\mm{R}_{s_1,s_2}(\varphi-\vartheta) \mm{R}_{s_1,s_3}(\varphi)\mm{R}_{s_2,s_3}(\vartheta) = \mm{R}_{s_2,s_3}(\vartheta)\mm{R}_{s_1,s_3}(\varphi)\mm{R}_{s_1,s_2}(\varphi-\vartheta)
\label{UYBE}
\end{equation}
for arbitrary representation parameters $s_1,s_2,s_3\in\CC$ and {\em additive} spectral parameters $\varphi,\vartheta\in\CC$, where $\mm{R}_{s_i,s_j}$ acts nontrivially on the $i$-th and $j$-th module of the triple.
For example $\mm{R}_{\frac{1}{2},\frac{1}{2}}(\varphi) = \mm{P} \mm{R}(\varphi)$ is the six-vertex $R$-matrix of Eq.~(\ref{eq:Rff}) up to permutation $\mm{P}$ of the pair of auxiliary spaces, while $\mm{R}_{\frac{1}{2},s}(\varphi) = \mm{L}(\varphi,s)$ is a generic non-compact (non-Hermitian) Lax operator (\ref{eq:Lax}) so that YBE (\ref{UYBE}) for $s_1=s_2=\frac{1}{2},s_3=s$ becomes the $RLL$ relation (\ref{eq:RLL}).

Taking the first two spaces as auxiliary and the third one as physical, ${\cal V}_{s}\otimes {\cal V}_{t}\otimes{\cal V}_{\frac{1}{2}} \equiv {\cal H}_\aa\otimes{\cal H}_\bbb\otimes{\cal H}_{\rm p}$, i.e., $s_1=s,s_2=t,s_3=\frac{1}{2}$,
and writing an infinite-dimensional, so-called {\em exterior} $R$-matrix as $\mm{R}_{\aa,\bbb}(\varphi,s,t)\equiv \mm{P}_{\aa,\bbb} \mm{R}_{s,t}(\varphi)$, Eq.~(\ref{UYBE}) yields an alternative $RLL$ relation swapping 
spectral {\em and} representation parameters in auxiliary tensor product of Lax operators
\begin{equation}
\!\!\!\!\!\!\!\!\!\!\!\!\!\!\!\!\!\mm{R}_{\aa,\bbb}(\varphi-\vartheta,s,t) \mm{L}_{\aa,x}(\varphi,s) \mm{L}_{\bbb,x}(\vartheta,t) = \mm{L}_{\aa,x}(\vartheta,t) \mm{L}_{\bbb,x}(\varphi,s) \mm{R}_{\aa,\bbb}(\varphi-\vartheta,s,t).
\label{RLLa}
\end{equation}
Moreover, due to $U_q(\mathfrak{sl}_2)$ invariance, the tensor product of highest-weight states, spanning a scalar representation, should always be left and right invariant under the $R$-matrix
\begin{equation}
\mm{R}_{\aa,\bbb} (\varphi,s,t) \ket{0}_\aa\ket{0}_\bbb =  \ket{0}_\aa\ket{0}_\bbb,\quad
\bra{0}_\aa\bra{0}_\bbb\mm{R}_{\aa,\bbb}(\varphi,s,t) = \bra{0}_\aa\bra{0}_\bbb.
\label{Rfix}
\end{equation}

Remarkably, the relation (\ref{RLLa}) together with (\ref{Rfix}) immediately implies commutativity of a two-parametric, {\em highest-weight non-Hermitian transfer operator} (HNTO) $W_n(\varphi,s) \in{\rm End}({\cal H}^{\otimes n}_{\rm p})$
\begin{equation}
W_n(\varphi,s) := \bra{0}\mm{L}(\varphi,s)^{\otimes n}\ket{0},\quad [W_n(\varphi,s),W_n(\vartheta,t)]=0,
\label{HNTO}
\end{equation}
namely
\begin{eqnarray}
&&\!\!\!\!\!\!\!\!\!\!\!\!\!\!\!\!\!\!\!\!\!\!\!W_n(\varphi,s)W_n(\vartheta,t) = \bra{0}_\aa \bra{0}_\bbb\mm{R}_{\aa,\bbb}(\varphi-\vartheta,s,t)\left(\prod_{x=1}^n \mm{L}_{\aa,x}(\varphi,s)\mm{L}_{\bbb,x}(\vartheta,t)\right)\ket{0}_\aa\ket{0}_\bbb \nonumber\\
&&\!\!\!\!\!\!\!\!\!\!\!\!\!\!\!\!\!\!\!\!\!\!\!= \bra{0}_\aa \bra{0}_\bbb\left(\prod_{x=1}^n\mm{L}_{\aa,x}(\vartheta,t)\mm{L}_{\bbb,x}(\varphi,s)\right)\mm{R}_{\aa,\bbb}(\varphi-\vartheta,s,t) \ket{0}_\aa\ket{0}_\bbb = W_n(\vartheta,t)W_n(\varphi,s). \nonumber
\end{eqnarray}
Such a transfer matrix is essentially just the amplitude operator of NESS for a general asymmetric driving (\ref{eq:OmegaXXZ}), specifically $\Omega_n(\varphi,s,\chi) = W^T_n(\varphi,s)K(\sqrt{\chi})^{\otimes n}$ where the second, diagonal factor is inessential since it commutes with HNTO, 
$[W^T_n(\varphi,s),K(\sqrt{\chi})]\equiv 0$. 

One may thus define \cite{PIP} an {\em exterior integrability} of a nonequilibrium many-body density operator $R_\infty = \Omega_n(\varphi,s,\chi)[\Omega_n(\varphi,s,\chi)]^\dagger$ if Cholesky factor
$ \Omega_n(\varphi,s,\chi)$ forms a commuting family for any values of driving parameters. Note, however, that as a manifestation of {\em non-normality} of the transfer operator $W_n(\varphi,s)$, $[W_n(\varphi,s),W^T_n(\vartheta,t)]\neq 0$ in general, hence the density operator $R_\infty(\varphi,s,\chi)$ {\em does not} form a commuting family.

\subsection{Non-Hermitian transfer operators with broken spin reversal symmetry and quasilocal conservation laws}

HNTO (\ref{HNTO}) is neither a local operator, nor it is conserved in time as its time derivative  is a non-local object, namely using (\ref{HLT},\ref{Ltb}) we find:
\begin{eqnarray}
[H,W_n(\varphi,s)] &=& 2\sin\eta\bigl((\sigma^\z\sin\varphi\sin\eta s -\sigma^0  \cos\varphi\cos\eta s)\otimes W_{n-1}(\varphi,s) \nonumber \\
&-& W_{n-1}(\varphi,s) \otimes (\sigma^\z\sin\varphi\sin\eta s - \sigma^0\cos\varphi\cos\eta s)\bigr). \label{HNTOrel}
\end{eqnarray}
Yet, it can be used to generate a very interesting family of operators in terms of differentiation with respect to the spin representation parameter $s$ around the {\em scalar point} $s=0$
\begin{equation}
Z_n(\varphi) = \frac{1}{(\sin\varphi)^{n}}\partial_s W_n(\varphi,s)\vert_{s=0} - (\eta\cot\varphi) M^\z_n,
\label{Zdef}
\end{equation}
where $M^\z_n=\sum_{x=0}^{n-1}\one_{2^x}\otimes\sigma^\z\otimes\one_{2^{n-1-x}}$ is the conserved $\z-$component of magnetization.
The $s-$derivative can be implemented as MPA in terms of an additional `derivative anzilla' qubit ${\cal H}_{\rm c}=\CC^2$, 
\begin{equation}
Z_n(\varphi) = \bra{0}_{\rm a}\bra{0}_{\rm c}\mm{L}'_1(\varphi)\mm{L}'_2(\varphi)\cdots\mm{L}'_n(\varphi)
\ket{0}_{\rm a}\ket{1}_{\rm c} - (\eta\cot\varphi) M^\z_n,
\label{defZn}
\end{equation}
defining an {\em extended Lax operator} $\mm{L}'(\varphi) \in {\rm End}({\cal H}_{\rm a}\otimes {\cal H}_{\rm c}\otimes {\cal H}_{\rm p})$
\begin{equation}
\mm{{L}}'(\varphi) = \frac{1}{\sin\varphi}\pmatrix{
\mm{L}(\varphi,0) & \partial_s \mm{L}(\varphi,s)|_{s=0} \cr
0 & \mm{L}(\varphi,0)} =\mm{L}_0(\varphi) \one_{\rm c} + \mm{L}_1(\varphi)\sigma^+_{\rm c}, \label{defLp}
\end{equation}
where $\mm{L}_0(\varphi):= (\csc\varphi)\mm{L}(\varphi,0)$, $\mm{L}_1(\varphi):= (\csc\varphi)\partial_s\mm{L}(\varphi,s)|_{s=0}$.
We shall refer to the operator family $Z_n(\varphi)$ as the {\em modified} highest-weight non-Hermitian transfer operators (mHNTO). 
Note that the $Z$-operator (\ref{eq:Zs}) yielding the first-order perturbative expression of NESS is just $Z \equiv Z_n(\pi/2)^T$.
It can be shown  \cite{PI13} that in the massless regime (for a dense set of real $\eta$) $Z_n(\varphi)$ are quasilocal operators whose time-derivative is localized at the chain boundaries for a suitable domain $\varphi\in {\cal D}\subset\CC$.
Indeed, differentiating \eref{HNTOrel} w.r.t. $s$ at $s=0$ and using the definition (\ref{Zdef}) we immediately obtain a very insightful relation
\begin{eqnarray}
[H,Z_n(\varphi)] &=& 2\eta\sin\eta \left( \sigma^\z\otimes \one_{2^{n-1}} - \one_{2^{n-1}}\otimes\sigma^\z\right) \nonumber \\ 
&-& 2\sin\eta\cot\varphi\left( \sigma^0\otimes Z_{n-1}(\varphi)-Z_{n-1}(\varphi)\otimes\sigma^0\right). \label{HZc}
\end{eqnarray}
Writing the Lax operator components $\mm{L}^{'\alpha} \in {\rm End}({\cal H}_{\rm a}\otimes {\cal H}_{\rm c})$, $\mm{L}^\alpha \in {\rm End}({\cal H}_{\rm a})$, 
via $\mm{L}'(\varphi) = \sum_{\alpha\in{\cal J}} \mm{L}^{'\alpha}(\varphi) \otimes \sigma^\alpha$, 
$\mm{L}^{'\alpha}(\varphi) = \mm{L}^\alpha_0(\varphi) \one_{\rm c} + \mm{L}^\alpha_1(\varphi)\sigma^+_{\rm c} $
satisfying {\em boundary transition conditions}
\begin{eqnarray}
&&\bra{0}_{\rm a}\bra{0}_{\rm c} \mm{L}^{'0} = \bra{0}_{\rm a}\bra{0}_{\rm c}, \quad\quad \bra{0}_{\rm a}\bra{0}_{\rm c} \mm{L}^{'+} = 0, \nonumber \\
&&{\mm{L}}^{'0} \ket{0}_{\rm a}\ket{1}_{\rm c} =  \ket{0}_{\rm a}\ket{1}_{\rm c}, \quad\quad {\mm{L}}^{'-} \ket{0}_{\rm a}\ket{1}_{\rm c} =  0, \nonumber \\
&&{\mm{L}}^{'\z} \ket{0}_{\rm a}\ket{1}_{\rm c} = \eta\cot\varphi \ket{0}_{\rm a} \ket{0}_{\rm c}, \quad{\mm{L}}^{' \z,\pm} \ket{0}_{\rm a}\ket{0}_{\rm c} =  0, 
\label{bc}
\end{eqnarray}
one sees that mHNTOs allow for an expression in terms of open boundary translationally invariant sum of local operators
\begin{equation}
Z_n(\varphi) = \sum_{r=2}^n \sum_{x=0}^{n-r} \one_{2^{x}} \otimes q_r(\varphi) \otimes \one_{2^{n-r-x}},
\label{Zloc}
\label{qm}
\end{equation}
in terms of local $r-$point densities $q_r(\varphi) \in {\rm End}({\cal H}^{\otimes r}_{\rm p})$ with explicit MPA representation, which is obtained by careful inspection of the definitions (\ref{defZn},\ref{defLp})  
\begin{equation}
\!\!\!\!\!\!\!\!\!\!\!\!\!\!\!\!\!\!\!\!\!\!\!\!\!\!\!\!\!\!\!\!
q_r(\varphi) = \sum_{\alpha_2\ldots \alpha_{r-1}\in{\cal J}} \bra{0}\mm{L}^-_0(\varphi) \mm{L}^{\alpha_2}_0(\varphi)\cdots\mm{L}^{\alpha_{r-1}}_0(\varphi)\mm{L}^+_1(\varphi) \ket{0}  \sigma^- \otimes \sigma^{\alpha_2}\cdots \sigma^{\alpha_{r-1}}\otimes \sigma^+.\quad
\label{qr2}
\end{equation}
Using the local operator sum ansatz (\ref{Zloc}) one is able to rewrite the RHS of (\ref{HZc}) in a form of a sum of operators localized at the boundaries
\begin{eqnarray}
[H,Z_n(\varphi)] &=& 2\eta\sin\eta \left( \sigma^\z\otimes \one_{2^{n-1}} - \one_{2^{n-1}}\otimes\sigma^\z\right) \nonumber \\ 
&+&  2\sin\eta\cot\varphi \sum_{r=2}^n \left(q_r(\varphi)\otimes \one_{2^{n-r}} - \one_{2^{n-r}}\otimes q_r(\varphi) \right).
\end{eqnarray} 
We shall now demonstrate  (following \cite{PI13,P14c}) that there are important parameter regimes for which the operator sequence $\{ q_r(\varphi); r=2,3\ldots\}$ is quickly decreasing in a suitable operator norm, so the operator family (\ref{Zloc}) can be considered as {\em quasilocal} and {\em almost conserved}. 
Moreover, the operators $Z_n(\varphi)$ are represented as strictly {\em lower triangular matrices} with zero diagonal, $\bra{\ul{\nu}}Z_n(\varphi)\ket{\ul{\nu}'} = 0$ if ${\rm ord}(\ul{\nu}) < {\rm ord}(\ul{\nu}')$,
$\bra{\ul{\nu}}Z_n(\varphi)\ket{\ul{\nu}} \equiv 0$ for any $\varphi$. Hence they are {\em non-diaginalizable} for any $n>1$ as their spectrum contains only $0$.

However, in order to formulate our results precisely we first state two definitions of `essential' locality of translationally invariant operators:

\smallskip\noindent
 {\em Quasilocality:} An operator sequence $Z_n \in {\rm End}({\cal H}^{\otimes n}_{\rm p})$ which can be written as an open boundary translationally invariant sum of local operators $q_r$, like (\ref{Zloc}), for {\em any} $n$, is called {\em quasilocal} if there exist positive constants $\gamma,\xi > 0$, such that 
\begin{equation}
\| q_r \|_{\rm HS} < \gamma e^{-\xi r},
\end{equation} 
where, for any matrix $a$,
\begin{equation}
\| a\|^2_{\rm HS} := \frac{\tr(a^\dagger a)}{\tr\one}
\end{equation} is a normalized Hilbert-Schmidt norm which satisfies a nice extensivity property
\begin{equation}
\| a \|_{\rm HS} = \|\,a \otimes \one_d \|_{\rm HS},\quad \forall d,
\label{extensivity}
\end{equation}
as well as the normalized Cauchy-Schwartz inequality
\begin{equation}
\left |\frac{\tr (a b)}{\tr \one}\right| \le \| a\|_{\rm HS}\,\| b\|_{\rm HS}.
\label{CS}
\end{equation}
We remark that the Hilbert-Schmidt operator norm is the natural norm for high-temperature statistical mechanics as it is linked to an infinite temperature, tracial state $\omega_0(a) = \tr a/\tr\one$, specificallly $\|a\|^2_{\rm HS} = \omega_0(a^\dagger a)$.
Note also that it satisfies a useful inequality in relation to a $C^*$ operator norm $\| b \|^2 = \sup_{\omega} \omega(b^\dagger b)$, namely for any pair of bounded operators $a,b$ (say, elements of ${\rm End}({\cal H}^{\otimes n}_{\rm p})$), $\|a b\|_{\rm HS} \le \| a\|_{\rm HS} \| b \|$.

\smallskip\noindent
{\em Pseudolocality:} An operator sequence $Z_n \in {\rm End}({\cal H}^{\otimes n}_{\rm p})$ of the form (\ref{Zloc}) is called {\em pseudolocal} if there exists a positive constant $K > 0$, such that 
\begin{equation}
\| Z_n\|^2_{\rm HS} \le K n.
\end{equation}

Clearly, quasilocality implies pseudolocality as follows straightforwardly from the definitions. 
In order to demonstrate locality of mHNTO $Z_n(\varphi)$ for $XXZ$ one needs to study the sequence of Hilbert Schmitd norms $\| q_r(\varphi)\|_{\rm HS}$.
In fact,  for a set of commensurate anisotropies $\eta = \frac{\pi l}{m}$, $l,m\in\ZZ^+$, densely covering the easy-plane regime $|\Delta|<1$, one can explicitly study a general inner product \cite{P14c} 
\begin{equation}
\!\!\!\!\!\!\!\!\!\!\!\!\!\!\!\!\!\!\!\!\!\!\!\!\!\!\!\!\!\!
\kappa_r(\varphi,\varphi'):=\frac{1}{2^{r}}\tr\left( q^T_r(\varphi) q_r(\varphi')\right) = \left(\frac{2\eta s_1}{\sin\varphi \sin\varphi'}\right)^2 \bra{1} \mm{T}(\varphi,\varphi')^{r-2}\ket{1}, \quad {\rm r \ge 2}.
\label{rr}
\end{equation}
in terms of iterating the following finite, $m-$dimensional transfer matrix acting over a vector space ${\rm lsp}\{\ket{k},k=1,\ldots,m\}$,
\be
\!\!\!\!\!\!\!\!\!\!\!\!\!\!\!\!\!\!\!\!\!\!\!\!\!\!\!\!\!\!\!\!\!\mm{T}(\varphi,\varphi') \!\!=\!\!\!\sum_{k=1}^{m-1} (c^2_k +\cot\varphi\cot\varphi' s^2_k)\ket{k}\!\bra{k} + \sum_{k=1}^{m-2}\frac{|s_{k}s_{k+1}|}{2\sin\varphi\sin\varphi'}\left(\ket{k}\!\bra{k\!+\!1} + \ket{k\!+\!1}\!\bra{k}\right),  \label{TT}
\ee
where $c_k := \cos(\pi l k/m), s_k := \sin(\pi l k/m)$. Straightforward calculation \cite{PI13,P14c} shows that the leading eigenvalue $\tau=e^{-2\xi}$ of $\mm{T}(\varphi,\varphi')$ has modulus smaller than $1$, i.e. $\xi>0$, and hence $Z_n(\varphi)$ being quasilocal, exactly if $\varphi,\varphi'$ belong to the vertical strip ${\cal D}_m = \{\varphi; |{\rm Re}\varphi-\frac{\pi}{2}| < \frac{\pi}{2m}\}$.
Moreover, on ${\cal D}_m$ one can easily compute the full extensive normalized Hilbert-Schmidt inner-product 
\begin{equation}
(A,B) = \frac{\tr A^\dagger B}{\tr\one},
\label{eq:HSip}
\end{equation} which turns the space of observables ${\rm End}({\cal H}_{\rm p}^{\otimes n})$ into a Hilbert space, 
between all mHNTO.
Namely, for any pair $\varphi,\varphi'\in{\cal D}_m$
\begin{eqnarray}
\left(Z_n(\bar{\varphi}),Z_n(\varphi')\right)&=&\sum_{r=2}^n (n-r+1)\kappa_r(\varphi,\varphi')\nonumber \\
&=& n\sum_{r=2}^\infty \kappa_r(\varphi,\varphi') - \sum_{r=2}^\infty (r-1) \kappa_r(\varphi,\varphi') + {\cal O}(n e^{-2\xi n}) \nonumber \\
&=& n K(\varphi,\varphi') + {\cal O}(n^0), \label{innerZ}, \\
\left(Z^T_n(\bar{\varphi}),Z_n(\varphi')\right)&=& 0, \nonumber
\end{eqnarray}
where
\begin{eqnarray}
K(\varphi,\varphi') &=& \sum_{r=2}^\infty \kappa_r(\varphi,\varphi') = \left(\frac{2\eta s_1}{\sin\varphi \sin\varphi'}\right)^2 \bra{1}(\one-\mm{T}(\varphi,\varphi'))^{-1}\ket{1} \nonumber \\
&=& -\frac{8\eta^2}{\sin\varphi \sin\varphi'} \frac{\sin((m-1)(\varphi+\varphi'))}{\sin(m(\varphi+\varphi'))}.
\label{gm}
\end{eqnarray}
With a simple technical trick \cite{P14c,PPSA14} one can, again for commensurate anisotropies $\eta = \pi l/m$ densely covering the gapless regime $|\Delta|<1$, define a set of quasilocal conserved operators which are {\em exactly conserved} for the $XXZ$ Hamiltonian with periodic (or even twisted \cite{P14c}) boundary conditions,
\be
H_{\rm pbc} = H + 2 \sigma^+ \otimes \one_{2^{n-2}} \otimes \sigma^- +2 \sigma^-\otimes \one_{2^{n-2}} \otimes \sigma^+ + \Delta \sigma^\z \otimes \one_{2^{n-2}} \otimes \sigma^\z,\quad
\ee
namely 
\begin{eqnarray}
Y_n(\varphi) &=&  \frac{1}{(\sin\varphi)^{n}}\partial_s V_n(\varphi,s)\vert_{s=0}-(\eta\cot\varphi) M^\z_n \nonumber \\
&=& \tr_{\rm a} \bra{0}_{\rm c}\mm{L}'_1(\varphi) \mm{L}'_2(\varphi)\cdots \mm{L}'_n(\varphi) \ket{1}_{\rm c}-(\eta\cot\varphi) M^\z_n,
\label{Ydef1}
\end{eqnarray}
where
\begin{equation}
V_n(\varphi,s) = \tr_{\rm a}\left\{\mm{L}_1(\varphi,s)\mm{L}_2(\varphi,s) \cdots \mm{L}_n(\varphi,s)  \right\},
\label{Vdef}
\end{equation}
with the auxiliary space being naturally truncated to ${\cal H}'_\aa = {\rm lsp}\{ \ket{k}, k = 0,1\ldots,m-1\}$ since $\bra{m-1}\mm{L}^{-}(\varphi,s)\equiv 0$.
The Sutherland condition (\ref{eq:LOD}) then immediately implies $[H_{\rm pbc},V_n(\varphi,s)]\equiv 0$, and consequently [via (\ref{Ydef1})], exact conservation 
\begin{equation} 
[H_{\rm pbc},Y_n(\varphi)]=0,\quad\forall \varphi \in \CC,
\end{equation}
while YBE (\ref{RLLa}) implies
\begin{equation}
[Y_n(\varphi),Y_n(\varphi')] = 0,\quad\forall \varphi,\varphi'\in\CC.
\end{equation}
Furthermore, one can easily show \cite{P14c} that the difference between mHNTO $Z_n(\varphi)$ and the so-called {\em modified periodic non-Hermitian transfer operator} (mPNTO) $Y_n(\varphi)$ is exponentially small away from the boundaries.
More precisely, mPNTO is again quasilocal in the sense that it can be written as a translationally invariant sum of local operators (with the same densities as $Z_n(\varphi)$)
\begin{equation}
Y_n(\varphi) = \sum_{r=2}^n \sum_{x=0}^{n-1} \hat{\cal S}^x ( \one_{2^{n-r}}\otimes q_r(\varphi)) + y_n(\varphi)
\label{lem}
\end{equation}
where $\hat{\cal S} : {\rm End}({\cal H}^{\otimes n}_{\rm p}) \to {\rm End}({\cal H}^{\otimes n}_{\rm p})$ is a left-shift rotation map which is completely specified by the action on the Pauli basis
\begin{equation}
\hat{\cal S}(\sigma^{\alpha_0}\otimes \sigma^{\alpha_1} \otimes\cdots\sigma^{\alpha_{n-2}}\otimes\sigma^{\alpha_{n-1}}) = \sigma^{\alpha_{1}}\otimes \sigma^{\alpha_2}\otimes\cdots\sigma^{\alpha_{n-1}}\otimes\sigma^{\alpha_{0}},
\end{equation}
and the remainder is exponentially small in Hilbert-Schmidt norm $\| y_n(\varphi)\|_{\rm HS} = {\cal O}(e^{-\xi n})$. Similarly to HNTO (\ref{innerZ}), the family of mPNTO has asymptotically the same kernel of inner products
\begin{equation}
\!\!\!\!\!\!\!\!\!\!\!\!
\left(Y_n(\bar{\varphi}),Y_n(\varphi')\right)= n K(\varphi,\varphi') + {\cal O}(e^{-2\xi n}),\quad
\left(Y^T_n(\bar{\varphi}),Y_n(\varphi')\right) = {\cal O}(\e^{-2\xi n}).
\end{equation}

The standard algebraic Bethe ansatz (ABA) machinery \cite{KBI93,F94,GM95} allows one to derive a sequence of strictly local translationally invariant and exactly conserved operators $Q_r$, $r=2,3\ldots$,
 $[Q_r,Q_{r'}]=0$,  
\be
Q_r = \partial^{r-1}_\varphi \log V_n(\varphi,\half)|_{\varphi=\frac{\eta}{2}} =  
\sum_{x=0}^{n-1} \hat{\cal S}^x ( \one_{2^{n-r}}\otimes q^{(r)}),
\label{LCL}
\ee
with the first term in the series being proportional to the hamiltonian  $Q_2\propto H_{\rm pbc}$,
and where $q^{(r)} \in {\rm End}({\cal H}_{\rm p}^{\otimes r})$ are the corresponding local densities.
Importantly, the ABA transfer operator $V_n(\varphi,\half)$ is spin reversal (\ref{eq:parity}) invariant, and hence are all local conserved operators
\begin{equation}
P V_n(\varphi,\half) P^{-1} = V_n(\varphi,\half),\quad P Q_r P^{-1} = Q_r.
\end{equation}
On the other hand, the non-Hermitian transfer operators, and the corresponding quasilocal conserved operators, satisfy a more specific PT-like symmetry (analogous to the one discussed in \cite{P12b}), following from (\ref{eq:spinrevL})
\begin{equation}
P Z_n(\varphi) P^{-1} = Z^T_n(\pi-\varphi),\quad
P Y_n(\varphi) P^{-1} = Y^T_n(\pi-\varphi).
\end{equation}
This means that the modified (and quasilocal for $\eta=\pi l/m$) transfer operators can be decomposed into even and odd w.r.t. spin reversal
\begin{equation}
Z^\pm_n(\varphi) = Z_n(\varphi) \pm P Z_n(\varphi) P^{-1},\quad Y^\pm_n(\varphi) = Y_n(\varphi) \pm P Y_n(\varphi) P^{-1},
\end{equation}
such that
\begin{eqnarray}
&&[H_{\rm pbc},Y^\alpha_n(\varphi)]=0,\quad P Y^\alpha_n(\varphi) P^{-1} = \alpha Y^\alpha_n(\varphi),\\
&&(Y^{\alpha}_n(\bar{\varphi}),Y^{\alpha'}_n(\varphi)) = n \delta_{\alpha,\alpha'} K(\varphi,\varphi') + {\cal O}(e^{-2\xi n}),
\end{eqnarray}
$\alpha,\alpha'\in\{\pm\}$, and similar relations for $Z^\alpha_n(\varphi)$ with open boundaries.

\subsection{Lower bounds on high temperature ballistic transport coefficients}

Existence of quasilocal conserved operators is extremely interesting for deriving bounds on linear-response transport coefficients.
For example, considering the {\em extensive} current 
\be
J = \sum_{x=0}^{n-1} \hat{\cal S}^x(\one_{2^{n-2}}\otimes j)
\ee
(with periodic boundary conditions), the famous Green-Kubo formula expresses the corresponding conductivity $\sigma'(\omega)$ (in our case spin-conductivity if $j$ is a local spin current (\ref{eq:j}), but
for a general discussion $J$ can be any current linked to an appropriate conservation law) in terms of the current auto-correlation function. In particular,
\be
\sigma'(\omega) = \lim_{t\to\infty}\lim_{n\to\infty}\frac{\beta}{n} \int_0^t \dd t' e^{\ii\omega t'} (J(t'),J(0))_\beta
\ee
where $J(t) := e^{\ii H_{\rm pbc} t} J e^{-\ii H_{\rm pbc}t}$ is the Heisenberg dynamics of the corresponding current operator, and
\be
(A,B)_\beta= \frac{1}{\beta} \int_0^\beta \dd \lambda \frac{\tr\left( A^\dagger e^{-\lambda H_{\rm pbc} }B e^{-(\beta-\lambda) H_{\rm pbc}}\right)}{\tr\,e^{-\beta H_{\rm pbc}}}
\ee
is the Kubo-Mori inner product, which reduces to Hilber-Schmidt inner product (\ref{eq:HSip}) in the limit of infinite temperature, $\lim_{\beta\to 0}(A,B)_{\beta} \equiv (A,B)$.
When d.c. ($\omega=0$) conductivity diverges, one defines the (spin) Drude weight $D$
\be
\sigma'(\omega) = 2\pi D \delta(\omega) + \sigma_{\rm reg}(\omega)
\ee
which can be again expressed with a Green-Kubo-like formula
\be
D = \lim_{t\to\infty}\lim_{n\to\infty}\frac{\beta}{2t n} \int_0^t \dd t' (J(t'), J(0))_\beta.
\ee
Drude weight can be considered as a ballistic transport coefficient, and $D>0$ signals an ideal, ballistic spin transport in an extended system at finite temperatures \cite{ZP03,HHB07}.
At high temperature $\beta\to 0$, the leading order Drude weight can be expressed as
\be
D = \beta D_\infty + {\cal O}(\beta^2),\quad D_\infty = \lim_{t\to\infty}\lim_{n\to\infty}\frac{1}{2t n} \int_0^t \dd t' (J(t'),J).
\ee
For integrable quantum systems, having a (countable) set of local extensive conserved operators $\{ Q_r \}$ (\ref{LCL}), 
Zotos, Naef and Prelov\v sek \cite{ZNP} suggested to use Mazur bound \cite{M69}, to rigorously estimate the
high-temperature Drude weight from below,  
\be
D_\infty \ge \lim_{n\to\infty} \frac{1}{2n} \sum_{r,r'} (J,Q_r) (K^{-1})_{r,r'} (Q_{r'},J)
\label{eq:MBC}
\ee
where $K_{r,r'} := (Q_r,Q_{r'})$ is a positive-definite matrix of inner-products of independent conserved operators. Usually, one picks such linear combinations of $Q_r$ which are mutually orthogonal, say, using a Gram-Schmidt procedure, so the above formula simplifies
with $(K^{-1})_{r,r'} = \delta_{r,r'} \|Q_r\|^{-2}_{\rm HS}$.

However, the situation becomes interesting and quite intricate for systems with $\ZZ^2$ symmetries, like spin-reversal, parity-hole, etc, such that the corresponding spin or charge current is {\em odd} under the transformation 
\be
P J P^{-1} = -J,
\ee 
and the symmetry in the corresponding equilibrium state is un-broken, $\tr(X e^{-\beta H}) = \tr (P X P^{-1} e^{-\beta H}), \forall X$, i.e. in the absence of external magnetic fields, chemical potentials, etc, such that
$P H P^{-1} = H$. In the case of $XXZ$ model this immediately implies that the RHS of Mazur bound (\ref{eq:MBC}) has to vanish, since $(J,Q_m) = - (P J P^{-1},P Q_m P^{-1})  = -(J,Q_m) = 0$, and one has to rely on effective theories
and approximations \cite{S12,SPA11}.
However, the situation drastically changes due the presence of quasilocal conserved operators with odd spin reversal symmetry $\{ Y^-_n(\varphi) \}$.
Even replacing a single operator from this set $Y^{-}_n(\pi/2)$ for commensurate anisotropy $\eta=\pi l/m$ (or, equivalently, $Z_n(\varphi/2)$ for open boundaries \cite{IP13}) into RHS of (\ref{eq:MBC}) one obtains
a non-vanishing lower bound \cite{P11a,IP13,PPSA14}
\begin{equation}
D_\infty \ge D_Z = \sin^2(\pi l/m) \frac{m}{m-1}.
\end{equation}
However, one can do much better than that! Considering the full continuous family of non-Hermitian quasilocal conserved operators $\{Y_n(\varphi),\varphi\in{\cal D}_m \}$ replacing a countable family $\{Q_m,m=1,2\ldots\}$ 
in (\ref{eq:MBC}) on can write the Mazur bound
\begin{equation}
D_\infty \ge D_K \ge D_Z,\quad D_K = \frac{1}{2}{\rm Re} \int_{{\cal D}_m}\!\!\!\dd^2\varphi\,a(\varphi) f(\varphi),
\end{equation}
where $a(\varphi) := (j_{x,x+1},Y^-_n(\varphi))$, in terms of a solution $f(\varphi)$ of a complex Fredholm equation of the first kind, involving the kernel (\ref{gm}) (see Ref.~\cite{P14c} for details)
\begin{equation}
\frac{1}{2}\int_{{\cal D}_m}\!\!\!K(\varphi,\varphi') f(\varphi')\dd^2\varphi' = \overline{a(\bar{\varphi})}.
\end{equation}
In our case $a(\varphi) = \ii/4$, and the Fredholm equation has a simple explicit solution $f(\varphi) = -\frac{\ii}{\pi} m s_1^2 |\sin\varphi|^{-4}$, yielding an explicit expression for the {\em optimised} Mazur bound
\be
D_K = \frac{\sin^2(\pi l/m)}{\sin^2(\pi/m)} \left(1 - \frac{m}{2\pi}\sin\left(\frac{2\pi}{m}\right)\right).
\ee
It could be tempting to speculate that this bound is in fact saturating the high-temperature spin-Drude weight (see Fig.~\ref{fig:DK}), in particular since it agrees very well with the state-of-the-art time dependent DMRG simulations \cite{Karrasch}.
\begin{figure}
\centering
\includegraphics[scale=0.75]{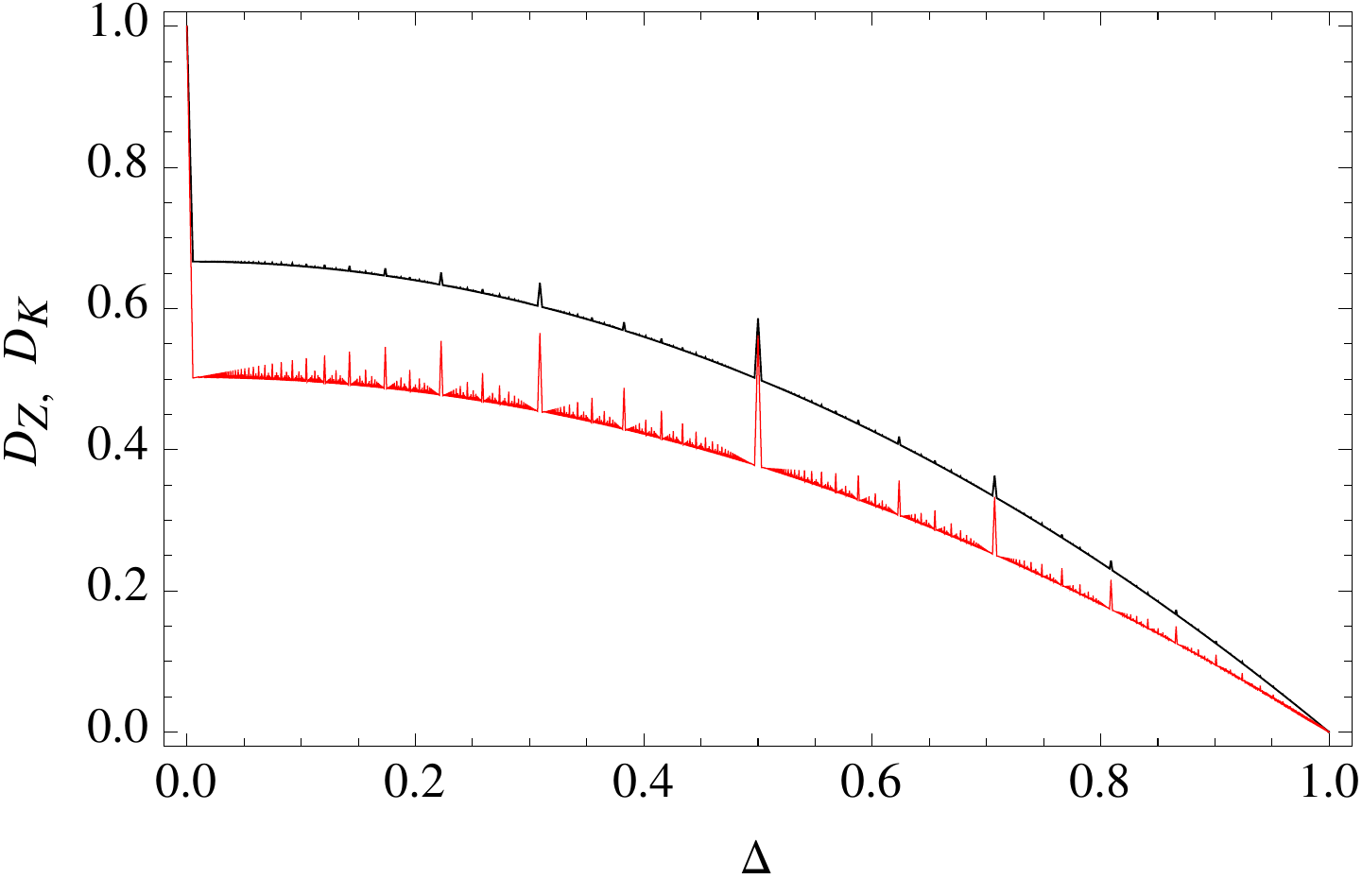}
\caption{(from \cite{PI13})
Optimized Mazur lower bound on high-temperature spin Drude weight $C_K$ (black) \cite{PI13} versus the less optimal bound $C_1$ of Ref.~\cite{P11a} (red)
which is based on
a single quasilocal almost-conserved operator $Z_n(\varphi=\pi/2)$ (or $Y_n(\pi/2)$), for $XXZ$ model as a function of anisotropy parameter $\Delta$.
}
\label{fig:DK}
\end{figure}
Note, however, that both, the non-optimal and optimised lower bounds $D_Z,D_K$ are no-where continuous (`fractal') functions of the anisotropy parameter $\Delta = \cos\eta$. This suggests an interesting option that 
thermodynamic properties such as transport coefficients of strongly interacting systems might be no-where continuous functions for a finite range of parameters. Similar analysis could be attempted for finite inverse temperatures $\beta > 0$, but then the evaluation of the $K(\varphi,\varphi')$ and the function $a(\varphi)$ should be evaluated either numerically or approximately (perturbatively).

\subsection{Models with non-deformed symmetries and lower bounds on high temperature diffusion constants}

\label{subsect:undeformed}

Analysis of the previous section showed how new quasilocal operators can be distilled from NESS density operator in boundary dissipatively driven model with trigonometric $R$-matrix (corresponding to $U_q(\mathfrak{sl}_2)$ quantum symmetry with uni-modular $q$). We expect that similar results could be obtained for other integrable models sharing the same $R$-matrix, say Sine-Gordon field theory in $1+1$ dimensions, or its discretisation, the quantum Hirota equation (see e.g. \cite{F94}).
However, our conserved operators $Y_n(\varphi)$ or $Z_n(\varphi)$ are no-longer quasilocal when $|q| > 1$, when the amplitudes of the MPA diverge super-exponentially, or even in the un-deformed $\mathfrak{sl}_2$ symmetric limit $q\to1$.

In such a case, for the $XXX$ spin 1/2 model, the central quasilocal operator $Z^T(\varphi=\pi/2)$ goes to 
\be 
Z = -\ii\partial_\varepsilon\Omega_n|_{\varepsilon=0}= \sum_{x=1}^{n-1}\sum_{y=x+1}^{n} \sigma^+_x \sigma^-_y,
\label{eq:Yang1}
\ee 
so it becomes quadratically extensive in the sense that 
$\| Z\|^2_{\rm HS} \simeq \half q n^2$ with $q$ being some real constant, while still satisfying almost conservation, or conservation law property $[H,Z] = -\sigma^\z_1 + \sigma^\z_n$. 
One finds very similar behaviour for other models with Lie (underformed) symmetries presented in this overview, namely Fermi-Hubbard and spin-1 Lai-Sutherland models.

For example, for symmetrically boundary driven Hubbard model, one can define an analogous operator as 
 \begin{eqnarray}
\!\!\!\!\!\!\!\!\!\!\!\!\!\!\!\!\!\!
&& Z = -\ii \partial_\varepsilon \Omega_n(\lambda=0,w = \ihalf\varepsilon)  = \label{eq:Yang2}\\
\!\!\!\!\!\!\!\!\!\!\!\!\!\!\!\!\!\!
&& \sum_{x=1}^{n-1}(\sigma^+_x\sigma^-_{x+1}+\tau^+_x\tau^-_{x+1})  -  2u\sum_{x,y}^{x<y} (-1)^{x-y}\sigma^+_x\left(\prod_{z=x+1}^{y-1} \sigma^\z_z \right) \sigma^-_y \tau^+_x \left(\prod_{z=x+1}^{y-1}\tau^\z_z\right)\tau^-_y, \nonumber
\end{eqnarray}
which is again quadratically superextensive $\| Z\|_{\rm HS}^2 \simeq \half q n^2$ and satisfies almost conservation property 
\be
[H,Z] = -\sigma^\z_1 - \tau^\z_1 + \sigma^\z_n + \tau^\z_n.
\ee 
Note that both operators, (\ref{eq:Yang1}) and (\ref{eq:Yang2}), can be viewed as level-1 generators of the Yangian symmetry of the respective models, truncated to a finite size $n$ \cite{bernard,uglov}.
Similarly as in the case of $XXZ$ model, the entire two-parameter set of operators $\Omega_n(\lambda,w)$ (\ref{eq:OmHub}) can be considered as an HNTO, namely one finds by explicit computation that
\be
[\Omega_n(\lambda,w),\Omega_n(\lambda',w')] = 0,\quad \forall \lambda,\lambda',w,w'\in\CC.
\ee
However, proving the existence or even explicitly constructing the corresponding intertwiner, the (infinite-dimensional) exterior $R$-matrix, remains a problem for the future.

In the absence of local and quasilocal operators which are odd under spin reversal, or particle-hole transformation
\footnote{Reader should remember that the ABA transfer operator and all the derived local conserved operators are {\em even} under $P$.}, 
$P$, one may attempt use the almost-conserved Hermitian operator $Q=\ii (Z-Z^\dagger)$, $P Q P^{-1} = -Q$ to estimate d.c. spin transport coefficients. However, due to super-extensivity, the contribution of $Q$ to Drude weight is vanishing in TL. Nevertheless, using 
a careful estimation of the effects of time evolution of the derivative $[H,Q]$ (which is by assumption of almost-conservation localised near the boundaries) by means of the Lieb-Robinson bound \cite{LR72}, one can find  a rigorous lower bound on the Green-Kubo expression for the high-temperature diffusion constant \cite{P14b}
\begin{equation}
D_{\rm diff} = \lim_{t\to\infty}\lim_{n\to\infty}\frac{1}{n}\int_{-t}^t \dd t' (J(t'),J) = \lim_{\beta,\omega\to 0} \frac{\sigma'(\omega)}{\beta}.
\end{equation}
For the locally interacting hamiltonian $H = \sum_x h_{x,x+1}$, local current observable $J =\sum_x j_{x,x+1}$ and general quadratically extensive almost conserved operator $Q$,
with $q = \lim_{n\to\infty} \frac{1}{n^2} \|Q\|^2_{\rm HS} > 0$, one finds a general theorem \cite{P14b}, stating that
\begin{equation}
D_{\rm diff} \ge \frac{ |(j_{x,x+1},Q)|^2 }{8 v q},
\end{equation}
where $v$ is the Lieb-Robinson velocity \cite{hastings,bruno}, which can be estimated in terms of the norm of local Hamiltonian density, $v \le 6 \| h \|$, hence $D_{\rm diff}\ge |(j_{x,x+1},Q)|^2/(48 \| h\| q)$.
Applying to $XXX$ and Fermi-Hubbard models this result implies strict bounds on the respective spin and charge diffusion constants, $D^{XXX}_{\rm diff} \ge 1/6$, $D^{\rm Hubbard}_{\rm diff} \ge 1/(3 u^2)$, which, from a rigorous point of view, prove that the high-temperature spin/charge transport in these models cannot be sub-diffusive or even insulating.
Whereas in the Heisenberg model this bound may be superfluous, as DMRG numerical simulations suggest that the high-temperature spin transport in the isotropic point seems to be anomalous (super-diffusive but sub-ballistic, and hence $D_{\rm diff}=\infty$) \cite{Z11b}, 
in the Hubbard model the bound seem to be less trivial as the numerics suggests diffusive transport \cite{PZ12}.

\section{Discussion}
\label{disc}

\subsection{Open problems}

We shall close the presentation of this growing subject with a list of, to author's taste, most urgent open problems.

\begin{itemize}

\item We see currently no analytical technique to compute the Liouvillian spectrum, its gap and decay modes, i.e. to solve the full Liouvillian eigenproblem $\LL v_j = \lambda_j v_j$ in the models where NESS (fixed point) is an exactly solvable MPA.
Even in the simplest non-trivial (strongly-interacting) case of integrable NESS, e.g. in the boundary driven $XXX$ spin $1/2$ chain, we currently do not understand how to build higher decay modes $v_j$, with decay rates $\lambda_j, \Re \lambda_j < 0$ (see e.g. Refs.~\cite{P08,P10,medvedyeva} for such results on non-interacting systems). It is not even clear at present if the integrability of NESS implies that the problem of diagonalizing the full Liouvillian $\LL$ needs to be integrable.
One should  perhaps note that non-trivial statements can be made about the structure of decay mode spectrum based on rather general properties of the Liouvillian such as an analogue of the topical PT-symmetry \cite{P12b}.

\item
Alternatively, one may try to construct exact solutions for relaxation dynamics $\rho(t)$ directly in the time domain for specific non-trivial initial states $\rho(0)$, in analogy with SEP \cite{S01}. 
For example, one may devise a {\em quench protocol}, where $\rho(0)$ is an exact MPA NESS of an integrable model. Then, at $t=0$, one suddenly changes the parameter of the Hamiltonian or of the dissipator, such that the NESS of the after-quench problem remains integrable.
It is perhaps reasonable to expect that then the full dynamics $\rho(t)$ remains integrable, i.e. exactly solvable, as well.

\item So far, one is able to write exact MPA solutions of NESS only for integrable quantum chains with very specific dissipative boundary conditions, which can be phrased as a pure source and a pure sink.
More general boundary conditions can be treated only perturbively in the system-bath coupling constant.
This situation is quite different than in SEP \cite{BE07,S01} where one can typically exactly treat the most general local boundary conditions.
One perhaps needs to extend the Skylanin's concept of reflection algebra \cite{SkRA} to quantum Liouville space formalism.

\item The NESS density operator $\rho_\infty$ with exact MPA structure could be compared with its equilibrium integrable counterpart, which is the Gibbs operator $e^{-\beta H}$ where $H$ is a Hamiltonian of an integrable quantum chain.
$\log \rho_\infty$ can thus be considered as a kind of nonequilibrium integrable Hamiltonian and the eigenvalues of $\rho_\infty$ can be considered as probabilities.
More generally, $\rho_\infty$ determines the nonequilibrium thermodynamic state of the system and one may be interested in computing its observable properties.
Thus it would be desirable to have a Bethe-ansatz for diagonalization of $\rho_\infty$. In Ref.~\cite{PIP} the first step of such protocol has been outlined, i.e. the single quasi-particle spectrum of $\rho_\infty$ for boundary driven $XXX$ spin $1/2$ chain has been calculated, but problems with higher quasi-particle excitations have been identified.

\item The NESS density operator $\rho_\infty$ only entails average steady-state properties of the system. In order to access fluctuation properties, such as e.g. cumulants of the current, one needs to go beyond the master equation and consider the so-called {\em full-counting-statistics} \cite{EHM} or analogous {\em large-deviation-theory} formalism \cite{T}. It has been shown in Ref.~\cite{Mallick}, that the $k$-th cumulant problem is exactly solvable by MPA for classical ASEP for any $k$. Also, for a boundary driven quantum $XXZ$ spin $1/2$ chain we were able to compute all cumulants of the current perturbatively in the system-bath coupling strength. However, an analog quantum problem to Ref.~\cite{Mallick}, i.e. computing exact current cummulants in boundary driven $XXZ$ spin-1/2 chain, remains open.
A partial, perturbative, result in this direction has been achieved in Ref.~\cite{BP14}, where all cumulants of the current have been calculated in the lowest two orders of system-bath coupling.

\item 
All exact NESS solutions presented in this topical review refer to one-dimensional chains with ultra local dissipation acting only on the first and the last site of the chain.
It would be tempting and physically desirable to extend our nonequilibrium integrability techniques to quantum field theories with incoherent particle sources and sinks at or near the boundaries. Prime candidates are Sine-Gordon and Lieb-Liniger models, highly relevant for low-energy condensed matter or cold atom physics. The main difficulty is entailed in regularisation of the dissipator for a field theory, namely it should correspond to a source/sink of a particle with a well defined single-particle wave-function \cite{spohn} which cannot be located strictly at a point since this would correspond to infinite energy.

\item 
Motivated by numerical results of Ref.~\cite{PZ13}, showing a phase transition from ballistic to diffusive spin-transport in a classical lattice Landau-Lifshitz spin chain (see e.g. \cite{FT}) which can be considered as an integrable classical limit of the
$XXZ$ chain, one may be tempted to formulate a consistent
 classical limit for integrable nonequilibrium boundary driven models (i.e. `classical exterior integrability'). For example, one may write a boundary driven lattice Landau-Lifshitz model with Langevin noise processes attached to the end sites and attempt to construct an exact MPA for the steady state (classical NESS).
 In full analogy with the story on $XXZ$ spin $1/2$ chain, one again has a spin reversal symmetry, with respect to which all classical local conserved quantities are even, but one expects to derive new quasilocal conserved quantities with broken spin-reversal symmetry which could explain the
 ballistic classical spin transport.
 
 \item
We note that an alternative path to integrability in open nonequilibrium systems in terms of scattering state formalism, which is specially suited for {\em quantum impurity problems}, has been outlined in Ref.~\cite{MA06}.
It is not clear whether and how a link to the boundary dissipation approach discussed in this review can be made, and whether the latter could be implemented to treat integrable impurity problems.
\end{itemize}

\subsection{Conclusion}

This topical review presented a state-of-the-art (or better to say, a snap-shot in developing the theory of) exact MPA solutions of steady states of dissipatively boundary driven quantum integrable chains. An attempt to make a coherent presentation covering a variety of different integrable models under the same footing has been made. The key constructive (algebraic) methods have been identified and related to general methods of quantum integrability, such as the Lax structure and Yang-Baxter equation. However, important distinction to integrable equilibrium problems should be underlined, namely in dissipatively driven quantum chains one should consider non-unitary irreducible representations of the quantum (deformed), or Lie symmetries of the model, where the representation parameter of these algebraic structures is connected to the noise strength at the chain ends.

Apart from reviewing previously published material in a coherent and self-contained manner, this article contains also several original scientific results, the most notable being: (i) MPA for NESS in asymmetrically driven $XXZ$ chain and with arbitrary (asymmetric) transverse fields at the boundary, and (ii) exact asymptotic computation of the nonequilibrium partition function for the $XXX$ model (but also extending to other models with un-deformed Lie symmetries) which is the basis for computing nonequilibrium  thermodynamics and observables.

\section*{Acknowledgements}

I wish to thank Enej Ilievski and Vladislav Popkov, for fruitful collaboration on several essential parts of the work reviewed in this article. I also thank Berislav Bu\v ca, Hosho Katsura, Ugo Marzolino, Marko Medenjak, Keiji Saito, Gunter Sch\" utz, Frank Verstraete, and Marko \v Znidari\v c for enlightening discussions and/or collaborations on closely related topics. Finally, I acknowledge financial support by the grants
P1-0044, J1-5439 and N1-0025 of Slovenian Research Agency.

\end{document}